\documentclass[12pt]{article}
\usepackage[utf8]{inputenc}
\usepackage[margin=1in]{geometry}
\usepackage{hyperref}
\usepackage{amsmath, amssymb, amsthm}
\usepackage{graphicx}
\usepackage{enumerate}
\usepackage{natbib}
\usepackage{url} 
\usepackage[titlenumbered,ruled]{algorithm2e}
\usepackage{bm,bbm}
\usepackage{float}
\usepackage{mathabx}
\usepackage{lipsum}

\newcommand{\blind}{1}

\def\Sig{\Sigma}
\def\sig{\sigma}
\def\R{\mathbb{R}}
\def\P{\mathbb{P}}
\def\E{\mathbb{E}}
\def\eps{\epsilon}

\def\lam{\lambda}

\def\mA{\mathcal{A}}
\def\gam{\gamma}

\DeclareMathOperator*{\argmin}{arg\,min}

\newtheorem{theorem}{Theorem}[section]
\newtheorem{lemma}{Lemma}[section]

\newtheorem{remark}{Remark}[section]
\newtheorem{condition}{Condition}[section]

\newcommand{\ignore}[1]{}
\begin{document}

\def\spacingset#1{\renewcommand{\baselinestretch}%
{#1}\small\normalsize} \spacingset{1}


\if1\blind
{
  \title{\bf Transfer Learning in Large-scale Gaussian Graphical Models with False Discovery Rate Control}
  \author{Sai Li\footnote{Department of Biostatistics, University of Pennsylvania, Philadelphia, PA 19104 (E-mail: \emph{Sai.Li@pennmedicine.upenn.edu})}, \ T. Tony Cai \footnote{Department of Statistics, the Wharton School, University of Pennsylvania, Philadelphia, PA 19104 (E-mail:\emph{tcai@wharton.upenn.edu})} \ and \ Hongzhe Li \footnote{Department of Biostatistics, University of Pennsylvania, Philadelphia, PA 19104 (E-mail: \emph{hongzhe@upenn.edu}).}}
    \date{}
  \maketitle
\fi

\if0\blind
{
  \bigskip
  \bigskip
  \bigskip
  \begin{center}
    {\LARGE\bf Transfer Learning in Large-scale Gaussian Graphical Models with False Discovery Rate Control}
\end{center}
  \medskip
} \fi

\bigskip
\begin{abstract}
Transfer learning for high-dimensional Gaussian graphical models (GGMs) is studied with the goal of estimating the target GGM by utilizing the data from similar and related  auxiliary studies.  The similarity between the target graph and each auxiliary graph is characterized by the sparsity of a divergence matrix. An estimation algorithm, Trans-CLIME, is proposed and shown to attain a faster convergence rate than the minimax rate in the single study  setting. Furthermore, a debiased Trans-CLIME estimator is introduced and shown to be element-wise asymptotically normal. It is used to construct a multiple testing procedure for edge detection with false discovery rate control.  The proposed estimation and multiple testing procedures demonstrate superior numerical performance in simulations and are applied  to infer the gene networks in a target brain tissue by leveraging the gene expressions from multiple other brain tissues. A significant decrease in prediction errors and a significant increase in power for link detection are observed.
\end{abstract}

\noindent%
{\it Keywords:} Inverse covariance matrix, meta learning, debiased estimator, multiple testing.
\vfill

\newpage
\spacingset{1.5} 
\section{Introduction}
Gaussian graphical models (GGMs), which represent the dependence structure among a set of random variables, have been widely used to model the conditional dependence relationships in many applications, including gene regulatory networks and brain connectivity maps \citep{DM07, VGP10, Zhao14, GZS19}. In the classical  setting with data from a single study,  the estimation of high-dimensional GGMs has been well studied in a series of papers, including  penalized likelihood methods  \citep{YL07, Fan09, Glasso, Roth08} and convex optimization based methods  \citep{Cai11, ACLIME, Tiger}.  The minimax optimal rates are studied in \citet{ACLIME} and a review can be found in \citet{Cai17}.
\citet{Liu13} considers  the inference in GGMs based on a node-wise regression approach and \citet{Zhao15} studies the estimation optimality and inference for individual entries.

Methods for estimating a single GGM have also been extended to simultaneously estimating multiple graphs when data from  multiple studies are available.  For example, \citet{JointGGM, JointGLasso, Cai2016joint} consider jointly estimating multiple GGMs with some penalties for inducing common structures among different graphs.  This problem falls in the category of multi-task learning  \citep{Lounici09, Wain12}, whose goal is to jointly estimate several related graphs. 

Due to high dimensionality and relatively small sample sizes in many modern applications, estimation of GGMs based on data from a single study often has large uncertainty and low power in detecting links in the corresponding graphs.
However,  the blessing is that samples from some different but related studies can be abundant.  Particularly, 
 for a given target study, there might be other similar studies where we expect some  similar dependence structures among the same set of variables. One example  is to infer the gene regulatory networks among a set of genes for a given issue.  Although gene regulatory networks are expected to vary from tissue to tissue,  certain shared regulatory structures are expected and have indeed been observed \citep{Pierson15, fagny17}. This paper introduces a transfer learning approach to improve the estimation and inference accuracy for the gene regulatory network in one target tissue by incorporating the data in other tissues. 
 
 Transfer learning  techniques have been developed in a range of applications, including pattern recognition, natural language processing, and drug discovery \citep{PY09, Turki17, Bastani18}. 
 It  has been studied in different settings with various similarity measures, but only a few of them offer statistical guarantees.  \citet{CW19} investigates nonparametric classification in transfer learning and proposes minimax and adaptive classifiers. 
 In linear regression models, \citet{LCL20} considers the estimation of high-dimensional regression coefficient vectors when the difference between the auxiliary and the target model is sufficiently sparse and proves the minimax optimal rate. 
 \citet{MJ20} proposes an algorithm that assumes all the auxiliary studies and the target study share a common, low-dimensional linear representation. Transfer learning in general functional classes have been studied in \citet{MJ20theory} and \citet{HK20}.

Our proposed transfer learning algorithm aims to improve the estimation and inference accuracy for GGM in a target study by transferring information  from multiple related studies.  This is different from the multi-task learning  outlined above, where the goal is to simultaneously estimate multiple graphs.  In terms of theoretical results, the convergence rate for estimating the target graph  in transfer learning can be faster than the corresponding rate in the multi-task learning.

 
\subsection{Model set-up}
\label{sec1-1}
Suppose that we observe \textit{i.i.d.} samples $x_i\in\R^p$ generated from $N(0,\Sig)$, $i=1,\dots,n$, and the parameter of interest is the precision matrix $\Omega=\Sig^{-1}$. Indeed, $\Omega$ uniquely determines the conditional dependence structure and the corresponding graph.  If the $i$-th and $j$-th variables are conditionally dependent in the target study, there is an undirected edge between the $i$-th and $j$-th nodes in the Gaussian graph and, equivalently, the $(i,j)$-th entry of $\Omega$ is nonzero. Our focus is on the estimation and inference for high-dimensional sparse Gaussian graphs where $p$ can be much larger than $n$ and $\Omega$ is sparse such that each column of $\Omega$ has at most $s$ nonzero elements with $s\ll p$.  

In the transfer learning setting, in addition to the observations $\{x_1, ..., x_n\}$ from the target distribution  $N(0,\Sig)$, we also observe samples from $K$ auxiliary studies. For $k=1,\dots, K$, the observations $x^{(k)}_i\in\R^p$ are independently generated from $N(0,\Sig^{(k)})$, $i=1,\dots,n_k$.
 Let $\Omega^{(k)}=\{\Sig^{(k)}\}^{-1}$ be the precision matrix of the $k$-th study, $k=1,\dots, K$. If some knowledge can be transferred to the target study, a certain level of similarity needs to be possessed by the auxiliary models and the target one.  
 
 To motivate our proposed similarity measure,  consider the relative entropy, or equivalently the Kullback–Leibler (KL) divergence, between the $k$-th auxiliary model and the target model. That is,
\begin{equation}
\label{eq-KL}
\mathcal{D}_{KL}(N_{\Sig^{(k)}}\parallel N_{\Sig})=\frac{1}{2}\textup{Tr}(\Delta^{(k)})-\frac{1}{2}\log det(I_p+\Delta^{(k)}) ~~\text{for}~~\Delta^{(k)}=\Omega\Sig^{(k)}-I_p,
\end{equation}
where $N_{\Sig^{(k)}}$ and $N_{\Sig}$ denote the normal distributions with mean zero and covariance matrix $\Sig^{(k)}$ and $\Sig$, respectively.
The KL-divergence is parametrized by the matrix $\Delta^{(k)}$ and we call $\Delta^{(k)}$ the $k$-th divergence matrix. 
We characterize the difference between $\Omega$ and $\Omega^{(k)}$ via
\begin{equation}
\label{eq-D}
\mathcal{D}_q(\Omega^{(k)},\Omega)=\max_{1\leq j\leq p}\|\Delta^{(k)}_{j,.}\|_q+\max_{1\leq j\leq p}\|\Delta^{(k)}_{.,j}\|_q
\end{equation}
for some fixed $q\in[0,1]$. In words, $\mathcal{D}_q(\Omega,\Omega^{(k)})$ is the maximum row-wise $\ell_q$-sparsity of $\Delta^{(k)}$ plus the maximum column-wise $\ell_q$-sparsity. Both the row-wise and column-wise norms are taken into account because $\Delta^{(k)}$ is non-symmetric.
The quantity $\mathcal{D}_q(\Omega^{(k)},\Omega)$ measures the ``relative distance'' between $\Omega$ and $\Omega^{(k)}$ in the sense that $\mathcal{D}_q(\Omega^{(k)},\Omega)=\mathcal{D}_q(c\Omega^{(k)},c\Omega)$ for any constant $c>0$. Notice that the spectral norm of $\Delta^{(k)}$  is upper bounded by $\mathcal{D}_1(\Omega^{(k)},\Omega)$,  which further provides an upper bound on the KL-divergence.
We also define $\mA_q$ to be a subset of $\{1,\dots,K\}$ such that
\begin{equation}
\label{eq-A}
\max_{k\in\mA_q} \mathcal{D}_q(\Omega^{(k)},\Omega)\leq h.
\end{equation}
We call $\mA_q$ the informative set under the difference measure $\mathcal{D}_q$ since all the auxiliary studies in $\mA_q$ have the discrepancy no larger than $h$. 

We develop estimation and inference procedures for GGMs given the informative set $\mA_q$ for any fixed $q\in[0,1]$. 
To the best of our knowledge,  estimation and inference of graphical models have not been studied in the transfer learning setting.

\subsection{Our contributions}

A transfer learning algorithm, called Trans-CLIME,  is proposed for estimating the target GGM. The proposed algorithm is inspired by the CLIME  introduced in \citet{Cai11} in the single study setting and is computationally efficient.  Furthermore, edge detection with uncertainty quantification is considered. Specifically, we construct the confidence interval for an edge of interest and perform multiple testing for all the edges with false discovery rate (FDR) control. The statistical inference is based on a new debiasing procedure, which can be coupled with any initial graph estimators. The debiasing step can be analytically computed in one step. We demonstrate its asymptotic validity for inference and its application to multiple testing with FDR control.

Theoretically, we establish the minimax optimal rate of convergence for estimating the GGMs with transfer learning in Frobenius norm by providing matching minimax upper and lower bounds.  We also establish the optimal rate of convergence for estimating individual entries in the graph.
These convergence rates are faster than the corresponding minimax rates in the classical single study setting, where no auxiliary samples are available or used. Our proposed Trans-CLIME and debiased Trans-CLIME are shown to be rate optimal under proper conditions.

\subsection{Organization and notation}
The rest of this paper is organized as follows. In Section \ref{sec2}, we propose an algorithm for estimating the graph in transfer learning with $q=1$. We study statistical inference for each edge of the graph in Section \ref{sec3}. In Section \ref{sec4}, we consider multiple testing of all the edges in the graph with false discovery rate guarantee. We establish the minimax lower and upper bounds for any fixed $q\in[0,1]$ in  Section \ref{sec-l0}. In Section \ref{sec-simu}, we study the numerical performance of Trans-CLIME in comparison to some other relevant methods. We then present an application of the proposed methods to estimate gene regulatory graphs based on data from multiple brain tissues in Section \ref{sec-data}, Finally,  Section \ref{sec-diss} concludes the paper. The the proofs and other supporting information are given in the Supplementary Materials.

For a matrix $A\in\R^{p\times p}$, let $A_j$ denote the $j$-th column of $A$. For any fixed $j\leq p$, we call $\|A_j\|_2$ the column-wise $\ell_2$-norm of $A$.  Let $\|A\|_{\infty,2}=\max_{j\leq p}\|A_j\|_2$, $\|A\|_{\infty,1}=\max_{j\leq p}\|A_j\|_1$, $\|A\|_{\infty,\infty}=\max_{i,j\leq p}|A_{i,j}|$, and $\|A\|_1=\sum_{j=1}^p\|A_j\|_1$. Let $\|A\|_2$ denote the spectral norm of $A$ and $\|A\|_F$ denote the Frobenius norm of $A$.
For a symmetric matrix $A$, let $\Lambda_{\max}(A)$ and $\Lambda_{\min}(A)$ denote the largest and smallest eigenvalues of $A$, respectively. We use $c_0,c_1,\dots$ and $C_0,C_1,\dots$ as generic constants which can be different at different places.
   
 \section{GGM estimation given the informative set}
     \label{sec2}
In this section, we study transfer learning in GGM estimation when the informative set $\mA_q$ is known. We focus on the difference measure with $q=1$. The subscript of $\mA_1$ will be abbreviated in the sequel without special emphasis.
In Section \ref{sec2-1}, we introduce the rationale for the proposed algorithm. Our proposal is introduced in Section \ref{sec2-2} and its theoretical properties are studied in Section \ref{sec2-3}.

\subsection{Rationale and moment equations}
\label{sec2-1}
Statistical methods in parametric models are always derived based on some moment equations. For estimating the GGMs, the likelihood is a natural objective function to optimize \citep{Glasso, Roth08}. The score function of the maximum likelihood estimator gives the following moment equation:
\begin{equation}
\label{m0}
   \Sig\Omega-I_p=0.
\end{equation}
The idea of CLIME \citep{Cai11} is to solve an empirical version of (\ref{m0}) and to encourage the sparsity of the estimator. Specifically, the CLIME estimator is given as 
\begin{align}
\widehat{\Omega}^{(\textup{CL})}&=\argmin_{\Omega\in\R^{p\times p}} \|\Omega\|_1\label{eq-CLIME}\\
&\text{subject to}~\|\widehat{\Sig}\Omega-I_p\|_{\infty,\infty}\leq \lam_{\textup{CL}},\nonumber
\end{align}
where $\widehat{\Sig}$ is a sample covariance matrix and $\lam_{\textup{CL}}>0$ is a tuning parameter.

In the context of transfer learning, we re-express the moment equation (\ref{m0}) to incorporate auxiliary information. Specifically, for $k=1,\dots, K$,
\[
I_p=\Sig^{(k)}\Omega^{(k)}=\Sig^{(k)}\Omega-(\Delta^{(k)})^{\intercal},\label{eq-m3}
\]
where $\Delta^{(k)}$ is the divergence matrix defined in (\ref{eq-KL}). To simultaneously leverage all the informative auxiliary studies, we further define the weighted average of the covariance and divergence matrices 
\[
  \Sig^{\mA}=\sum_{k\in\mA}\alpha_k\Sig^{(k)}~\text{and}~\Delta^{\mA}=\sum_{k\in\mA}\alpha_k\Delta^{(k)},
\]
where $\alpha_k=n_k/n_{\mA}$ for $n_{\mA}=\sum_{k\in \mA} n_k$.
The moment equation we consider is
\begin{equation}
\Sig^{\mA}\Omega-(\Delta^{\mA})^{\intercal}-I_p=0,\label{eq-m3}
\end{equation}
where  $\Sig^{\mA}$ in (\ref{eq-m3}) is an average parameter over $\mA$ and it incorporates the auxiliary information. 
 The moment equation (\ref{eq-m3}) motivates our procedure. First, we will estimate $\Delta^{\mA}$ 
based on the following moment equation:
\begin{equation}
\label{eq-m3a}
  \Sig  \Delta^{\mA}-(\Sig^{\mA}-\Sig)=0.
\end{equation}
Once $\Delta^{\mA}$ is identified, we can estimate our target $\Omega$ via (\ref{eq-m3}).

In most problems of interest, the similarity between $\Omega^{(k)}$ and $\Omega$ can be weak, even for $k\in\mA$, i.e. the unknown $h$ can be large. In this case, information  transfer may negatively affect  the learning performance of the target problem, i.e., the ``negative transfer'' \citep{HK20}. To address this issue, we will further perform an aggregation step. The aggregation methods and theory have been extensively studied in the existing literature, to name a few, \citet{RT11, Tsybakov14, Qagg, Dai18}. This type of methods can guarantee that, loosely speaking, the aggregated estimator has prediction performance comparable to the best prediction performance achieved by the initial estimators.

\subsection{Trans-CLIME algorithm}
\label{sec2-2}
We  introduce our proposed transfer learning algorithm, Trans-CLIME. For the data from target study, we split them into two disjoint folds. Let $\mathcal{I}$ be a subset of $\{1,\dots,n\}$ such that $|\mathcal{I}|=cn$ for some constant $0<c<1$. Let $\mathcal{I}^c$ denote the complement of $\mathcal{I}$. Let 
\[
\widehat{\Sig}=\frac{1}{|\mathcal{I}|}\sum_{i\in\mathcal{I}}x_ix_i^{\intercal}~~\text{and}~~\widetilde{\Sig}=\frac{1}{|\mathcal{I}^c|}\sum_{i\in\mathcal{I}^c}x_ix_i^{\intercal}.
\]
We will use $\widehat{\Sig}$ in Step 1 and 2 and will use $\widetilde{\Sig}$ in Step 3. 
For the auxiliary data, let $\widehat{\Sig}^{\mA}=\sum_{k\in\mA} (X^{(k)})^{\intercal}X^{(k)}/(\sum_{k\in \mA}n_k)$ denote the sample covariance based on the informative auxiliary samples.
Compute the single-study CLIME estimator $\widehat{\Omega}^{(\textup{CL})}$ via (\ref{eq-CLIME}) with input $\widehat{\Sig}$.

\underline{Step 1}. Compute
\begin{align}
\widehat{\Delta}^{(init)}=&\argmin_{\Delta\in\R^{p\times p}} \|\Delta\|_1 \label{eq-Delta-init}\\
&\text{subject to} \quad\|\widehat{\Sig}\Delta-(\widehat{\Sig}^{\mA}-\widehat{\Sig})\|_{\infty,\infty}\leq \lam_{\Delta}.\nonumber
\end{align}
The optimization in (\ref{eq-Delta-init}) is a CLIME-type estimator based on the moment equation (\ref{eq-m3a}). The obtained $\widehat{\Delta}^{(init)}$ is column-wise sparse but not necessarily row-wise sparse. We refine $\widehat{\Delta}^{(init)}$ as follows.
\begin{align}
\widehat{\Delta}^{\mA}=&\argmin_{\Delta\in\R^{p\times p}} \|\Delta\|_1 \label{eq-DeltaA-est}\\
&\text{subject to} \quad\|\Delta-\widehat{\Delta}^{(init)}-\widehat{\Omega}^{(\textup{CL})}(\widehat{\Sig}^{\mA}-\widehat{\Sig}-\widehat{\Sig}\widehat{\Delta}^{(init)})\|_{\infty,\infty}\leq 2\lam_{\Delta}.\nonumber
\end{align}
The optimization (\ref{eq-DeltaA-est}) can be understood as an adaptive thresholding of the bias-corrected $\widehat{\Delta}^{(init)}$, $\widehat{\Delta}^{(init)}+\widehat{\Omega}^{(\textup{CL})}(\widehat{\Sig}^{\mA}-\widehat{\Sig}-\widehat{\Sig}\widehat{\Delta}^{(init)})$. It is a more sophisticated version of hard thresholding and it does not require the knowledge of unknown parameters. The resulted $\widehat{\Delta}^{\mA}$ is row-wise $\ell_1$-sparse and will be used in the next step.

\underline{Step 2}. For $\widehat{\Delta}^{\mA}$ defined in (\ref{eq-DeltaA-est}), compute
\begin{align}
\widehat{\Theta}=&\argmin_{\Theta\in\R^{p\times p}} \|\Theta\|_1\label{eq-Theta-est}\\
&\text{subject to} \quad\|\widehat{\Sig}^{\mA}\Theta-(\widehat{\Delta}^{\mA}+I_p)^{\intercal}\|_{\infty,\infty}\leq \lam_{\Theta}.\nonumber
\end{align}
This step is a CLIME-type optimization based on the moment equation (\ref{eq-m3a}). As we have discussed in Section \ref{sec2-1}, $\widehat{\Theta}$ may not be as good as the single-study estimator if the similarity is weak. Hence, we perform a least-square aggregation in Step 3. The least square aggregation has been well-studied for regression type of problems \citep{Tsybakov14}. In this work, we aggregate the single-study CLIME estimator and $\widehat{\Theta}$ to produce a final graph estimator. Loosely speaking, the moment equation which motivates $\hat{v}_j$ is
\[
   \widetilde{\Sig}(\widehat{\Omega}^{(\textup{CL})}_{j},\widehat{\Theta}_{j})v_j-e_j\approx 0.
\]
Notice that the sample splitting step guarantees that both $\widehat{\Theta}$ and $\widehat{\Omega}^{(\textup{CL})}$ are independent of the samples used for aggregation.

\underline{Step 3}. For $j=1,\dots, p$, compute
\begin{align*}
&\widehat{W}(j)=\begin{pmatrix}
(\widehat{\Omega}^{(\textup{CL})}_{j})^{\intercal}\widetilde{\Sig} \widehat{\Omega}_j^{(\textup{CL})}& (\widehat{\Omega}^{(\textup{CL})}_{j})^{\intercal}\widetilde{\Sig}\widehat{\Theta}_{j}\\
(\widehat{\Omega}^{(\textup{CL})}_{j})^{\intercal}\widetilde{\Sig}\widehat{\Theta}_{j} & \widehat{\Theta}_{j}^{\intercal}\widetilde{\Sig} \widehat{\Theta}_j
\end{pmatrix},\quad \hat{v}_j=\{\widehat{W}(j)\}^{-1}\begin{pmatrix}
\widehat{\Omega}^{(\textup{CL})}_{j,j}\\
\widehat{\Theta}_{j,j}
\end{pmatrix}\in\R^2,
\end{align*}
where $\widehat{\Omega}^{(\textup{CL})}$ is defined in (\ref{eq-CLIME}).
For $j=1,\dots,p$, let
\[
\widehat{\Omega}_j=(\widehat{\Omega}^{(\textup{CL})}_{j},\widehat{\Theta}_{j})\hat{v}_j.
\]

Computationally, all the optimizations in three steps can be separated into $p$ independent optimizations, analogous to the original CLIME algorithm.  This makes the computation scalable.

We mention a significant difference between the Trans-CLIME algorithm and the transfer learning in high-dimensional regression such as the oracle Trans-Lasso in \citet{LCL20}.  For the linear regression problems, the performance of the oracle Trans-Lasso is justified when $\{\Sig^{(k)}\}_{k\in \mA}$ are close enough to $\Sig$, i.e. the designs can be moderately heterogeneous. For the current problem, 
we can rephrase it using the node-wise regression point of view such that each column of $\Omega$ can be viewed as a target regression parameter, and  the covariance matrix of the designs, $\Sig^{(k)}$, are different up to the similarity constraint on $\Omega^{(k)}$ and $\Omega$.
The similarity constraint imposed on $\Delta^{(k)}$ allows larger heterogeneity on the design matrices than that in the regression setting. We get around the challenge of heterogeneous designs by choosing proper moment equations introduced in Section \ref{sec2-1}.

\subsection{Convergence rate of Trans-CLIME}
\label{sec2-3}
In this subsection, we provide theoretical guarantees for the Trans-CLIME algorithm.
We assume the following condition in our theoretical analysis.
\begin{condition}[Gaussian graphs]
\label{cond1}
For $i=1,\dots,n$, $x_i\in\R^p$ are i.i.d. distributed as $N(0,\Sig)$. For each $k\in\mA$, $x_i^{(k)}$ are i.i.d. distributed as $N(0,\Sig^{(k)})$ for $i=1,\dots,n_k$. It holds that $1/C\leq \Lambda_{\min}(\Sig)\leq \Lambda_{\max}(\Sig)\leq C$ and  $1/C\leq \min_{k\in \mA}\Lambda_{\min}(\Sig^{(k)})\leq \max_{k\in \mA}\Lambda_{\max}(\Sig^{(k)})\leq C$.
\end{condition}
The Gaussian assumption facilitates the justification of the restricted eigenvalue conditions of the empirical covariance matrices. The Gaussian distribution of the primary data also simplifies the limiting distribution of our proposed estimator for inference. 

The parameter space we consider is
\begin{align}
\label{eq-Gam}
\mathbb{G}_q(s,h)&=\left\{(\Omega,\Omega^{(1)},\dots,\Omega^{(K)}):~\max_{1\leq j\leq p}\|\Omega_j\|_0\leq s,~\max_{k\in\mA_q}\mathcal{D}_q(\Omega,\Omega^{(k)})\leq h\right\}.
\end{align}
 We mention that the parameter space for GGMs in single study setting \citep{Zhao15} can be written as $\mathbb{G}_q(s,\infty)$ under Condition \ref{cond1} for any $q\in[0,1]$. This is because $\mathbb{G}_q(s,\infty)$ allows the auxiliary study to be arbitrarily far away from the target study and hence the worse case scenario is equivalent to the setting where only the primary data is available.

In the following, we demonstrate the convergence rate of Trans-CLIME under Condition \ref{cond1}. Let $\delta_n=\sqrt{\log p/n}\wedge h$.
\begin{theorem}[Convergence rate of Trans-CLIME]
\label{thm2}
Assume Condition \ref{cond1} .
Let the Trans-CLIME estimator $\widehat{\Omega}$ be computed with 
\[
   \lam_{\Delta}= c_1\sqrt{\frac{\log p}{n}},~\text{ and} ~~\lam_{\Theta}= c_2\sqrt{\frac{\log p}{n_{\mA}}},
   \] where $c_1$, and $c_2$ are large enough constants.
 If $s^2\log p\leq c_3n$, then for any true models in $\mathbb{G}_1(s,h)$, we have
\begin{align}
&\frac{1}{p}\|\widehat{\Omega}-\Omega\|_{F}^2\vee \|\widehat{\Omega}_j-\Omega_j\|_2^2=O_P\left(\frac{s\log p}{n_{\mA}+ n}+h\delta_n\wedge \frac{s\log p}{n}+\frac{1}{n}\right)\label{eq-re1}
\end{align}
 for any fixed $1\leq j\leq p$.
\end{theorem}
Theorem \ref{thm2} demonstrates that under proper choice of tuning parameters, upper bounds can be obtained in column-wise $\ell_2$-norm and in Frobenius norm. The sparsity condition $s\sqrt{\log p}\lesssim\sqrt{n}$ guarantees  a sufficiently fast convergence rate of $\widehat{\Delta}^{\mA}$ and the restricted eigenvalue conditions in Step 2. This sparsity condition has also been considered in \citet{Cai16, Tiger} for establishing the minimax optimality results.

 We first explain the convergence rate of $\widehat{\Omega}$ in column-wise $\ell_{2}$-norm. As all the $\Omega^{(k)}$, $k\in\mA$, share the column-wise $s$-sparse matrix $\Omega$,
  the term $s\log p/(n_{\mA}+n)$ comes from estimating $\Omega$ based on $n_{\mA}+n$ independent samples. 
The term $h\delta_n$ comes from the convergence rate of  $\widehat{\Delta}^{\mA}$ in row-wise $\ell_2$-norm. It is dominated by a relatively small sample size $n$ because the divergence matrix can only be identified based on the primary samples. 
The minimal term $h\delta_n\wedge s\log p/n$ is the faster convergence rate achieved by $\widehat{\Theta}$ and $\widehat{\Omega}^{(\textup{CL})}$, which is a consequence of the least square aggregation performed in Step 3. However, there is a cost of aggregation, which is $O_P(n^{-1})$ in the current problem and it is negligible in most parameter spaces of interest. 

 To understand the gain of transfer learning, we compare the current results with the convergence rate of CLIME in single study setting.
 \begin{remark}
 \label{re1}
Assume Condition \ref{cond1} and $s^2\log p=o(n)$. For the CLIME estimator $\widehat{\Omega}^{(\textup{CL})}$ defined in (\ref{eq-CLIME}) with $\lam_{\textup{CL}}=c_1\sqrt{\log p/n}$ with large enough $c_1$, it can be shown that for any true models in $\mathbb{G}_1(s,\infty)$,
\begin{align*}
  &\|\widehat{\Omega}^{(\textup{CL})}_j-\Omega_j\|_2^2\vee\frac{1}{p}\|\widehat{\Omega}^{(\textup{CL})}-\Omega\|_F^2=O_P\left(\frac{s\log p}{n}\right)
\end{align*}
for any fixed $1\leq j\leq p$.
\end{remark}
We see that the convergence rate of $\widehat{\Omega}$ in Frobenius norm is no worse than the CLIME for any $s\geq 1$. Furthermore, $\widehat{\Omega}$ has faster convergence rate when $n_{\mA}\gg n$ and $h\delta_n\ll s\log p/n$ for $s\geq 1$. One sufficient condition for improvement is $n_{\mA}\gg n$ and $h\ll s\sqrt{\log p/n}$. That is, if the total sample size of informative auxiliary samples are much larger than the primary sample size and the similarity is sufficiently strong, then a significant amount of knowledge can be transferred by using Trans-CLIME.

\section{Inference for each entry in the graph}
\label{sec3}
In this section, we propose a debiasing scheme for inference of each entry in the graph. The main features of this method are its flexibility to couple with any initial graph estimator and  its computational efficiency.
We first introduce the rationale for our construction, then illustrate the method, and provide theoretical guarantees in the end.

\subsection{Rationale of debiasing entry-wise estimates}
To make inference of $\Omega_{i,j}$, we write it into a quadratic form:
\begin{align}
\label{quad1}
\Omega_{i,j}&=\Omega_{i}^{\intercal}\Sig\Omega_{j}=\Omega_{i}^{\intercal}\E[\Sig^n]\Omega_{j},
\end{align}
where $\Sig^n$ denotes the sample covariance matrix based on a subsample of the primary data. In many occasions, $\Sig^n$ can be computed based on all the primary data. Sometimes for a sharp theoretical analysis, sample splitting is performed and $\Sig^n$ can be computed based on a constant proportion of the primary data. Equation  (\ref{quad1}) holds for any inverse covariance matrix $\Omega$ not restricting to Gaussian random graphs.

Leveraging (\ref{quad1}), we are able to use the idea of debiasing quadratic forms \citep{CG18b} to make inference of $\Omega_{i,j}$. Specifically,  $\Omega_{i,j}$ takes the same format as the co-heritability if we view $\Omega_{i}$ and $\Omega_{j}$ as the regression coefficient vectors for two different outcomes and view $X$ as the measurements of genetic variants. Motivated by this observation, we arrive at the following debiased estimator of $\Omega_{i,j}$. Let $\Omega^{(init)}$ be any initial estimator of $\Omega$. The corresponding debiased estimator is 
\begin{align}
\Omega_{i,j}^{(db)}&=(\Omega^{(init)}_{i})^{\intercal}\Sig^n\Omega^{(init)}_{j}+(\Omega^{(init)}_{i})^{\intercal}(e_j-\widehat{\Sig}\Omega^{(init)}_{j})+(\Omega^{(init)}_{j})^{\intercal}(e_i-\Sig^n\Omega^{(init)}_{i})\nonumber\\
&=\Omega^{(init)}_{j,i}+\Omega^{(init)}_{i,j}-(\Omega^{(init)}_{j})^{\intercal}\Sig^n\Omega^{(init)}_{i}.\label{eq-Theta-db}
\end{align}
We mention that $\Omega^{(init)}$ is not necessarily symmetric and hence we distinguish $\Omega^{(init)}_{i,j}$ and $\Omega^{(init)}_{i,j}$. 
It is easy to see that the above debiasing procedure can be coupled with any $\widehat{\Omega}^{(init)}$, including, say, graphical Lasso \citep{Glasso}, CLIME \citep{Cai11},  multi-task graph estimators \citep{JointGGM, JointGLasso, Cai2016joint}, and our proposed Trans-CLIME. In comparison to \citet{Liu13} and \citet{Zhao15}, where the debiased estimators are constructed using  node-wise regression, our proposal in (\ref{eq-Theta-db}) is more flexible in incorporating various types of initial estimators. 
To distinguish the samples for constructing $\widehat{\Omega}^{(init)}$ and the samples used in $\Sig^n$, we will call the samples involved in $\Sig^n$ the debiasing samples.

  \subsection{Entry-wise confidence intervals}
  We now formally introduce the algorithm for debiasing the Trans-CLIME.
  
\vspace{0.1in}\begin{algorithm}[H]
 \SetKwInOut{Input}{Input}
    \SetKwInOut{Output}{Output}
\SetAlgoLined
 \Input{Trans-CLIME estimator $\widehat{\Omega}$, sample covariance matrix $\widetilde{\Sig}$, and confidence level $\alpha$}
\Output{Debiased estimator $\widehat{\Omega}_{i,j}^{(db)}$ and a confidence interval for $\Omega_{i,j}$}

\underline{Step 1}    
For each $1\leq i,j\leq p$,
\begin{equation}
\label{eq-db}
 \widehat{\Omega}_{i,j}^{(db)}
 =\widehat{\Omega}_{j,i}+\widehat{\Omega}_{i,j}-\widehat{\Omega}_j^{\intercal}\widetilde{\Sig}\widehat{\Omega}_{i}.
\end{equation}

  Estimate the variance of $\widehat{\Omega}_{i,j}^{(db)}$ via
  \[
     \widehat{V}_{i,j}=\widehat{\Omega}_{i,i}\widehat{\Omega}_{j,j}+ \widehat{\Omega}_{i,j}\widehat{\Omega}_{j,i}.
\]
\underline{Step 2}. $100\times (1-\alpha)$\% two-sided confidence interval for $\Omega_{i,j}$ is
\[
   \widehat{\Omega}_{i,j}^{(db)}\pm z_{1-\alpha/2} (\widehat{V}_{i,j}/n)^{1/2}.
\]
 \caption{\textbf{Confidence interval for $\Omega_{i,j}$}} 
 \label{alg-db}
\end{algorithm}
In Algorithm \ref{alg-db}, we only use a proportion of primary data, i.e., those involved in $\widetilde{\Sig}$, as debiasing samples, while the realization of $\widehat{\Omega}$ involves both primary and auxiliary information. This is because first, only the primary data are known to be unbiased; second, the samples involved in $\widetilde{\Sig}$ are ``weakly'' dependent with $\widehat{\Omega}$ and can provide a relatively sharp convergence rate. 
The variance estimator $\widehat{V}_{i,j}$ is based on the limiting distribution of $\widehat{\Omega}^{(db)}_{i,j}$ given that the observations are Gaussian distributed. 
 
  \subsection{Theoretical results of debiased Trans-CLIME estimator}
\label{sec3-3}

  \begin{theorem}[Asymptotic normality for debiased Trans-CLIME]
\label{thm-db}
Under the conditions of Theorem \ref{thm2}, for any true models in $\mathbbm{G}_1(s,h)$ and any fixed $1\leq i, j\leq p$,
\[
\widehat{\Omega}_{i,j}^{(db)}-\Omega_{i,j}=\widehat{\zeta}_{i,j}+\widehat{T}_{i,j},
\]
where
\[
   \frac{\sqrt{n}\widehat{\zeta}_{i,j}}{V_{i,j}^{1/2}}\xrightarrow{D}  N(0,1),
\]
$V_{i,j}=\Omega_{i,i}\Omega_{j,j}+\Omega^2_{i,j}$, and
\[
\widehat{T}_{i,j}=O_P\left(\frac{s\log p}{n_{\mA}+n}+h\delta_n\wedge \frac{s\log p}{n}\right)+o_P(n^{-1/2}),
\]
The variance estimator satisfies, for any $1\leq i,j\leq p$,
\[
   |\widehat{V}_{i,j}-V_{i,j}|\leq C\widehat{T}_{i,j}^{1/2}~~\text{for some constant $C>0$}.
\]
\end{theorem}
The term $\widehat{\zeta}_{i,j}$ is the asymptotic normal part. It has convergence rate $n^{-1/2}$ as $\widehat{\Omega}^{(db)}_{i,j}$ only incorporates primary data as debiasing samples.
The term $\widehat{T}_{i,j}$ is the remaining bias of $\widehat{\Omega}^{(db)}_{i,j}$ and its rate is dominated by the rate of $\widehat{\Omega}$ in column-wise $\ell_2^2$-norm. We conclude from Theorem \ref{thm-db} that the convergence rate of $\widehat{\Omega}^{(db)}_{i,j}$ is 
\[n^{-1/2}+\frac{s\log p}{n_{\mA}+n}+h\delta_n\wedge \frac{s\log p}{n}.
\]
In comparison, the minimax optimal rate for estimating $\Omega_{i,j}$ in $\mathbb{G}_1(s,\infty)$ is $n^{-1/2}+s\log p/n$ \citep{Zhao15}. We see that the convergence rate in the transfer learning setting is always no worse than the rate in the single study setting. 

We now discuss the improvement with transfer learning. For the asymptotic normality to hold, one  requires $\widehat{T}_{i,j}=o_P(n^{-1/2})$, which gives the sparsity condition that
\begin{equation}
\label{ssc1}
s\log p\ll n_{\mA}/\sqrt{n}~~\text{and}~~ h\delta_n\wedge s\log p/n\ll n^{-1/2}.
\end{equation}
In comparison, the sparsity condition given by the minimax rate in single study setting is $s\log p\ll\sqrt{n}$. We see that the sparsity condition in (\ref{ssc1}) is weaker when $ h\delta_n\ll s\log p/n$ and $n_{\mA}\gg n$. From the discussion below Theorem \ref{thm2}, we  conclude that if $\widehat{\Omega}$ has faster convergence rate than the single-study minimax estimator CLIME, then inference based on Trans-CLIME requires weaker sparsity conditions.

The main challenge in deriving Theorem \ref{thm-db} is that $\widehat{\Omega}$ does not have a sufficiently fast convergence rate in column-wise $\ell_1$-norm. 
As a result, we can only utilize the estimation guarantees in  column-wise $\ell_2$-norm. As far as we know, existing analysis of debiased procedures that only involve $\ell_2$-guarantees either require the debiasing  samples and the initial estimators to be independent or require stronger technical conditions \citep{CG17, JM18}. 
In the current analysis, $\widetilde{\Sig}$ is dependent with $\widehat{\Omega}$ due to the aggregation step. We carefully analyze this dependence and conclude a desirable bound without extra conditions or extra sample splits.

As we have mentioned at the beginning of this section, our proposed debiasing scheme can be applied to many other initial estimators for different purposes. In the Supplementary Materials, we prove that applying the proposed debiasing scheme to $\widehat{\Omega}^{(\textup{CL})}$, termed as debiased CLIME, is rate optimal for $\Omega_{i,j}$ in $\mathbb{G}_q(s,\infty)$. The debiased CLIME has the same asymptotic distribution as the node-wise regression estimators considered in \citet{Liu13} and \citet{Zhao15}.

\section{Edge detection with FDR control }
\label{sec4}

 An important task regarding the graphical models is edge detection with uncertainty quantification. 
That is, we consider testing
  \[
     (H_0)_{i,j}:  \Omega_{i,j}=0 ~~1\leq i< j\leq p.
  \]
This is a multiple testing problem with $q=p(p-1)/2$ hypotheses to test in total. For the uncertainty quantification, we consider
the false discovery proportion (FDP) and false discovery rate (FDR). Let $\widehat{\mathcal{R}}$ denote the set of rejected null hypothesis. The FDP and FDR are defined as, respectively,
  \[
     \text{FDP}(\widehat{\mathcal{R}})=\frac{\sum_{(i,j)\in\widehat{\mathcal{R}}} \mathbbm{1}((i,j)\in\mathcal{H}_0)}{|\widehat{\mathcal{R}}|\vee 1}~~\text{and}~~\text{FDR}(\widehat{\mathcal{R}})=\E[ \text{FDP}(\widehat{\mathcal{R}})],
  \]
where  $\mathcal{H}_0$ is the set of true nulls.
  Many algorithms have been proposed and studied for FDR and FDP control in various settings. Especially, \citet{Liu13} proposes an FDR control algorithm for GGMs which can be easily combined with our proposed debiased estimator.  The proposed procedure is presented as Algorithm \ref{alg-fdr}.

\vspace{0.1in}\begin{algorithm}[H]
 \SetKwInOut{Input}{Input}
    \SetKwInOut{Output}{Output}
\SetAlgoLined
 \Input{$\{\widehat{\Omega}^{(db)}_{i,j}\}_{i<j}$, $\{\widehat{V}_{i,j}\}_{i<j}$, and FDR level $\alpha$}
\Output{A set of selected edges $\widehat{\mathcal{R}}$}
\underline{Step 1}.
For $1\leq i<j\leq p$, let
  \[
      \widehat{z}_{i,j}=\frac{\sqrt{n}\widehat{\Omega}_{i,j}^{(db)}}{\widehat{V}_{i,j}^{1/2}},
  \]
  where $\widehat{\Omega}_{i,j}^{(db)}$ and $\widehat{V}_{i,j}^{1/2}$ are defined in Algorithm \ref{alg-db}.

\underline{Step 2}.
  \begin{equation}
  \label{eq-that}
        \hat{t}=\inf\left\{t\in[0,\sqrt{2\log q-2\log\log q}]:~~\frac{q\Phi^c(t)}{\max\{\sum_{1\leq i<j\leq p}\mathbbm{1}(|\widehat{z}_{i,j}|\geq t),1\}}\leq \alpha\right\}.
  \end{equation}
  If (\ref{eq-that}) does not exist, we set $\hat{t}=\sqrt{2\log q}$.
  
  \underline{Step 3}.
  The rejected hypotheses are
\[
\widehat{\mathcal{R}}=\{(i,j): |\widehat{z}_{i,j}|\geq \hat{t},1\leq i<j\leq p\}.
\]

 \caption{\textbf{Edge detection with FDR control at level $\alpha$}} \label{alg-fdr}
\end{algorithm}

\subsection{Theoretical background for multiple testing}
Let $q_0=|\mathcal{H}_0|$ denote the cardinality of $\mathcal{H}_0$ and $q=(p^2-p)/2$ denote the total number of hypotheses to test.
Define a subset of  random variables ``highly'' correlated with the $i$-th variable
\[
\mathcal{C}_i(\gam)=\left\{j:1\leq j\leq p,j\neq i,|\Omega_{i,j}|\geq (\log p)^{-2-\gam}\right\}.
\]
  \begin{theorem}[FDR control]
\label{thm-fdr}
Let $p\leq n^r$ for some $r>0$ and $q_0\geq cp^2$ for some $c>0$. 
Assume the conditions of Theorem \ref{thm2},
\[
 s(\log p)^{3/2}\ll n_{\mA}/\sqrt{n}, ~~ h\delta_n\wedge s\log p/n\ll (n\log p)^{-1/2}
  \] and $\max_{1\leq i\leq p} |\mathcal{C}_i(\gam)|=O(p^\rho)$ for some $\rho<1/2$ and $\gam>0$.
We have
\begin{align*}
\lim_{(n,p)\rightarrow \infty} \frac{\textup{FDR}(\widehat{\mathcal{R}})}{\alpha q_0/q}=1~\text{and}~\frac{\textup{FDP}(\widehat{\mathcal{R}})}{\alpha q_0/q}\rightarrow 1\text{in probability}
\end{align*}
as $(n,p)\rightarrow \infty$.
\end{theorem}
Theorem \ref{thm-fdr} implies that Algorithm \ref{alg-fdr} can asymptotically control FDR and FDP at nominal level under certain conditions. The sample size condition in Theorem \ref{thm-fdr} guarantees that the remaining bias of $\widehat{\Omega}^{(db)}_{i,j}$ is uniformly $o_P((n\log p)^{-1/2})$. The condition on the cardinality of $\mathcal{C}_i(\gam)$ guarantees that the $z$-statistics have mild correlations such that the FDR control is asymptotically valid. 
The proof of Theorem \ref{thm-fdr} is largely based on the proof in \citet{Liu13} and some technical improvements in \citet{JJ19}.

\section{Minimax optimal rates for $q\in[0,1]$}
\label{sec-l0}
In this section, we establish the minimax upper and lower bounds for estimation and inference of GGMs in the parameter space $\mathbb{G}_q(s,h)$ for any fixed $q\in[0,1]$. In practice, the setting with $q\in[0,1)$ can imply relatively strong similarity conditions. Hence, we only provide the theoretical results for $q\in[0,1)$. 

\subsection{Optimal rates under Frobenius norm}

\begin{theorem}[Minimax bounds under Frobenius norm]
\label{thm-mini-frob2}
Assume Condition \ref{cond1} and $3<s\log p<c_1n$ for some small constant $c_1$. (i) If $ h\leq c_2n/\log p$ for some small enough constant $c_2$, then for some positive constants $C_1$, $C_2$ and $C_3$,
\begin{align*}
&\inf_{\widehat{\Omega}}\sup_{\mathbb{G}_0(s,h)}\P\left(\frac{1}{p}\|\widehat{\Omega}-\Omega\|_F^2\geq C_1\left\{\frac{s\log p}{n_{\mA_0}+n}+(h\wedge s)\frac{\log p}{n}\right\}\right)>1/4.\\
&\inf_{\widehat{\Omega}}\sup_{\mathbb{G}_0(s,h)}\P\left(\frac{1}{p}\|\widehat{\Omega}-\Omega\|_F^2\geq C_2\left\{\frac{s\log p}{n_{\mA_0}+n}+(h\wedge s)\frac{\log p}{n}\right\}\right)\leq \exp(-C_3\log p).
\end{align*}
(ii) If $ h^q(\log p/n)^{1-q/2}<c_4<\infty$. Then for any fixed $q\in(0,1]$, there are some positive constants $C_4$, $C_5$, and $C_6$ such that
\begin{align*}
&\inf_{\widehat{\Omega}}\sup_{\mathbb{G}_q(s,h)}\P\left(\frac{1}{p}\|\widehat{\Omega}-\Omega\|_F^2\geq C_4\left\{\frac{s\log p}{n_{\mA_q}+n}+h^q\delta_n^{2-q}\wedge \frac{s\log p}{n}\right\}\right)>1/4.\\
&\inf_{\widehat{\Omega}}\sup_{\mathbb{G}_q(s,h)}\P\left(\frac{1}{p}\|\widehat{\Omega}-\Omega\|_F^2\geq C_5\left\{\frac{s\log p}{n_{\mA_q}+n}+h^q\delta_n^{2-q}\wedge \frac{s\log p}{n}\right\}\right)\leq \exp(-C_6\log p).
\end{align*}
\end{theorem}
Theorem \ref{thm-mini-frob2} establishes the minimax optimal rates under Frobenius norm. These lower bounds generalize the existing lower bound in $\mathbb{G}_q(s,\infty)$ \citep{ACLIME} to allow for arbitrarily small $h$.  We first mention that the estimator of $\Omega$ which achieves the minimax upper bounds depends on the relative magnitude of $h$ and $s$ and hence is not adaptive.
In fact, a minimax optimal estimator for $q=1$ is $\widehat{\Theta}$ if $h\delta_n\lesssim s\log p/n$ and $\widehat{\Omega}^{(\textup{CL})}$ if $h\delta_n\gg s\log p/n$. The cut-off concerns whether the informative auxiliary samples are useful or not, which depends on unknown parameters. In comparison, the Trans-CLIME estimator does not depends on the unknown parameter and is minimax optimal when $h\gtrsim n^{-1/2}$ under the conditions of Theorem \ref{thm2}. 


\subsection{Optimal rates for  estimating $\Omega_{i,j}$}
\begin{theorem}[Minimax bounds for estimating $\Omega_{i,j}$]
\label{thm-mini-db2}
Assume Condition \ref{cond1} and $3<s\log p<c_1n$ for some small constant $c_1$. (i) If $1\leq h\leq c_2n/\log p$ for some small enough constant $c_2$, then for some constant $C_1>0$,
\[
\inf_{\widehat{\Omega}}\sup_{\mathbb{G}_0(s,h)}\P\left(|\widehat{\Omega}_{i,j}-\Omega_{i,j}|\geq C_1 \left\{n^{-1/2}+\frac{s\log p}{n_{\mA_0}+n}+(h\wedge s)\frac{\log p}{n}\right\}\right)>1/4.
\]
For any constant $\eps_0>0$, there exists constant $C_2$ depending on $\eps_0$ such that
\[
\inf_{\widehat{\Omega}}\sup_{\mathbb{G}_0(s,h)}\P\left(|\widehat{\Omega}_{i,j}-\Omega_{i,j}|\geq C_2\{n^{-1/2}+\frac{s\log p}{n_{\mA_0}+n}+(h\wedge s)\frac{\log p}{n}\}\right)\leq \eps_0.
\]
(ii) If $ h^q(\log p/n)^{1-q/2}<c_3$ for some small enough constant $c_3$, then for any fixed $q\in(0,1]$, 
\[
\inf_{\widehat{\Omega}}\sup_{\mathbb{G}_q(s,h)}\P\left(|\widehat{\Omega}_{i,j}-\Omega_{i,j}|\geq C_3\{R_q+\frac{s\log p}{n_{\mA_q}+n}+ h^q\delta_n^{2-q}\wedge \frac{s\log p}{n}\}\right)>1/4,
\]
where $C_3$ is a positive constant and $R_q=(n_{\mA_q}+n)^{-1/2}+n^{-1/2}\wedge h$.
For any constant $\eps_0>0$, there exists constant $C_4$ depending on $\eps_0$ such that for any $q\in(0,1]$,
\[
\inf_{\widehat{\Omega}}\sup_{\mathbb{G}_q(s,h)}\P\left(|\widehat{\Omega}_{i,j}-\Omega_{i,j}|\geq C_4\{R_q+\frac{s\log p}{n_{\mA_q}+n}+ h^q\delta_n^{2-q}\wedge \frac{s\log p}{n}\}\right)\leq \eps_0.
\]
\end{theorem}
Theorem \ref{thm-mini-db2} establishes the minimax optimal rates for estimating each entry in the graph.
This lower bound generalizes the existing lower bound in $\mathbb{G}_q(s,\infty)$ \citep{Zhao15} to allow for arbitrarily small $h$.  We see that when $q=0$, the parametric rate is $n^{-1/2}$, which is same as in single study setting.
When $q\in(0,1]$, the parametric rate $R_q$ can be sharper than $n^{-1/2}$.  We know illustrate this phenomenon in details with $q=1$. 

For $q=1$, a minimax optimal estimator of $\Omega_{i,j}$ is $\widehat{\Omega}^{(db)}_{i,j}$ when $h\gtrsim n^{-1/2}$ and is debiased Trans-CLIME using $X^{(k)}$, $k\in\mA\cup\{0\}$, as debiasing samples when $h\ll n^{-1/2}$. In the scenario $h\ll n^{-1/2}$, the informative auxiliary studies are very similar to the target study and using $n_{\mA}+n$ debiasing samples can have faster parametric rate, $(n+n_{\mA})^{-1/2}$, with bias no larger than $h$. However, the central limit theory may not hold for the rate optimal estimator when $h\ll n^{-1/2}$. This is because the parametric rate is dominated by the bias $h$ when $(n+n_{\mA})^{-1/2}\lesssim h\lesssim n^{-1/2}$. In contrast, $\widehat{\Omega}^{(db)}_{i,j}$ has parametric rate $n^{-1/2}$ and its asymptotic normality holds for arbitrarily small $h$ under the conditions of Theorem \ref{thm-db}. Hence, $\widehat{\Omega}^{(db)}_{i,j}$ is a proper choice for statistical inference.  

\section{Numerical experiments}
\label{sec-simu}
We compare the performance of three methods using simulations. The first one is the proposed Trans-CLIME.  The second one is CLIME that only uses the data from the target study. 
 The third one is the Trans-CLIME which  assumes  $\mA=\{1,\dots, K\}$, which includes data from both informative and non-informative studies,  denoted by ``pooled''. We include the last method to understand the robustness of Trans-CLIME to non-informative auxiliary studies.
For the choice of tuning parameters, we consider $\lam_{\textup{CL}}=2c_n\sqrt{\log p/n}$ for CLIME. We pick $c_n$ to minimize the prediction error defined in (\ref{eq-Qhat}) based on five fold cross-validation. For the Trans-CLIME, we set $\lam_{\Delta}=2\hat{\sig}\sqrt{\log p/n}$ and $\lam_{\Theta}=2c_n\sqrt{\log p/n_{\mA}}$ where $c_n$ is the same as in the CLIME optimization. For the pooled-CLIME, the tuning parameters are set in the same way as in Trans-CLIME except that $n_{\mA}$ is replaced by $\sum_{k=1}^Kn_k$. For Trans-CLIME based methods, we split the target data into two folds such that $\widehat{\Omega}^{(\textup{CL})}$ and $\widehat{\Theta}$ are computed based on $2n/3$ samples and the aggregation step (Step 3) is based on the rest $n/3$ samples. For $\hat{v}_j$ obtained in Step 3, we project it onto a two-dimensional positive simplex. This is because the oracle $v_j$ is in that simplex. For the debiased Trans-CLIME, we use all the primary data as debiasing samples as it has a better empirical performance. The R code for the three methods is available at \url{https://github.com/saili0103/TransCLIME}.

We set $n=150$, $p=200$, $K=5$, and $n_k=300$ for $k=1,\dots,K$. We consider two types of precision matrix $\Omega$.
\begin{itemize}
\item[(i)] Banded matrix with bandwidth 8. For $1\leq i, j\leq p$, $\Omega_{i,j}=2\times 0.6^{|i-j|}\mathbbm{1}(|i-j|\leq 7)$. 

\item[(ii)] Block diagonal matrix  with block size 4, where  each block is Toeplitz $(1.2,0.9,0.6,0.3)$. 
\end{itemize} 
For $k\in \mA$, $\{\Delta^{(k)}\}_{i,j}$ is zero with probability 0.9 and is nonzero with probability 0.1. If an entry is nonzero, it is randomly generated from $U[-r/p,r/p]$ for $r\in\{10,20,30\}$. For such divergence matrices $\Delta^{(k)}$, the empirical $h \in\{5.72, 11.45, 15.89\}$ for $r\in\{10,20,30\}$ in setting (i) and the empirical $h \in\{9.11, 18.18, 24.37\}$ for $r\in\{10,20,30\}$  in setting (ii).
For $\Omega^{(k)}$, $k\notin \mA$, we generate $\{\Omega^{(k)}\}_{i,j}=1.5\mathbbm{1}(i=j)+\delta_{i,j}$, where $\delta_{i,j}$ is zero with probability 0.9 and is 0.2 with probability 0.1. For $k=1,\dots,K$, we symmetrize $\Omega^{(k)}$ and if $\Omega^{(k)}$ is not positive definite, we redefine $\Omega^{(k)}$ to be its positive definite projection. The positive definite projection is realized via R package ``BDCoColasso'' \citep{spd}.

\subsection{Estimation results}
In Figure \ref{fig1}, we report the estimation errors in Frobenius norm for three methods in setting (i) and (ii), respectively. As the number of informative auxiliary studies increases, $n_{\mA}$ increases and the estimation errors of two Trans-CLIME based methods decrease. As $r$ increases, the estimation errors of all three Trans-CLIME based methods increase.  Trans-CLIME has a faster convergence rate than the pooled version. This is because $K-|\mA|$ non-informative studies are used in the pooled version, which affects the convergence rates. We see from the pooled version that the Trans-CLIME algorithm is robust to the non-informative auxiliary studies as the performance of the pooled version is always not much worse than the single-study CLIME.
\begin{figure}[H]
\includegraphics[width=0.99\textwidth, height=5.5cm]{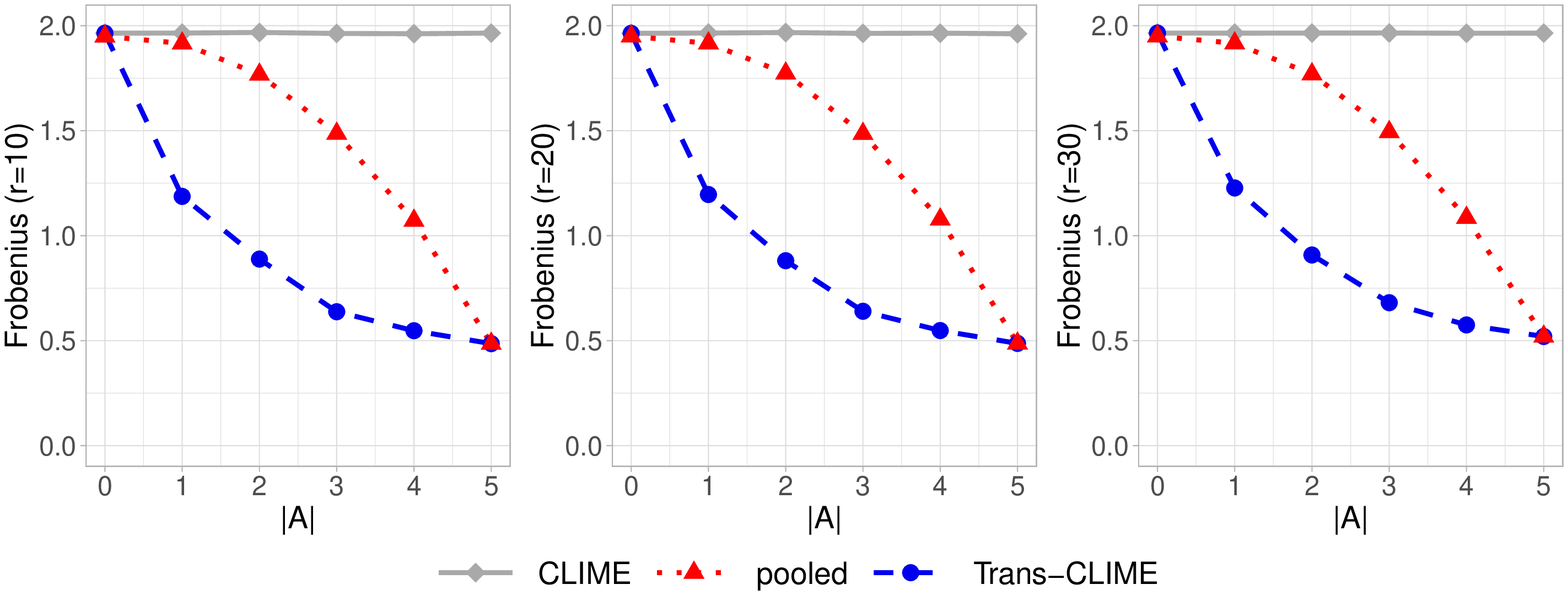}
\includegraphics[width=0.99\textwidth, height=5.5cm]{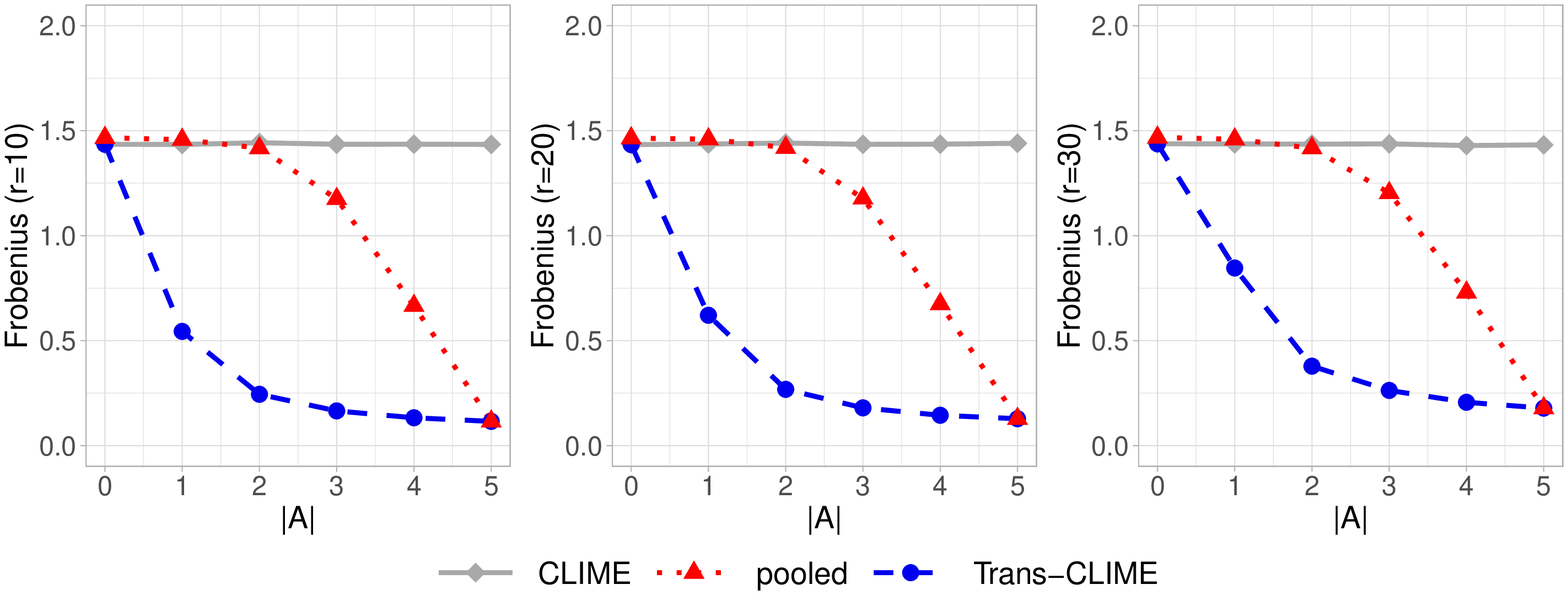}
\caption{Estimation errors in Frobenius norm for banded $\Omega$ (first row) and for block diagonal $\Omega$ (second row)  as a function of the number of informative studies (out of a total of  $K=5$ studies) for different values of $r$.}
\label{fig1}
\end{figure}

\subsection{Prediction errors}
We use the negative log-likelihood as the risk function for prediction. Specifically, we generate $x_i^{(test)}\sim N(0,\Sig)$ for $i=1,\dots,n_{\text{test}}=100$ and $x_i^{(test)}$ are independent of the samples for estimation. We evaluate the out-of-sample prediction error of an arbitrary graph estimator $\Omega^{(init)}$ in the following way. We symmetrize $\Omega^{(init)}$ and compute the positive definite projection of the symmetrized $\Omega^{(init)}$, denoted by $\Omega^{(init)}_{+}$. The prediction error of $\Omega^{(init)}$ is evaluated via
\begin{equation}
\label{eq-Qhat}
   \widehat{Q}(\Omega^{(init)})=\frac{1}{p}\left\{\frac{1}{2n_{\text{test}}}\sum_{i=1}^{n_{\text{test}}}\textup{Tr}(x_i^{(test)}(x_i^{(test)})^{\intercal}\Omega^{(init)}_+)-\frac{1}{2}\log det(\Omega^{(init)}_+)\right\}.
\end{equation}
In Figure \ref{fig3}, one can see that the prediction errors exhibit similar patterns as the estimation errors reported above.
\begin{figure}[H]
\includegraphics[width=0.99\textwidth, height=5.5cm]{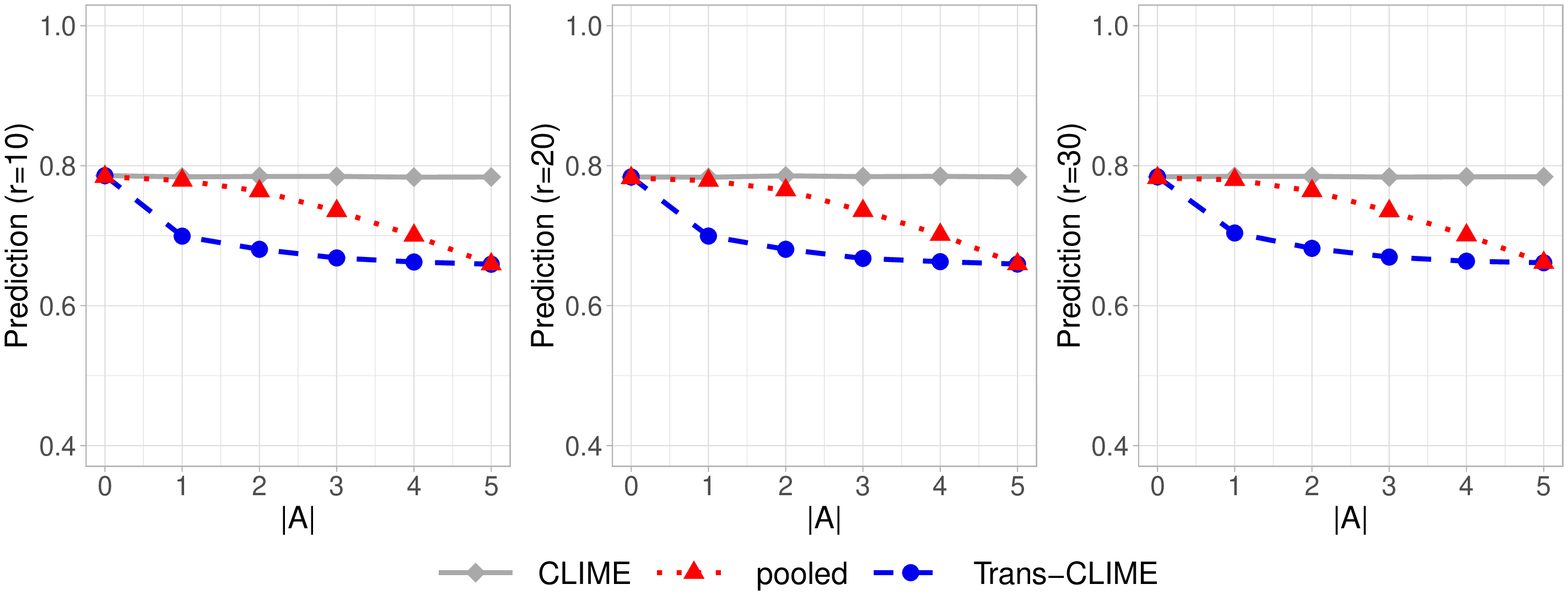}
\includegraphics[width=0.99\textwidth, height=5.5cm]{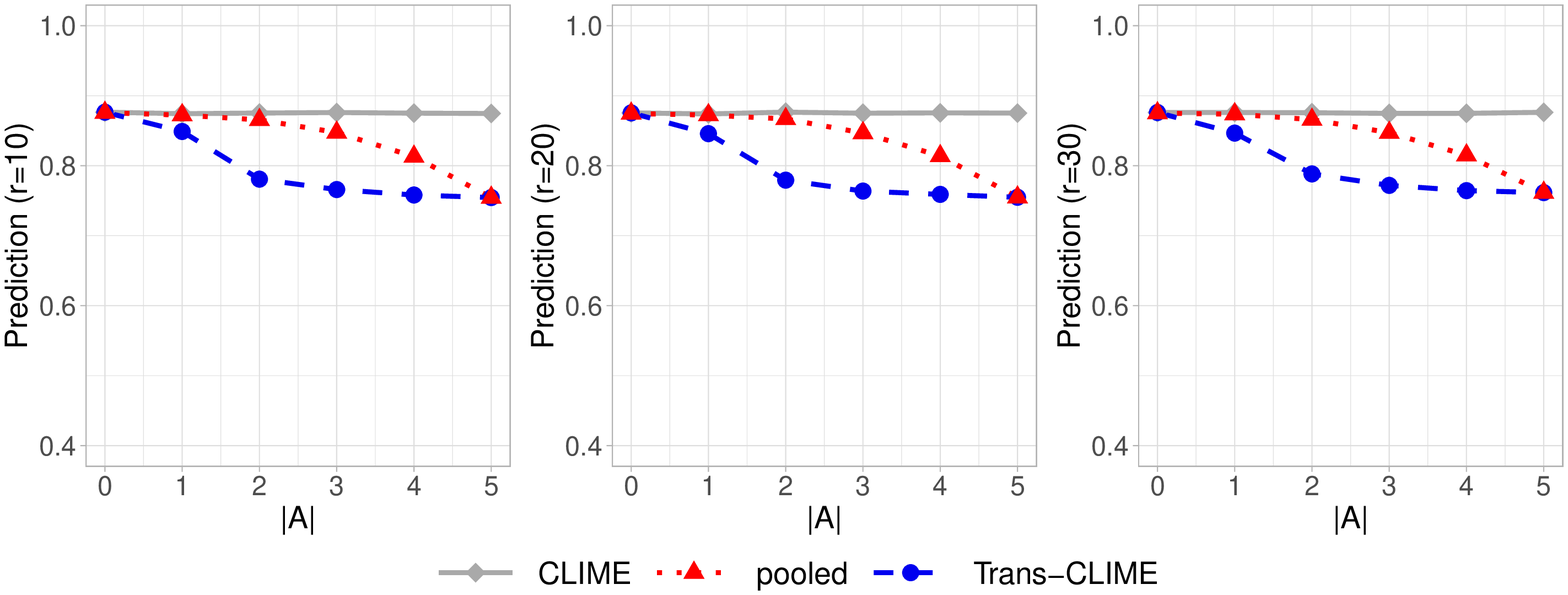}
\caption{Prediction errors with banded $\Omega$ (first row) and block diagonal $\Omega$ (second row) as a function of the number of informative studies (out of a total of  $K=5$ studies) for different values of $r$. }
\label{fig3}
\end{figure}

\subsection{FDR control}
We then consider FDR control at level $\alpha=0.1$  for the three methods introduced above. We still consider two types of target graphs defined in (i) and (ii).
From Figure \ref{fig-fdr1} and Figure \ref{fig-fdr2}, we see that all three methods have empirical FDR no larger than the nominal level. Specifically, the FDR of Trans-CLIME is closer to the nominal level. In terms of power, the Trans-CLIME has higher power when $\mA$ is nonempty.  We observe  the robustness of Trans-CLIME in the sense that the FDR is under control even if non-informative studies are included. However, the power can be lower than CLIME when some non-informative studies are included. We also  observe that the power for the banded $\Omega$ is much lower than the power for the block diagonal $\Omega$. This is because a proportion of entries  in the banded graph are weak, which are hard detect.
\begin{figure}[H]
\includegraphics[width=0.99\textwidth, height=5.5cm]{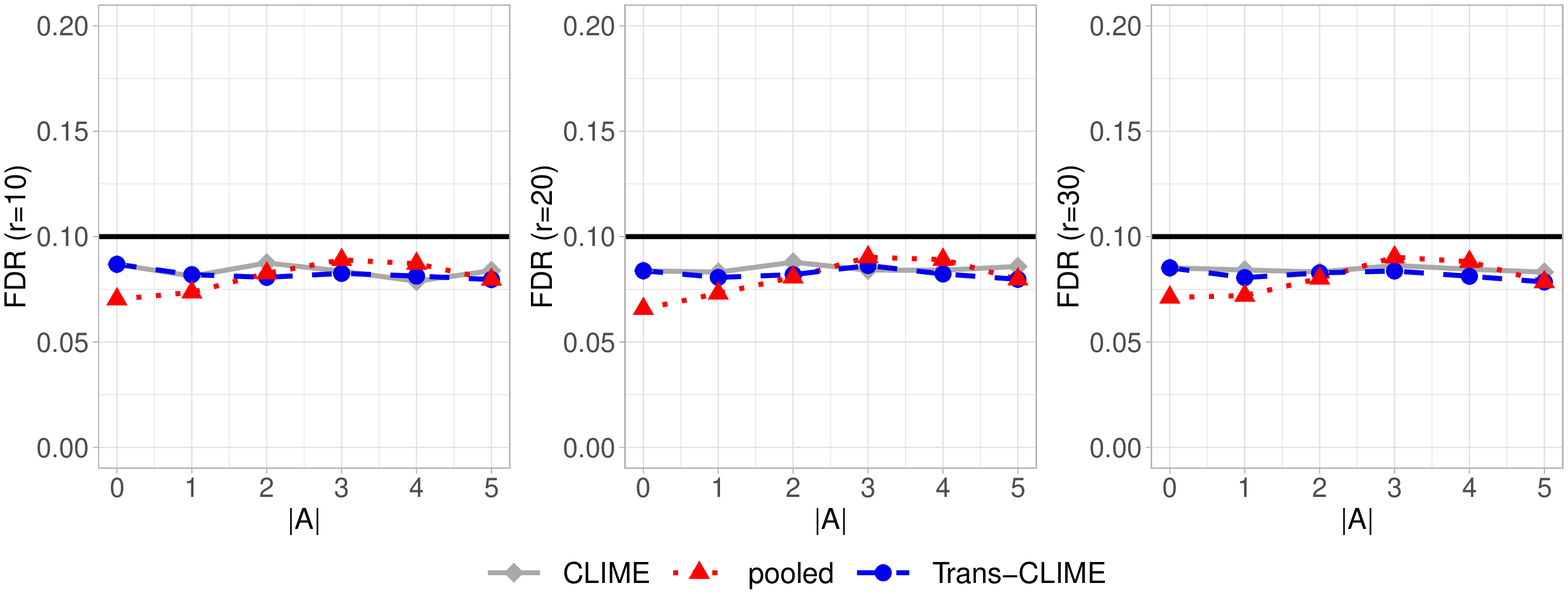}
\includegraphics[width=0.99\textwidth, height=5.5cm]{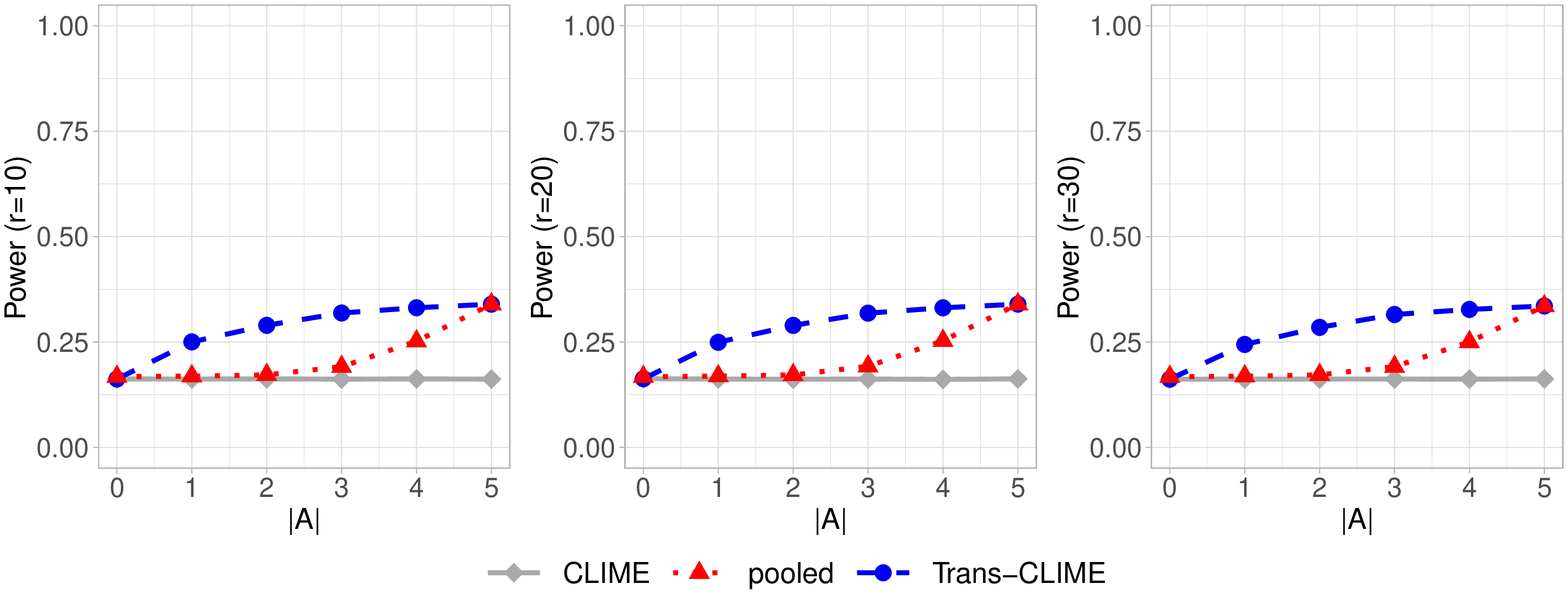}
\includegraphics[width=0.99\textwidth, height=5.5cm]{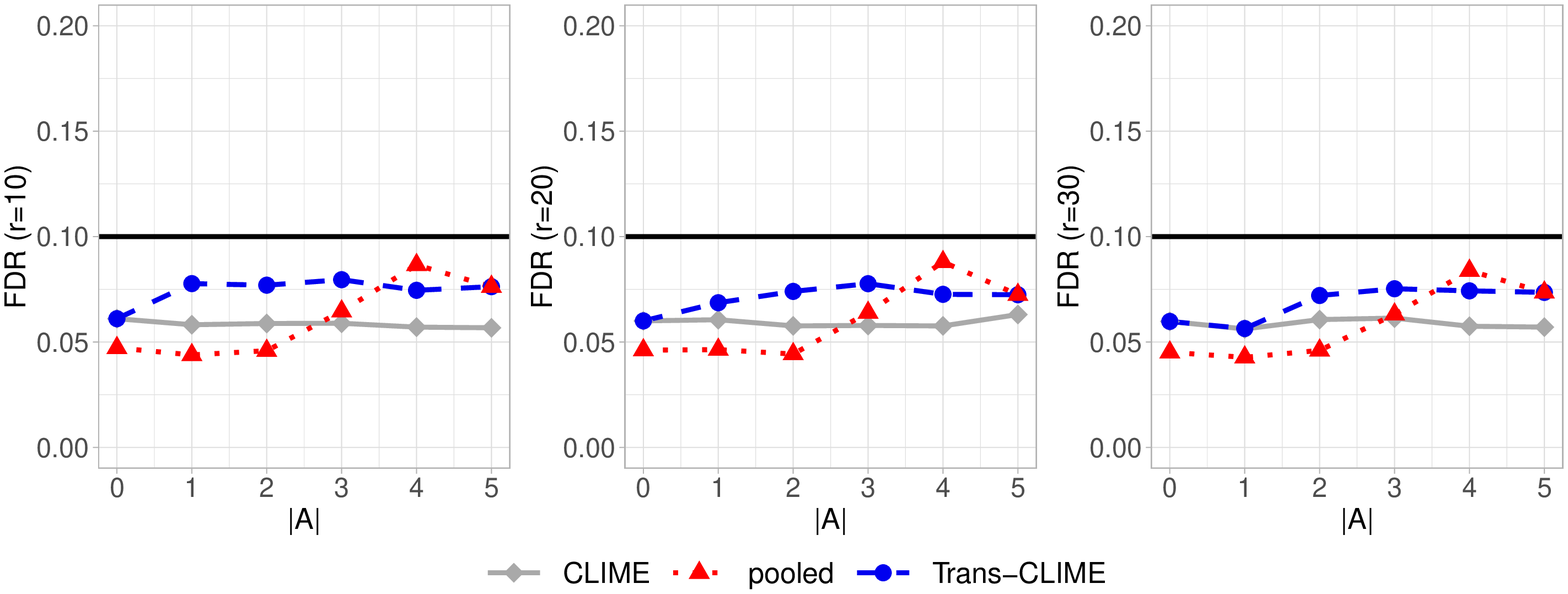}
\includegraphics[width=0.99\textwidth, height=5.5cm]{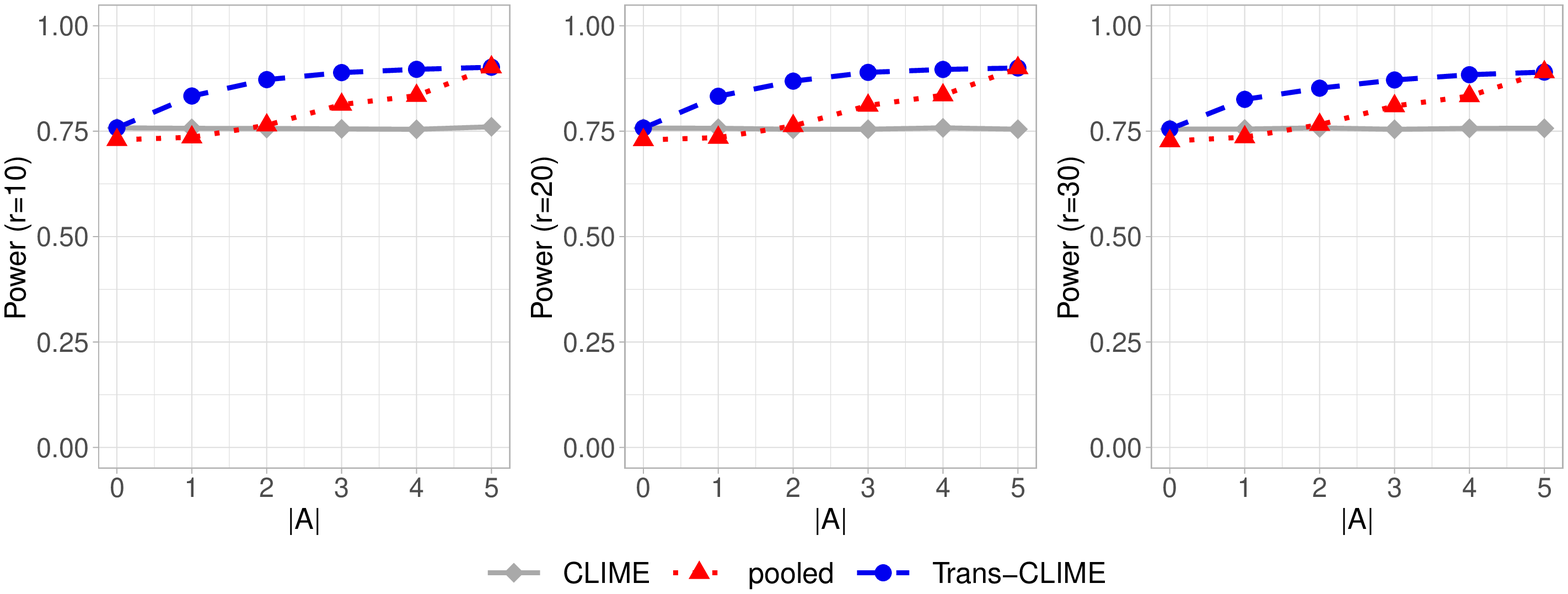}
\caption{The FDR and power with three methods at nominal level 0.1 as a function of the number of informative studies (out of $K=5$) and $r$ for banded $\Omega$ (first and second rows) and block diagonal $\Omega$ (third and fourth rows).}
\label{fig-fdr1}
\end{figure}
\ignore{
\begin{figure}
\includegraphics[width=0.99\textwidth, height=5.5cm]{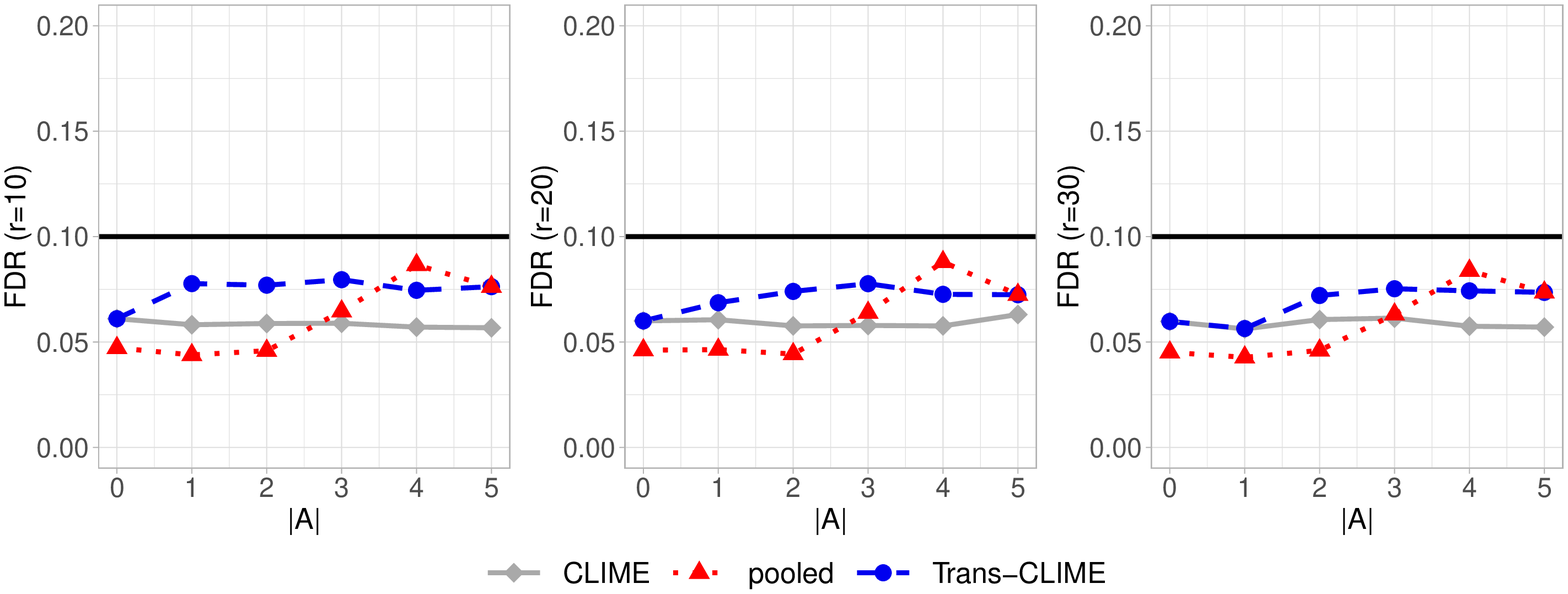}
\includegraphics[width=0.99\textwidth, height=5.5cm]{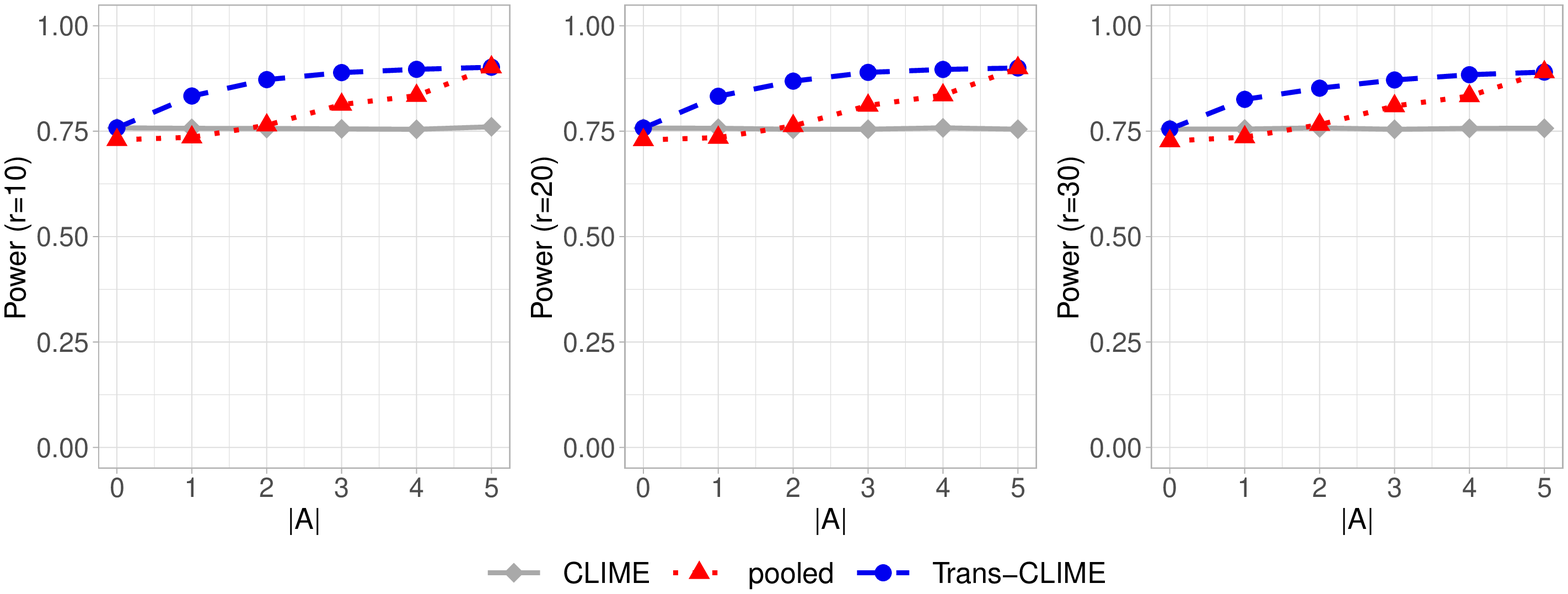}
\caption{The FDR (first row) and power (second row) with three methods at nominal level 0.1 as a function of the number of informative studies (out of $K=5$) and $r$ for block diagonal $\Omega$.}
\label{fig-fdr2}
\end{figure}}

\section{Gene networks detection in multiple tissues}
\label{sec-data}
In this section, we apply our proposed algorithms to detect gene networks in different tissues using the Genotype-Tissue Expression (GTEx) data (\url{https://gtexportal.org/}). Overall, the data sets measure  gene expression levels in  49 tissues from 838 human donors, comprising a total of  1,207,976 observations of 38,187 genes.  We focus on genes related to central nervous system neuron differentiation, annotated as \texttt{GO:0021953}. 
This gene set includes a total of 184 genes.  A complete list of the  genes can be found at \url{https://www.gsea-msigdb.org/gsea/msigdb/cards/GO_CENTRAL_NERVOUS_SYSTEM_NEURON_DIFFERENTIATION}. 

Our goal is to estimate and detect the gene network in a target brain tissue. Since we use 20\% of the samples to compute test errors, the sample size for the target tissue should not be too small. We therefore consider each brain tissue with at least 100 samples as the target tissue in each experiment.   We use the data from multiple other brain tissues as auxiliary samples with $K=12$. We remove the genes that  have missing values in these 13 tissues, resulting a total of 141 genes for the graph construction. The average sample size in each tissue is 115. A complete list of tissues and their sample sizes are given in the Supplementary Materials.

We apply CLIME and Trans-CLIME to estimate the Gaussian graph among these 141 genes in multiple target brain tissues. We first compare the prediction performance of CLIME and Trans-CLIME, where  we randomly split the samples of the target tissue into five folds. We fit the model with four folds of the samples and compute the prediction error with the rest of the samples. We report the mean of the prediction errors, each based on a different fold of the samples. The prediction errors are measured by the negative log-likelihood defined in (\ref{eq-Qhat}). 

The prediction results are reported in the left panel of Figure \ref{fig1-data}. We see that the prediction errors based on Trans-CLIME are significantly lower than those based on CLIME in many cases, indicating that the brain tissues in GTEx possess relatively high similarities in gene associations. On the other hand, these brain tissues are also heterogeneous in the sense that the improvements with transfer learning are significant in some tissues (e.g., A.C. cortex and F. cortex) and they are relatively mild in others (e.g., C. hemisphere and Cerebellum).

We then apply  Algorithm \ref{alg-fdr} with $\alpha=0.1$ to identify the connections among these genes.  The proportion of detected edges are reported in the right panel of Figure \ref{fig1-data}. We see that the percentages of detected edges are relatively low, implying that the networks are sparse. We see that Trans-CLIME has larger power than CLIME in almost all the tissues in detecting the gene-gene links,  agreeing with  our simulation results. In Figure \ref{fig2-data}, we evaluate the similarities among the tissues in terms of the degree distributions of the constructed graphs.  Specifically, we examine the degrees of nodes in A.C. cortex  in comparison to the degrees of nodes in the other nine tissues, all estimated using Trans-CLIME.  We see that the degree distribution in A.C. cortex is relatively similar to the degree distributions in Cortex, and F. cortex. 

In the Supplementary Material  (Section \ref{sec-hub}), we report the hubs detected by these two methods in different tissues  and observe  that many hubs appear more than once in different tissues based on the results of debiased Trans-CLIME, further demonstrating a certain level of similarity in gene regulatory networks among different brain  tissues. For example, for A.C Cortex and with Trans-CLIME, we are able to identify the hub genes SOX1, SHANK3, ATF5, and SEMA3A. These genes are either the known transcriptional factors (SOX1, ATF5) and have been shown to be related to neurological diseases, including the leading autism gene SHANK3 \citep{Lutzeaaz3267} and gene-related to motor neurons in ALS patients (Sema3A) \citep{sema3a}. In comparison, the graphs estimated using CLIME in single tissue are sparse and do not reveal any of these hub genes. 

 \begin{figure}[H]
 \includegraphics[height=5.7cm,width=0.49\textwidth]{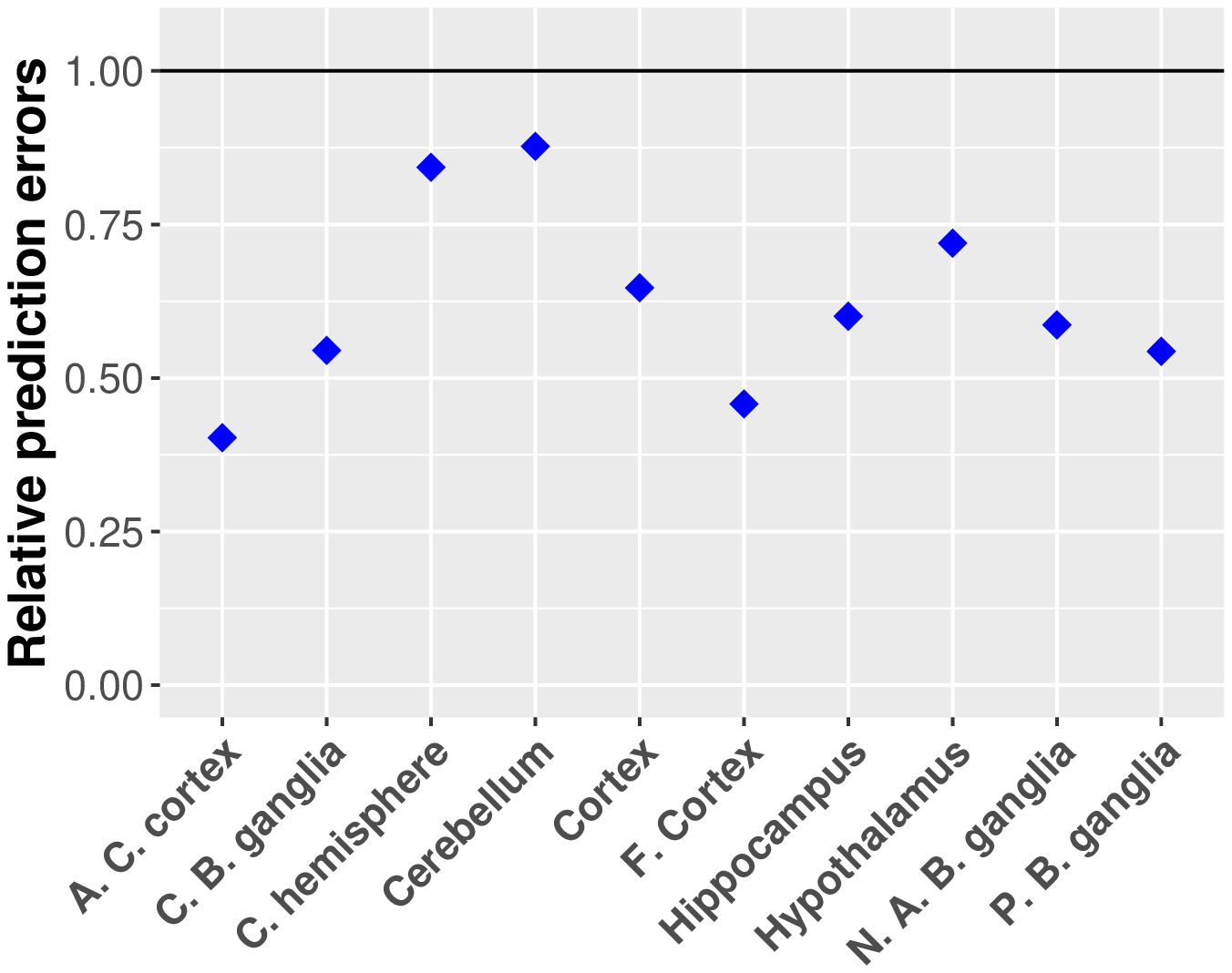}
   \includegraphics[width=0.495\textwidth, height=5.9cm]{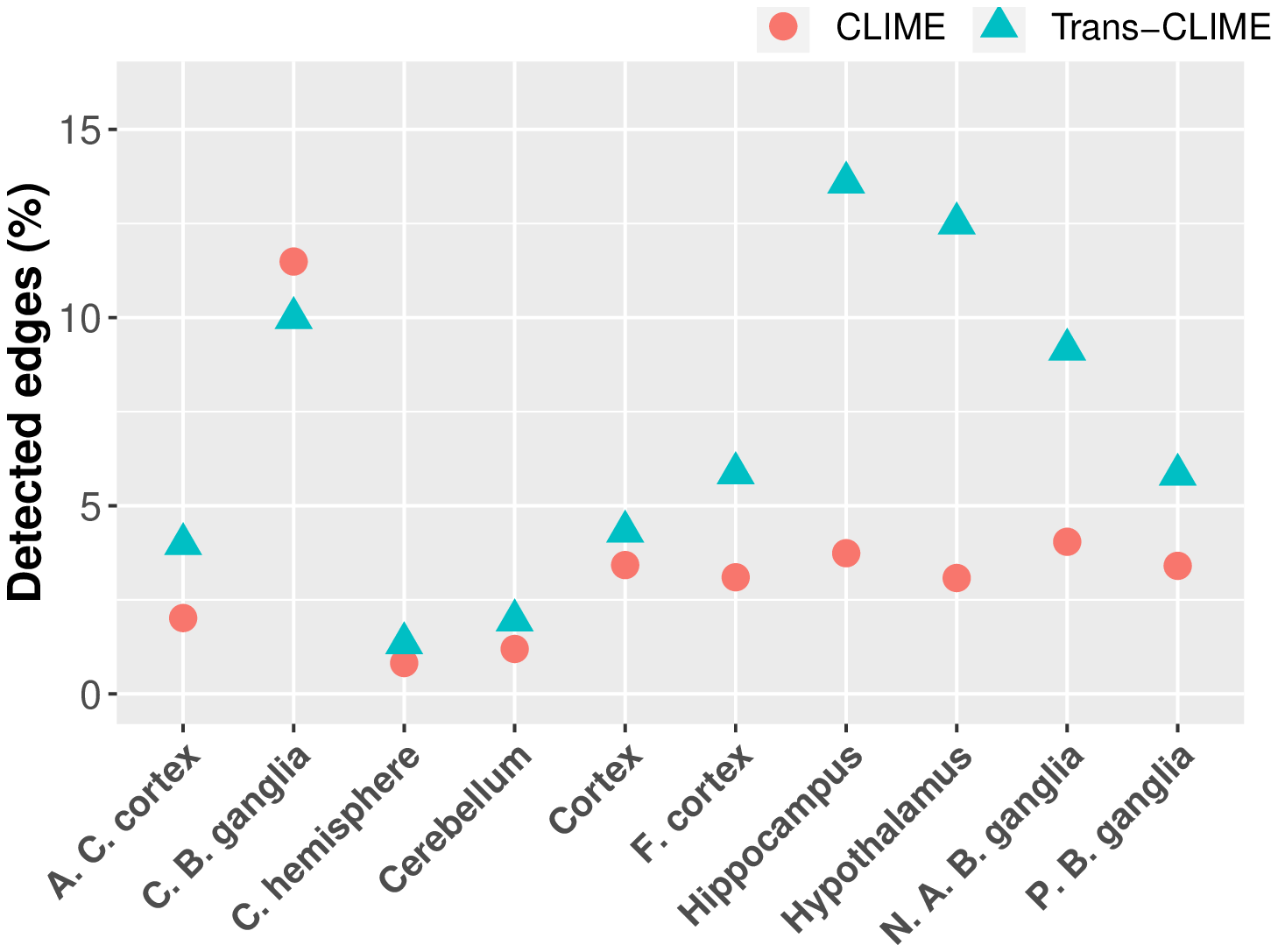}
 \centering
 \caption{Results of analysis of 10 different brain tissues. The left panel presents the prediction errors of Trans-CLIME relative to the prediction errors of CLIME for 10 different target tissues. The right panel presents the number of detected edges divided by $p(1-p)$ using CLIME and Trans-CLIME with FDR=0.1. The full names of the target tissues are given in the supplementary files.}
 \label{fig1-data}
 \end{figure}

  \begin{figure}[H]
  \includegraphics[width=0.995\textwidth, height=6.5cm]{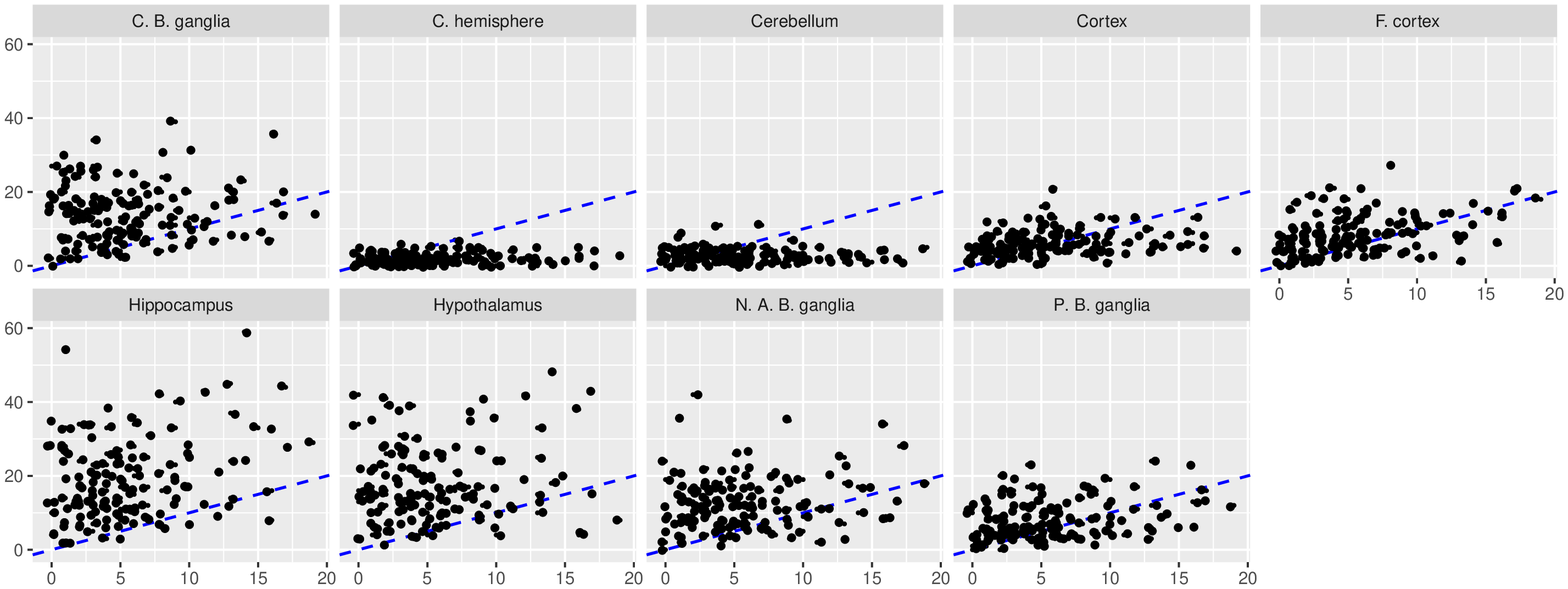}
 \caption{Comparison of the node degree distribution based on the graph estimated by Trans-CLIME for each of the tissue at FDR level of 10\%.  The $x$-axis represents the degrees of the nodes in A. C. cortex  and the $y$-axis represents the degrees of the nodes in nine other tissues. The dashed line is diagonal.}
  \label{fig2-data}
 \end{figure}
 
 \section{Discussion}
 \label{sec-diss}
 In this paper,  we have studied the estimation and inference of Gaussian graphical models with transfer learning. We assume the divergence matrices between the target graph and the informative auxiliary graphs are relatively sparse. Our proposed algorithm Trans-CLIME admits a faster convergence rate than the minimax rate in single study setting under mild conditions. The Trans-CLIME estimator can be further debiased for statistical inference.

 A practical challenge in transfer learning is to identify the informative auxiliary studies, i.e., the set $\mA_q$. While our proposal is guaranteed to be no worse than the single-study minimax estimator, it may not be the most efficient way to use the auxiliary studies. In the high-dimensional regression problem, \citet{LCL20} proposes to first rank all the auxiliary studies according to their similarities to the target and then perform a model selection type of model aggregation. They prove that the aggregated estimator can be adapted to $\mA_q$ under certain conditions. In a more recent paper, \citet{HK20} proves that, loosely speaking, if the ranks of the auxiliary studies can be recovered, then performing empirical risk minimization in a cross-fitting manner can achieve adaptation to $\mA_q$ to some extent in some functional classes. For the high-dimensional GMMs, heuristic rank estimators can also be derived using their connections to linear models, based on which one can perform aggregation towards an adaptive estimator. However, theoretical analysis for such rank estimators may require strong conditions, especially in the high-dimensional scenario. Adaptation to $\mA_q$ is an important topic for further studies.
  
  \section*{FUNDING}
  This research was supported by NIH grants R01GM123056 and R01GM129781.
  
  \bigskip
\begin{center}
{\large\bf SUPPLEMENTARY MATERIAL}
\end{center}
Supplement to ``Transfer Learning in Large-scale Gaussian Graphical Models with False Discovery Rate Control''.
In the Supplementary Materials, we provide the proofs of theorems and more results for data applications.
  
  {\small
\bibliographystyle{chicago}
  \bibliography{TL-GGM}
  }
  
  \appendix

Let $S_j$ denote the support of $\Omega_j$ for $j=1,\dots, p$. For an arbitrary matrix $A$, let $r_j(A)=(A_{j,.})^{\intercal}$. Let $\mathcal{B}_q(r)$ denote the $\ell_q$-ball centered at zero with radius $r$.

  \section{Proof of Theorem \ref{thm2}}
  To prove Theorem \ref{thm2}, we present two lemmas. In Lemma \ref{lem0-thm1}, we show the convergence rate of $\widehat{\Omega}^{(\textup{(CL)}}$. In Lemma \ref{thm1}, we show the $\widehat{\Theta}$ when $h\lesssim s\sqrt{\log p/n}$. 
  \subsection{Two useful lemmas}
  \begin{lemma}[Convergence rate of CLIME]
\label{lem0-thm1}
Under the conditions of Theorem \ref{thm2}, we have
\begin{align*}
&\P\left(\max_j\|\widehat{\Omega}^{(\textup{(CL)}}-\Omega\|_{\infty,2}^2\geq c_1s\frac{\log p}{n}\right)\leq \exp\{-c_2\log p\}+\exp\{-c_3n\}\\
&\P\left(\max_j\|\widehat{\Omega}^{(\textup{(CL)}}-\Omega\|_{\infty,1}\geq c_1s\sqrt{\frac{\log p}{n}}\right)\leq \exp\{-c_2\log p\}+\exp\{-c_3n\}
\end{align*}
for some positive constants $c_1$, $c_2$ and $c_3$.
\end{lemma}

\begin{lemma}[Convergence rate of $\widehat{\Theta}$ in Step 2]
\label{thm1}
Under the conditions of Theorem \ref{thm2} and $h\lesssim s\sqrt{\log p/n}$ and $n_{\mA}\gtrsim n$, we have for any true models in $\mathbb{G}_1(s,h)$ and fixed $1\leq j\leq p$,
\begin{align*}
  &\P\left(\|\widehat{\Theta}_j-\Omega_j\|_2^2\geq c_1\frac{s\log p}{n_{\mA}}+c_2h\delta_n\right)\leq \exp\{-c_3\log p\}+\exp\{-c_4n\}
\end{align*}
for some positive constants $c_1$, $c_2$, $c_3$ and $c_4$.
\end{lemma}

\subsection{Proof of Theorem \ref{thm2}}
\begin{proof}[Proof of Theorem \ref{thm2}]
We now prove the theoretical properties of aggregation in Step 3.
We will first show that for any fixed $1\leq j\leq p$,
\begin{align}
\label{eq-w3}
\P\left(\|\widehat{\Omega}_j-\Omega_j\|_2^2\geq  \frac{c_1t}{n}+ \frac{c_2s\log p}{n_{\mA}+n}+c_3h\delta_n\wedge \frac{s\log p}{n}\right)\leq \exp\{-c_4t\}+2\exp\{-c_5\log p\}.
\end{align}
Next, we will show that
\begin{align*}
\P\left(\frac{1}{p}\|\widehat{\Omega}-\widehat{\Theta}\|_F^2\geq \frac{c_1t}{n}+c_2\frac{s\log p}{n+n_{\mA}}+c_3h\delta_n\wedge \frac{s\log p}{n}\right)\leq \frac{c_4}{t}+\exp\{-c_5\log p\}.
\end{align*}

For any $v\in\R^2$,
\begin{align}
\label{eq3-pf3}
&\|(\widehat{\Omega}^{(\textup{CL})}_j,\widehat{\Theta}_j)\hat{v}_j-\Omega_j\|_2^2\leq  2\|(\widehat{\Omega}^{(\textup{CL})}_j,\widehat{\Theta}_j)(\hat{v}_j-v_j)\|_2^2+2\|(\widehat{\Omega}^{(\textup{CL})}_j,\widehat{\Theta}_j)v_j-\Omega_j\|_2^2.
\end{align}

When $h\leq s\sqrt{\log p/n}\leq c_1$ and $n\leq n_{\mA}$, we consider $v_j=(0,1)^{\intercal}$.
\begin{align*}
\hat{v}_j&=\{\widehat{W}(j)\}^{-1}\begin{pmatrix}
\widehat{\Omega}^{(\textup{CL})}_{j,j}-(\widehat{\Omega}^{(\textup{CL})}_j)^{\intercal}\widetilde{\Sig}\widehat{\Theta}_j\\
\widehat{\Theta}_{j,j}-\widehat{\Theta}_j^{\intercal}\widetilde{\Sig}\widehat{\Theta}_j
\end{pmatrix}+\{\widehat{W}(j)\}^{-1}\begin{pmatrix}
(\widehat{\Omega}^{(\textup{CL})}_j)^{\intercal}\widetilde{\Sig}\widehat{\Theta}_j\\
\widehat{\Theta}_j^{\intercal}\widetilde{\Sig}\widehat{\Theta}_j
\end{pmatrix}\\
&=\{\widehat{W}(j)\}^{-1}\begin{pmatrix}
(\widehat{\Omega}^{(\textup{CL})}_{j})^{\intercal}(e_j-\widetilde{\Sig}\Omega_{j})\\
\widehat{\Theta}_j^{\intercal}(e_j-\widetilde{\Sig}\Omega_{j})
\end{pmatrix}+
\{\widehat{W}(j)\}^{-1}\begin{pmatrix}
(\widehat{\Omega}^{(\textup{CL})}_j)^{\intercal}\widetilde{\Sig}(\widehat{\Theta}_j-\Omega_{j})\\
\widehat{\Theta}_j^{\intercal}\widetilde{\Sig}(\widehat{\Theta}_j-\Omega_{j})
\end{pmatrix}+v_j.
\end{align*}

Let $(\widehat{\Omega}^{(\textup{CL})}_j,\widehat{\Theta}_j)=U\Lambda V^{\intercal}$ denote its singular value decomposition for $U\in\R^{p\times 2}$, $\Lambda\in\R^{2\times 2}$. 
\begin{align*}
(\widehat{v}_j-v_j)\widehat{W}(j)(\hat{v}_j-v_j)&=(\hat{v}_j-v_j)^{\intercal}(\widehat{\Omega}^{(\textup{CL})}_j,\widehat{\Theta}_j)^{\intercal}(e_j-\widetilde{\Sig}\Omega_j)+(\hat{v}_j-v_j)^{\intercal}(\widehat{\Omega}^{(\textup{CL})}_j,\widehat{\Theta}_j)^{\intercal}\widetilde{\Sig}(\widehat{\Theta}_j-\Omega_j)\\
&=(\hat{v}_j-v_j)^{\intercal}V\Lambda U^{\intercal}(e_j-\widetilde{\Sig}\Omega_j)+(\hat{v}_j-v_j)^{\intercal}(\widehat{\Omega}^{(\textup{CL})}_j,\widehat{\Theta}_j)^{\intercal}\widetilde{\Sig}(\widehat{\Theta}_j-\Omega_j)\\
&\leq \|\Lambda V^{\intercal}(\hat{v}_j-v_j)\|_2\|\|U^{\intercal}(e_j-\widetilde{\Sig}\Omega_j)\|_2 +(\widehat{\Theta}_j-\Omega_j)^{\intercal}\widetilde{\Sig}(\widehat{\Theta}_j-\Omega_j)\\
&\quad +\frac{1}{2}(\widehat{v}_j-v_j)\widehat{W}(j)(\hat{v}_j-v_j),
\end{align*}
where the last step is by Young's inequality.
It implies that
\begin{align}
\frac{1}{2}(\widehat{v}_j-v_j)^{\intercal}\widehat{W}(j)(\hat{v}_j-v_j)\leq  \|V\Lambda(\hat{v}_j-v_j)\|_2\|U^{\intercal}(e_j-\widetilde{\Sig}\Omega_j)\|_2 +(\widehat{\Theta}_j-\Omega_j)^{\intercal}\widetilde{\Sig}(\widehat{\Theta}_j-\Omega_j).\label{eq1-pf3}
\end{align}
For the left hand side,
\begin{align*}
(\widehat{v}_j-v_j)^{\intercal}\widehat{W}(j)(\hat{v}_j-v_j)&=\langle (\widehat{\Omega}^{(\textup{CL})}_j,\widehat{\Theta}_j)(\hat{v}_j-v_j),\widetilde{\Sig}(\widehat{\Omega}^{(\textup{CL})}_j,\widehat{\Theta}_j)(\hat{v}_j-v_j)\rangle.
\end{align*}
We have,
\begin{align*}
   (\widehat{v}_j-v_j)\widehat{W}(j)(\hat{v}_j-v_j)&\geq \|\Lambda V^{\intercal}(\hat{v}_j-v_j)\|_2^2\Lambda_{\min}(U^{\intercal}\widetilde{\Sig}U).
\end{align*}
Since $U$ is independent of $\widetilde{\Sig}$, it is easy to show that
\[
   \Lambda_{\min}(U^{\intercal}\widetilde{\Sig}U)\geq  \Lambda_{\min}(U^{\intercal}\Sig U)-O_P(n^{-1/2})\geq \Lambda_{\min}(\Sig)-O_P(n^{-1/2}).
\]
We arrive at
\begin{align*}
 (\widehat{v}_j-v_j)\widehat{W}(j)(\hat{v}_j-v_j)\geq \|\Lambda V^{\intercal}(\hat{v}_j-v_j)\|_2^2(\Lambda_{\min}(\Sig)-O_P(n^{-1/2})).
\end{align*}
For the right hand side of (\ref{eq1-pf3}),
\begin{align*}
&\P\left(\|U^{\intercal}(e_j-\widetilde{\Sig}\Omega_j)\|_2\geq c_1t\right)\leq \exp(-c_2nt^2)\\
&\P\left((\widehat{\Theta}_j-\Omega_j)^{\intercal}\widetilde{\Sig}(\widehat{\Theta}_j-\Omega_j)\geq (\widehat{\Theta}_j-\Omega)^{\intercal}\Sig(\widehat{\Theta}_j-\Omega_j)(1+t)\right)\leq \exp(-nt^2).
\end{align*}
where the first line is due to $U$ is independent of $\widetilde{\Sig}$ and $ \|U_j\|_2=1$.
To summarize, for $v_j=(0,1)^{\intercal}$,
\begin{align}
\label{eq-v1}
\P\left(\|\Lambda V^{\intercal}(\hat{v}_j-v_j)\|_2^2\geq c_1\frac{t}{n}+c_2\|\widehat{\Theta}_j-\Omega_j\|_{2}^2\right)\leq \exp\{-c_3t\}+\exp\{-c_4\log p\}.
\end{align}
Notice that
\[
  \|\Lambda V^{\intercal}(\hat{v}_j-v_j)\|_2^2=\|(\widehat{\Omega}^{(\textup{CL})}_j,\widehat{\Theta}_j)(\hat{v}_j-v_j)\|_2^2.
\]
Invoking  that (\ref{eq3-pf3}), for $v=(0,1)^{\intercal}$,
\begin{align}
\label{eq-w1}
\P\left(\|\widehat{\Omega}_j-\Omega_j\|_2^2\geq c_1\frac{t}{n}+c_2\frac{s\log p}{n_{\mA}+n}+h\delta_n\right)\leq \exp\{-c_3t\}+\exp\{-c_4\log p\}+\exp\{-c_5n\},
\end{align}
where $\frac{s\log p}{n_{\mA}+n}+h\delta_n=\frac{s\log p}{n_{\mA}+n}+h\delta_n\wedge s\log p/n$ in the current scenario.

If  $h\geq s\sqrt{\log p/n}$ or $n>n_{\mA}$, we consider $v=(1,0)^{\intercal}$. Repeating above arguments, we have
\begin{align}
&\P\left(\|\Lambda V^{\intercal}(\hat{v}_j-v_j)\|_2^2\geq \frac{t}{n}+\|\widehat{\Omega}^{(\textup{CL})}-\Omega\|_{\infty,2}^2\right)\exp\{-c_3t\}+\exp\{-c_4\log p\}.\label{eq-v2}\\
&\P\left(\|\widehat{\Omega}_j-\Omega_j\|_2^2\geq \frac{t}{n}+\frac{s\log p}{n}\right)\leq \exp\{-c_3t\}+\exp\{-c_4\log p\}+\exp\{-c_5n\},\label{eq-w2}
\end{align}
where $s\log p/n\asymp \frac{s\log p}{n_{\mA}+n}+h\delta_n\wedge s\log p/n$ in the current scenario.
We can conclude (\ref{eq-w3}) from (\ref{eq-w1}) and (\ref{eq-w2}).

 Next, we establish the upper bound under Frobenius norm. We only prove for $v_j=(0,1)$, $1\leq j\leq p$.

It follows from (\ref{eq1-pf3}) that
\[
\sum_{j=1}^p\frac{1}{4}\|V_j\Lambda_j(\hat{v}_j-v_j)\|_2^2\leq  \sum_{j=1}^p\|U_j^{\intercal}(e_j-\widetilde{\Sig}\Omega_j)\|^2_2 +\sum_{j=1}^p(\widehat{\Theta}_j-\Omega_j)^{\intercal}\widetilde{\Sig}(\widehat{\Theta}_j-\Omega_j).
\]
That is,
\begin{align}
\frac{1}{4}\|\widehat{\Omega}-\widehat{\Theta}\|_F^2\leq  \sum_{j=1}^p\|U_j^{\intercal}(e_j-\widetilde{\Sig}\Omega_j)\|^2_2 +\|\widetilde{\Sig}^{1/2}(\widehat{\Theta}-\Omega)\|_F^2.\label{eq1-pf3}
\end{align}
For the second term on the RHS of (\ref{eq1-pf3}),
\begin{align*}
\|\widetilde{\Sig}^{1/2}(\widehat{\Theta}-\Omega)\|_F^2&\leq \|\Sig^{1/2}(\widehat{\Theta}-\Omega)\|_F^2+p\max_{1\leq j\leq p} |(\widehat{\Theta}_j-\Omega_j)^{\intercal}(\widetilde{\Sig}-\Sig)(\widehat{\Theta}_j-\Omega_j)|\\
&\leq Cp\|\widehat{\Theta}-\Omega\|_{\infty,2}^2+p\max_{1\leq j\leq p} |(\widehat{\Theta}_j-\Omega_j)^{\intercal}(\widetilde{\Sig}-\Sig)(\widehat{\Theta}_j-\Omega_j)|,
\end{align*}
where
\begin{align*}
\P\left(\max_{1\leq j\leq p} |(\widehat{\Theta}_j-\Omega_j)^{\intercal}(\widetilde{\Sig}-\Sig)(\widehat{\Theta}_j-\Omega_j)| \geq t\|\widehat{\Theta}-\Omega\|_{\infty,2}^2\right)\leq p\exp\{-c_2nt^2\}.
\end{align*}
Hence, for large enough constant $c_1$
\begin{align*}
\P\left(\frac{1}{p}\|\widetilde{\Sig}^{1/2}(\widehat{\Theta}-\Omega)\|_F^2\geq c_1\|\widehat{\Theta}-\Omega\|_{\infty,2}^2\right)\leq \exp\{-c_2\log p\}.
\end{align*}
For the first term on the RHS of (\ref{eq1-pf3}),
we know that 
\begin{align*}
\E[\sum_{j=1}^p\|U_j^{\intercal}(e_j-\widetilde{\Sig}\Omega_j)\|_2^2]=Cp/n,
\end{align*}
for some constant $C>0$. By Markov's inequality,
\begin{align*}
\P\left(\sum_{j=1}^p\|U_j^{\intercal}(e_j-\widetilde{\Sig}\Omega_j)\|_2^2\geq \frac{(C+t)p}{n}\right)\leq \frac{C}{t}.
\end{align*}
We arrive at
\begin{align*}
\P\left(\frac{1}{p}\|\widehat{\Omega}-\widehat{\Theta}\|_F^2\geq c_1\frac{t}{n}+c_2\frac{s\log p}{n+n_{\mA}}+c_3h\delta_n\right)\leq \frac{c_4}{t}+\exp\{-c_5\log p\}.
\end{align*}
Applying the above arguments for $v=(1,0)^{\intercal}$, one can obtain the desired bound under Frobenius norm.
\end{proof}

\subsection{Proof of Lemma \ref{lem0-thm1} and Lemma \ref{thm1}}

\begin{proof}[Proof of Lemma \ref{lem0-thm1}]
By Lemma 1 in \citet{Cai11}, for $1\leq j\leq p$,
\begin{align*}
   \widehat{\Omega}^{(\textup{CL})}_{j}&=\argmin_{\omega}\|\omega\|_1 \\
   &\text{subject to}~~ \|\widehat{\Sig}\omega-e_j\|_{\infty}\leq \lam_{\textup{CL}}.
\end{align*}
For $\lam_{\textup{CL}}\geq c\sqrt{\log p/n}$ with large enough constant $c$, $\Omega_j$ is a feasible solution to the above optimization. Hence,
\[
   (  \widehat{\Omega}^{(\textup{CL})}_{j}-\Omega_j)^{\intercal}\widehat{\Sig}( \widehat{\Omega}^{(\textup{CL})}_{j}-\Omega_j)\leq 2\|  \widehat{\Omega}^{(\textup{CL})}_{j}-\Omega_j\|_1\lam_{\textup{CL}}.
\]
Moreover, $\|\widehat{\Omega}^{(\textup{CL})}_{j}\|_1\leq \|\Omega_j\|_1$. Using the sparsity of $\Omega_j$, standard analysis lead to the desired results assuming $s\log p=o(n)$. Notice that Theorem \ref{thm2} assumes $s\sqrt{\log p}=O(\sqrt{n})$, hence $s\log p=o(n)$ is guaranteed.
\end{proof}

\begin{lemma}
\label{lem1-thm1}
Under the conditions of Theorem \ref{thm2}, when $h\lesssim s\sqrt{\log p/n}$ and $n_{\mA}\gtrsim n$, we have
\[
\P\left(\max_j\|r_j(\widehat{\Delta}^{\mA})-r_j(\Delta^{\mA})\|_2^2\geq c_1h\delta_n\right)\leq \exp(-c_2\log p)+\exp(-c_3n)
\]
for some positive constants $c_1$, $c_2$ and $c_3$.
\end{lemma}
\begin{proof}[Proof of Lemma \ref{lem1-thm1}]
We first show that $\Delta^{\mA}$ is a feasible solution to (\ref{eq-Delta-init}).
\begin{align*}
&\|\widehat{\Sig}\Delta^{\mA}-(\widehat{\Sig}^{\mA}-\widehat{\Sig})\|_{\infty,\infty}\\
&=\|(\widehat{\Sig}-\Sig)\Delta^{\mA}\|_{\infty,\infty}+\|\Sig\Delta^{\mA}-(\widehat{\Sig}^{\mA}-\widehat{\Sig})\|_{\infty,\infty},
\end{align*}
where
\begin{align*}
&\P(\|(\widehat{\Sig}-\Sig)\Delta^{\mA}\|_{\infty,\infty}>t)\leq p^2\exp\left\{-\frac{c_1nt^2}{\max_{j,k}\Sig_{k,k}(\Delta^{\mA}_{.,j})^{\intercal}\Sig\Delta^{\mA}_{.,j}}\right\}.\\
&\P(|\Sig\Delta^{\mA}-(\widehat{\Sig}^{\mA}-\widehat{\Sig})\|_{\infty,\infty}\geq t)\leq p^2\exp\left\{-\frac{c_2nt^2}{\max_{k}\Sig_{k,k}}\right\}.
\end{align*}
Notice that
\begin{align*}
\max_{j,k}\Sig_{k,k}(\Delta^{\mA}_{.,j})^{\intercal}\Sig\Delta^{\mA}_{.,j}=\max_{j\leq p} (\Sig^{\mA}-\Sig)_{.,j}^{\intercal}\Omega(\Sig^{\mA}-\Sig)_{.,j}\leq  C^3.
\end{align*}
Hence, for $\lam_{\Delta}\geq c_1\sqrt{\log p/n}$, $\Delta^{\mA}$ is feasible and
\[
   \|\widehat{\Delta}^{(init)}_{j}\|_1\leq \|\Delta^{\mA}_{j}\|_1\leq h.
\]
Therefore,
\[
\P\left(\max_j\|\widehat{\Delta}_j^{(init)}-\Delta_j^{\mA}\|_1\leq 2h\right)\geq 1-\exp(-c_1n)-\exp(-c_2\log p).
\]
\underline{Step (ii)}
We now provide a sup-norm bound on the error of 
\[
   \widehat{\Delta}^{(db)}=\widehat{\Delta}^{(init)}+\widehat{\Omega}^{(\textup{CL})}(\widehat{\Sig}^{\mA}-\widehat{\Sig}-\widehat{\Sig}\widehat{\Delta}^{(init)}).
\]
This is essentially a debiased estimator of $\Delta^{\mA}$.
\begin{align}
 \widehat{\Delta}^{(db)}- \Delta^{\mA}&=\widehat{\Delta}^{(init)}-\Delta^{\mA}+\widehat{\Omega}^{(\textup{CL})}(\widehat{\Sig}^{\mA}-\widehat{\Sig}-\widehat{\Sig}\Delta^{\mA})-\widehat{\Omega}^{(\textup{CL})}\widehat{\Sig}(\widehat{\Delta}^{(init)}-\Delta^{\mA})\nonumber\\
 &=(I_p-\widehat{\Omega}^{(\textup{CL})}\widehat{\Sig})(\widehat{\Delta}^{(init)}-\Delta^{\mA})+\widehat{\Omega}^{(\textup{CL})}(\widehat{\Sig}^{\mA}-\Sig^{\mA}-(\widehat{\Sig}-\Sig)(I_p+\Delta^{\mA}))\nonumber\\
 &=\underbrace{(I_p-\widehat{\Omega}^{(\textup{CL})}\widehat{\Sig})(\widehat{\Delta}^{(init)}-\Delta^{\mA})}_{rem_1}+\underbrace{\Omega(\widehat{\Sig}^{\mA}-\Sig^{\mA}-(\widehat{\Sig}-\Sig)(I_p+\Delta^{\mA}))}_{rem_2}\nonumber\\
 &\quad \quad + \underbrace{(\widehat{\Omega}^{(\textup{CL})}-\Omega)(\widehat{\Sig}^{\mA}-\Sig^{\mA}-(\widehat{\Sig}-\Sig)(I_p+\Delta^{\mA}))}_{rem_3},\label{eq-rem}
\end{align}
where
\begin{align*}
&\|rem_1\|_{\infty,\infty}\leq  \|I_p-\widehat{\Omega}^{(\textup{CL})}\widehat{\Sig}\|_{\infty}\|\max_{j\leq p}\|\widehat{\Delta}^{(init)}_j-\Delta^{\mA}_j\|_1\leq h\lam_{\textup{CL}}\\
&\|rem_2\|_{\infty,\infty}\leq C\sqrt{\frac{\log p}{n}}\\
&\|rem_3\|_{\infty,\infty}\leq \|\widehat{\Omega}^{(\textup{CL})}-\Omega\|_{\infty,1}\|\widehat{\Sig}^{\mA}-\Sig^{\mA}-(\widehat{\Sig}-\Sig)(I_p+\Delta^{\mA})\|_{\infty,\infty}\leq C\frac{s\log p}{n}.
\end{align*}
for a large enough constant $C$ with probability at least $1-\exp(-c_1\log p)-\exp(-c_2n)$. The last line follows from Lemma \ref{lem0-thm1} and $n_{\mA}\gtrsim n$.
To summarize, we have with probability at least $1-\exp(-c_1\log p)-\exp(-c_2n)$,
\begin{align*}
\left\| \widehat{\Delta}^{(db)}- \Delta^{\mA}\right\|_{\infty,\infty}\leq c_1\sqrt{\frac{\log p}{n}}+c_2\frac{s\log p}{n}.
\end{align*}
Under the assumption that $s\sqrt{\log p}\lesssim n$, we have
\begin{align}
\label{eq0-pf1}
\left\| \widehat{\Delta}^{(db)}- \Delta^{\mA}\right\|_{\infty,\infty}\leq c_3\sqrt{\frac{\log p}{n}}
\end{align}
with probability at least $1-\exp(-c_1\log p)-\exp(-c_2n)$.

\underline{Step (iii)}. We are left to analyze (\ref{eq-DeltaA-est}). It is easy to show that $\Delta^{\mA}$ is a feasible solution to (\ref{eq-DeltaA-est}) with probability at least $1-\exp(-c_1\log p)-\exp(-c_2n)$. Using the fact that
\[
   \|r_j(\widehat{\Delta}^{\mA})\|_1\leq \|r_j(\Delta^{\mA})\|_1\leq h,
\]
we have
\begin{align*}
\max_{j\leq p}\|r_j(\widehat{\Delta}^{\mA}-\Delta^{\mA})\|_2^2\leq \max_{j\leq p}\|r_j(\widehat{\Delta}^{\mA}-\Delta^{\mA})\|_1\|r_j(\widehat{\Delta}^{\mA}-\Delta^{\mA})\|_{\infty}\leq Ch\delta_n
\end{align*}
for $\delta_n=\sqrt{\frac{\log p}{n}} \wedge h$ with probability at least $1-\exp(-c_1\log p)-\exp(-c_2n)$.
\end{proof}

\begin{proof}[Proof of Lemma \ref{thm1}]
By Lemma 1 in \citet{Cai11}, for $1\leq j\leq p$,
\begin{align}
   \widehat{\Theta}_{j}&=\argmin_{\omega}\|\omega\|_1 \label{opt-j}\\
   &\text{subject to}~~ \|\widehat{\Sig}^{\mA}\omega-(e_j+r_j(\widehat{\Delta}^{\mA})) \|_{\infty}\leq \lam_{\Theta}.\nonumber
\end{align}
We consider the Lasso version for the $j$-th column
\[
   \widehat{\Theta}^L_{j}=\argmin \frac{1}{2}\omega^{\intercal}\widehat{\Sig}^{\mA}\omega-\omega^{\intercal}(e_j+r_j(\widehat{\Delta}^{\mA}))+ \lam_{\Theta}\|\omega\|_1.
\]
Let $\Theta^*=\Omega+\Omega^{\mA}(\widehat{\Delta}^{\mA}-\Delta^{\mA})^{\intercal}$. The main idea of this proof is that we view $\Theta^*_j$ as the true parameter and view $\Omega_j$ as a spare approximation of $\Theta^*_j$. The arguments are analogous to Theorem 6.1 of \citet{BRT09}. We cannot directly use their arguments because the loss function here is in a different format.
Oracle inequality:
\begin{align*}
&\frac{1}{2}\langle \widehat{\Theta}^L_{j},\widehat{\Sig}^{\mA}\widehat{\Theta}^L_{j}\rangle-(\widehat{\Theta}^L_{j})^{\intercal}(e_j +r_j(\widehat{\Delta}^{\mA}))+ \lam_{\Theta}\|\widehat{\Theta}^L_{j}\|_1\\
&\quad\leq  \frac{1}{2}\langle \Omega_{j},\widehat{\Sig}^{\mA}\Omega_{j}\rangle-\Omega_{j}^{\intercal}(e_j +r_j(\widehat{\Delta}^{\mA}))+ \lam_{\Theta}\|\Omega_{j}\|_1\\
\implies& \frac{1}{2} \langle \widehat{\Theta}^L_{j}-\Theta^*_{j},\widehat{\Sig}^{\mA}(\widehat{\Theta}^L_{j}-\Theta^*_{j})\rangle\leq \frac{1}{2} \langle \Omega_{j}-\Theta^*_{j},\widehat{\Sig}^{\mA}(\Omega_{j}-\Theta^*_{j})\rangle\\
&\quad +|\langle \widehat{\Theta}^L_{j}-\Omega_{j},\widehat{\Sig}^{\mA}\Theta^*_j-(e_j +r_j(\widehat{\Delta}^{\mA}))\rangle|+ \lam_{\Theta}\|\Omega_{j}\|_1 - \lam_{\Theta}\|\widehat{\Theta}^L_{j}\|_1.
\end{align*}
In the event that
\begin{align*}
     \mathcal{E}_1&=\left\{  \|\widehat{\Sig}^{\mA}\Theta^*_{j}-(e_j -r_j(\widehat{\Delta}^{\mA}))\|_{\infty}\leq \lam_{\Theta}/2, \inf_{\|u_{S_j^c}\|_1\leq 6\|u_{S_j}\|_1\neq 0}\frac{u^{\intercal}\widehat{\Sig}^{\mA}u}{\|u\|_2^2}\geq \phi_0>0,\right.\\
    &\quad\quad \left. \langle \Omega_{j}-\Theta^*_{j},\widehat{\Sig}^{\mA}(\Omega_{j}-\Theta^*_{j})\rangle\leq C<\infty\right\},
\end{align*}
we have
\begin{align*}
\frac{1}{2} \langle \widehat{\Theta}^L_{j}-\Theta^*_{j},\widehat{\Sig}^{\mA}(\widehat{\Theta}^L_{j}-\Theta^*_{j})\rangle&\leq \frac{1}{2} \langle \Omega_{j}-\Theta^*_{j},\widehat{\Sig}^{\mA}(\Omega_{j}-\Theta^*_{j})\rangle \\
&\quad+\frac{\lam_{\Theta}}{2}\| \widehat{\Theta}^L_{j}-\Omega_{j}\|_1+ \lam_{\Theta}\|\Omega_{j}\|_1 - \lam_{\Theta}\|\widehat{\Theta}^L_{j}\|_1\\
&\leq  \frac{1}{2} \langle \Omega_{j}-\Theta^*_{j},\widehat{\Sig}^{\mA}(\Omega_{j}-\Theta^*_{j})\rangle \\
&\quad+\frac{3\lam_{\Theta}}{2}\| \widehat{\Theta}^L_{S_j,j}-\Omega_{S_j,j}\|_1- \frac{\lam_{\Theta}}{2}\|\widehat{\Theta}^L_{S_j^c,j}-\Omega_{S_j^c,j}\|_1
\end{align*}
The left hand side can be lower bounded by
\[
\frac{1}{2} \langle \widehat{\Theta}^L_{j}-\Omega_{j},\widehat{\Sig}^{\mA}(\widehat{\Theta}^L_{j}-\Omega_{j})\rangle-\frac{1}{4} \langle \Omega_{j}-\Theta^*_{j},\widehat{\Sig}^{\mA}(\Omega_{j}-\Theta^*_{j})\rangle.
\]
As a result,
\begin{align*}
\frac{1}{2} \langle \widehat{\Theta}^L_{j}-\Omega_{j},\widehat{\Sig}^{\mA}(\widehat{\Theta}^L_{j}-\Omega_{j})\rangle
&\leq  \frac{3}{4} \langle \Omega_{j}-\Theta^*_{j},\widehat{\Sig}^{\mA}(\Omega_{j}-\Theta^*_{.,j})\rangle \\
&\quad+\frac{3\lam_{\Theta}}{2}\| \widehat{\Theta}^L_{S_j,j}-\Omega_{S_j,j}\|_1- \frac{\lam_{\Theta}}{2}\|\widehat{\Theta}^L_{S_j^c,j}-\Omega_{S_j^c,j}\|_1
\end{align*}
(i) If 
\[
  \frac{3\lam_{\Theta}}{2}\| \widehat{\Theta}^L_{S_j,j}-\Omega_{S_j,j}\|_1\geq \frac{3}{4} \langle \Omega_{j}-\Theta^*_{j},\widehat{\Sig}^{\mA}(\Omega_{j}-\Theta^*_{j})\rangle ,
\]
then
\[
 \|\widehat{\Theta}^L_{S_j^c,j}-\Omega_{S_j^c,j}\|_1\leq 6\| \widehat{\Theta}^L_{S_j,j}-\Omega_{S_j,j}\|_1
\]
 and we use the RE condition in event $\mathcal{E}_1$ to arrive at
\begin{align*}
    \frac{\phi_0}{2} \| \widehat{\Theta}^L_{.,j}-\Omega_{j}\|_2^2&\leq 3\lam_{\Theta}\| \widehat{\Theta}^L_{S_j,j}-\Omega_{S_j,j}\|_1\leq 3\sqrt{s_j}\lam_{\Theta}\| \widehat{\Theta}^L_{S_j,j}-\Omega_{S_j,j}\|_2\\
    &\leq 3\sqrt{s_j}\lam_{\Theta}\| \widehat{\Theta}^L_{S_j,j}-\Omega_{S_j,j}\|_2,
\end{align*}
which gives
\begin{align*}
 \| \widehat{\Theta}^L_{j}-\Omega_{j}\|_2&\leq 3\sqrt{s_j}\lam_{\Theta}/\phi_0.
\end{align*}

(ii) If 
  \[
  \frac{3\lam_{\Theta}}{2}\| \widehat{\Theta}^L_{S_j,j}-\Omega_{S_j,j}\|_1\leq \frac{3}{4} \langle \Omega_{j}-\Theta^*_{j},\widehat{\Sig}^{\mA}(\Omega_{j}-\Theta^*_{j})\rangle ,
\]
then
\[
   \lam_{\Theta}\|\widehat{\Theta}^L_{j}-\Omega_{j}\|_1\leq 3\langle \Omega_{j}-\Theta^*_{j},\widehat{\Sig}^{\mA}(\Omega_{j}-\Theta^*_{j})\rangle
\]
We use Theorem 1 in \citet{Raskutti10} to arrive at for any $ \|u\|_1\leq \langle \Omega_{j}-\Theta^*_{j},\widehat{\Sig}^{\mA}(\Omega_{j}-\Theta^*_{j})\rangle/\lam_{\Theta}$,
\[
   u^{\intercal}\widehat{\Sig}^{\mA}u\geq u^{\intercal}\Sig^{\mA}u/4-\|\Sig^{\mA}\|_2\langle \Omega_{j}-\Theta^*_{j},\widehat{\Sig}^{\mA}(\Omega_{j}-\Theta^*_{j})\rangle^2.
\]
Therefore,
\begin{align*}
\frac{1}{8}\Lambda_{\min}(\Sig^{\mA})\|\widehat{\Theta}^L_{,j}-\Omega_j\|_2^2\leq  \frac{1}{2} \langle \Omega_{j}-\Theta^*_{j},\widehat{\Sig}^{\mA}(\Omega_{j}-\Theta^*_{j})\rangle(1+o(1)).
\end{align*}
To summarize, in event $\mathcal{E}_1$, we have
\begin{align}
\label{eq2-pf2}
\|\widehat{\Theta}^L_{j}-\Omega_j\|_2^2\leq C \langle \Omega_{j}-\Theta^*_{j},\widehat{\Sig}^{\mA}(\Omega_{j}-\Theta^*_{j})\rangle+s\lam_{\Theta}^2.
\end{align}
We now  verify $\P(\mathcal{E}_1)\rightarrow 1$ and bound $\langle \Omega_{j}-\Theta^*_{j},\widehat{\Sig}^{\mA}(\Omega_{j}-\Theta^*_{j})\rangle$.
Notice that $\Theta^*_{j}$ satisfies
\begin{align*}
   &\widehat{\Sig}^{\mA}\Theta^*_{j}-(e_j +r_j(\widehat{\Delta}^{\mA}))=\widehat{\Sig}^{\mA}\Omega_j+\widehat{\Sig}^{\mA}\Omega^{\mA}r_j(\widehat{\Delta}^{\mA}-\Delta^{\mA})-e_j -r_j(\widehat{\Delta}^{\mA})\\
   &=\widehat{\Sig}^{\mA}\Omega_j-e_j -r_j(\Delta^{\mA})+(\widehat{\Sig}^{\mA}\Omega^{\mA}-I_p)r_j(\widehat{\Delta}^{\mA}-\Delta^{\mA})\\
   &=(\widehat{\Sig}^{\mA}-\Sig^{\mA})\Omega_j+(\widehat{\Sig}^{\mA}\Omega^{\mA}-I_p)r_j(\widehat{\Delta}^{\mA}-\Delta^{\mA}).
\end{align*}
It is easy to show that with probability at least $1-\exp(-c_1\log p)$,
\begin{align*}
\|(\widehat{\Sig}^{\mA}-\Sig^{\mA})\Omega_j\|_{\infty,\infty}\leq c_2\|\Sig^{\mA}\|_2^{1/2}\{\Omega_j^{\intercal}\Sig^{\mA}\Omega_j\}^{1/2}\sqrt{\frac{\log p}{n_{\mA}}}\leq c_3\sqrt{\frac{\log p}{n_{\mA}}},
\end{align*}
where the last step is due to Condition \ref{cond1}. Similarly, with probability at least $1-\exp(-c_1\log p)-\exp(-c_2n)$,
\begin{align*}
&\|(\widehat{\Sig}^{\mA}\Omega^{\mA}-I_p)r_j(\widehat{\Delta}^{\mA}-\Delta^{\mA})\|_{\infty,\infty}\leq \|\widehat{\Sig}^{\mA}\Omega^{\mA}-I_p\|_{\infty,\infty}\|r_j(\widehat{\Delta}^{\mA}-\Delta^{\mA})\|_{\infty,1}\\
&\leq h\|\Sig^{\mA}\|_2^{1/2}\|\Omega^{\mA}\|_2^{1/2}\sqrt{\frac{\log p}{n_{\mA}}}.
\end{align*}
By Condition \ref{cond1} and $h\lesssim s\sqrt{\log p/n}\leq c_1$, it suffices to take $\lam_{\Theta}\geq C\sqrt{\log p/n_{\mA}}$. The restricted eigenvalue condition in $\mathcal{E}_1$ holds with probability $1-\exp(-c_1\log p)-\exp(-c_2n)$ provided that $s\log p=o(n_{\mA})$.
 
Finally, we bound the following term
\begin{align*}
\langle \Omega_j-\Theta^*_{j},\widehat{\Sig}^{\mA}(\Omega_j-\Theta^*_{j})\rangle&=\langle \Omega^{\mA}r_j(\widehat{\Delta}^{\mA}-\Delta^{\mA}),\widehat{\Sig}^{\mA}\Omega^{\mA}r_j(\widehat{\Delta}^{\mA}-\Delta^{\mA})\rangle\\
&\leq \max_j\|r_j(\widehat{\Delta}^{\mA}-\Delta^{\mA})\|_2^2\sup_{\|u\|_1\leq h} u^{\intercal}\Omega^{\mA}\widehat{\Sig}^{\mA}\Omega^{\mA}u\\
&\leq c_1h\delta_n
\end{align*}
with probability at least $1-\exp(-c_1\log p)-\exp(-c_2n)$  provided that $h=O(1)$.
We have shown that (\ref{eq2-pf2}) holds with probability going to 1.

Finally, we establish the convergence rate of $\widehat{\Theta}$ based on $\widehat{\Theta}^L$. The arguments are similar to Theorem 5.1 in \citet{BRT09}. For completeness, we include it here. 
As $\widehat{\Theta}^L_j$ is a feasible solution to (\ref{opt-j}), we have 
\begin{align*}
(\widehat{\Theta}_j-\widehat{\Theta}^L_j)^{\intercal}\widehat{\Sig}^{\mA}(\widehat{\Theta}_j-\widehat{\Theta}^L_j)&\leq 
|\langle\widehat{\Theta}_j-\widehat{\Theta}^L_j,\widehat{\Sig}^{\mA}\widehat{\Theta}_j-e_j-r_j(\widehat{\Delta}^{\mA})\rangle|\\
&\quad+|\langle\widehat{\Theta}_j-\widehat{\Theta}^L_j,\widehat{\Sig}^{\mA}\widehat{\Theta}^L_j-e_j-r_j(\widehat{\Delta}^{\mA})\rangle|\\
&\leq 2\|\widehat{\Theta}_j-\widehat{\Theta}^L_j\|_1\lam_{\Theta}.
\end{align*}
The fact that $\|\widehat{\Theta}_j\|_1\leq \|\widehat{\Theta}^L_j\|_1$ implies that
\begin{align*}
\|\widehat{\Theta}_{S^c_j,j}-\widehat{\Theta}^L_{S^c_j,j}\|_1\leq \|\widehat{\Theta}_{S_j,j}-\widehat{\Theta}^L_{S_j,j}\|_1+2\|\widehat{\Theta}^L_{S^c_j,j}\|_1.
\end{align*}
If case (i) discussed above, 
\[
  \|\widehat{\Theta}^L_{S^c_j,j}\|_1=\|\widehat{\Theta}^L_{S^c_j,j}-\Omega_{S^c_j,j}\|_1\leq 6\|\widehat{\Theta}^L_{S_j,j}-\Omega_{S_j,j}\|_1\leq c_1s\lam_{\Theta}.
\]
In case (ii) discussed above,
\begin{align*}
\lam_{\Theta}\|\widehat{\Theta}^L_{S^c_j,j}\|_1\leq c_2(\widehat{\Delta}^{\mA}_j-\Delta^{\mA}_j)^{\intercal}\widehat{\Sig}^{\mA}(\widehat{\Delta}^{\mA}_j-\Delta^{\mA}_j)\leq c_2h\delta_n.
\end{align*}

We can then separately discuss the two cases: $\|\widehat{\Theta}_{S_j,j}-\widehat{\Theta}^L_{S_j,j}\|_1\geq 2\|\widehat{\Theta}^L_{S^c_j,j}\|_1$ and $\|\widehat{\Theta}_{S_j,j}-\widehat{\Theta}^L_{S_j,j}\|_1\leq 2\|\widehat{\Theta}^L_{S^c_j,j}\|_1$. Using previous arguments, one can easily prove desired results. 

\end{proof}

\section{Proof of debiased estimators}
\begin{proof}[Proof of Theorem \ref{thm-db}]
We start with the following decomposition.
\begin{align*}
\widehat{\Omega}_{i,j}^{(db)}-\Omega_{i,j}&=e_j^{\intercal}\widehat{\Omega}_i+e_i^{\intercal}\widehat{\Omega}_j-\widehat{\Omega}_i^{\intercal}\widetilde{\Sig}\widehat{\Omega}_j-e_j^{\intercal}\Omega_i\\
&=e_j^{\intercal}(\widehat{\Omega}_{i}-\Omega_{i})+(e_i^{\intercal}-\widehat{\Omega}_{i}^{\intercal}\widetilde{\Sig})\widehat{\Omega}_j\\
&=e_j^{\intercal}(\widehat{\Omega}_{i}-\Omega_{i})+(e_i^{\intercal}-\Omega_i^{\intercal}\widetilde{\Sig})\widehat{\Omega}_j-(\widehat{\Omega}_{i}-\Omega_i)^{\intercal}\widetilde{\Sig}\widehat{\Omega}_j\\
&=(e_j^{\intercal}-\widehat{\Omega}_j^{\intercal}\widetilde{\Sig})(\widehat{\Omega}_i-\Omega_i)+(e_i^{\intercal}-\Omega_i^{\intercal}\widetilde{\Sig})\widehat{\Omega}_j\\
&=(e_j^{\intercal}-\Omega_j^{\intercal}\widetilde{\Sig})(\widehat{\Omega}_i-\Omega_i)-(\widehat{\Omega}_j-\Omega_j)^{\intercal}\widetilde{\Sig}(\widehat{\Omega}_i-\Omega_i)+(e_i^{\intercal}-\Omega_i^{\intercal}\widetilde{\Sig})\widehat{\Omega}_j\\
&=(e_j^{\intercal}-\Omega_j^{\intercal}\widetilde{\Sig})(\widehat{\Omega}_i-\Omega_i)-(\widehat{\Omega}_j-\Omega_j)^{\intercal}\widetilde{\Sig}(\widehat{\Omega}_i-\Omega_i)+(e_i^{\intercal}-\Omega_i^{\intercal}\widetilde{\Sig})\Omega_j\\
&\quad+(e_i^{\intercal}-\Omega_i^{\intercal}\widetilde{\Sig})(\widehat{\Omega}_j-\Omega_j)
\end{align*}
It holds that
\begin{align*}
\widehat{\Omega}_{i,j}^{(db)}-\Omega_{i,j}
&=\Omega_{i,j}-\Omega_i^{\intercal}\widetilde{\Sig}\Omega_j +rem_{i,j},
\end{align*}
where 
\begin{align*}
rem_{i,j}&=\underbrace{(e_j^{\intercal}-\Omega_j^{\intercal}\widetilde{\Sig})(\widehat{\Omega}_i-\Omega_i)}_{R_{1,i,j}}-\underbrace{(\widehat{\Omega}_j-\Omega_j)^{\intercal}\widetilde{\Sig}(\widehat{\Omega}_i-\Omega_i)}_{R_{2,i,j}}\\
&\quad +\underbrace{(e_i^{\intercal}-\Omega_i^{\intercal}\widetilde{\Sig})(\widehat{\Omega}_j-\Omega_j)}_{R_{3,i,j}}
\end{align*}

Let $\widehat{\Omega}^{v}_j=(\widehat{\Omega}^{(\textup{CL})}_j,\widehat{\Theta}_j)v_j$ for $v_j=(0,1)^{\intercal}$ or $v_j=(1,0)^{\intercal}$. The first term on the RHS of $rem_{i,j}$ can be upper bounded by
\begin{align*}
|R_{1,i,j}|&\leq |(\widehat{\Omega}_j-\widehat{\Omega}^{v}_j)^{\intercal}V\Lambda U^{\intercal}(e_j-\Omega_j^{\intercal}\widetilde{\Sig})|+|(\widehat{\Omega}^{v}_j-\Omega_j)^{\intercal}(e_j-\Omega_j^{\intercal}\widetilde{\Sig})|\\
&= |(\hat{v}_j-v_j)^{\intercal}V\Lambda U^{\intercal}(e_j-\Omega_j^{\intercal}\widetilde{\Sig})|+|(\widehat{\Omega}^{v}_j-\Omega_j)^{\intercal}(e_j-\Omega_j^{\intercal}\widetilde{\Sig})|\\
&\leq \|\Lambda V^{\intercal}(\hat{v}_j-v_j)\|_2\|U^{\intercal}(e_j-\Omega_j^{\intercal}\widetilde{\Sig})\|_2+|(\widehat{\Omega}^{v}_j-\Omega_j)^{\intercal}(e_j-\Omega_j^{\intercal}\widetilde{\Sig})|,
\end{align*}
where $U\Lambda V^{\intercal}$ is the SVD of $(\widehat{\Omega}^{(\textup{CL})}_j,\widehat{\Theta}_j)$ defined in the proof of Theorem \ref{thm2}.
The first term can be bounded using (\ref{eq-v1}) for $v_j=(0,1)^{\intercal}$ and using (\ref{eq-v2}) for $v_j=(1,0)^{\intercal}$. As in the prof of Theorem \ref{thm2}, $\|U^{\intercal}(e_j-\Omega_j^{\intercal}\widetilde{\Sig})\|_2=O_P(n^{-1/2})$. Recall that by our construction, $\widehat{\Omega}^{(\textup{CL})}_j,\widehat{\Omega}_j$ are independent of $\widetilde{\Sig}$,
We have
\begin{align*}
  &|(\widehat{\Omega}^{v}_j-\Omega_j)^{\intercal}(e_j-\Omega_j^{\intercal}\widetilde{\Sig})|\leq 
  \frac{\min_{\omega\in\{\widehat{\Omega}^{(\textup{CL})}_j,\widehat{\Theta}_j\}}\|\Sig^{1/2}(\omega-\Omega_j)\|_2}{\sqrt{n}}.
\end{align*}
To summarize,
 \begin{align*}
|R_{1,i,j}|&=O_P\left(n^{-1}+\frac{\min_{\omega\in\{\widehat{\Omega}^{(\textup{CL})}_j,\widehat{\Theta}_j\}}\|\Sig^{1/2}(\omega-\Omega_j)\|_2}{\sqrt{n}}\right).
\end{align*}
The term $R_{3,i,j}$ can be similarly bounded. For the second term, we use similar arguments to show that
\begin{align*}
|R_{2,i,j}|&\leq 2(\widehat{\Omega}^{v}_j-\Omega_j)^{\intercal}\widetilde{\Sig}(\widehat{\Theta}_{i}^{v}-\Omega_i)+2(\widehat{\Omega}_j-\widehat{\Omega}^{v}_j)^{\intercal}\widetilde{\Sig}(\widehat{\Omega}_i-\widehat{\Theta}^{v}_{i}).
\end{align*}
Taking $v_j\in\{(0,1),(1,0)\}$ and $v_i\in\{(0,1),(1,0)\}$, $\widehat{\Omega}^{v}_j$ still is independent of $\widetilde{\Sig}$ and hence
\[
   (\widehat{\Omega}^{v}_j-\Omega_j)^{\intercal}\widetilde{\Sig}(\widehat{\Theta}_{i}^{v}-\Omega_i)=O_P(\min_{\theta\in\{\widehat{\Omega}^{(\textup{CL})}_j,\widehat{\Theta}_j\}}\|\theta-\Omega_j\|_2\min_{\theta\in\{\widehat{\Omega}^{(\textup{CL})}_i,\widehat{\Theta}_i\}}\|\theta-\Theta_i\|_2).
\]
For the second term, we have
\begin{align*}
&|(\widehat{\Omega}_j-\widehat{\Omega}^{v}_j)^{\intercal}\widetilde{\Sig}(\widehat{\Omega}_i-\widehat{\Omega}^{v}_{i})|\leq  \|\Lambda V^{\intercal}(\hat{v}-v)\|_2^2 \|U^{\intercal}\widetilde{\Sig}U\|_2=  \|\Lambda V^{\intercal}(\hat{v}-v)\|_2^2O_P(1).
\end{align*}
Using (\ref{eq-v1}) and (\ref{eq-v2}) again, we arrive at
\begin{align*}
\P\left(|R_{2,i,j}|\geq c_1\min_{\omega\in\{\widehat{\Omega}^{(\textup{CL})}_j,\widehat{\Theta}_j\}}\|\omega-\Omega_j\|_2\min_{\omega\in\{\widehat{\Omega}^{(\textup{CL})}_i,\widehat{\Theta}_i\}}\|\omega-\Theta_i\|_2+\frac{t}{n}\right)\leq \exp\{-c_3t\}+2\exp\{-c_4\log p\}.
\end{align*}
To summarize,
\begin{align}
|rem_{i,j}|&\leq c_1\min_{\omega\in\{\widehat{\Omega}^{(\textup{CL})}_j,\widehat{\Theta}_j\}}\|\omega-\Omega_j\|_2\min_{\omega\in\{\widehat{\Omega}^{(\textup{CL})}_i,\widehat{\Theta}_i\}}\|\omega-\Theta_i\|_2+\frac{t}{n} \nonumber\\
&\quad~ +c_2\frac{\min_{\omega\in\{\widehat{\Omega}^{(\textup{CL})}_j,\widehat{\Theta}_j\}}\|\Sig^{1/2}(\omega-\Omega_j)\|_2}{\sqrt{n}}+c_3\frac{\min_{\omega\in\{\widehat{\Omega}^{(\textup{CL})}_i,\widehat{\Theta}_i\}}\|\Sig^{1/2}(\omega-\Omega_i)\|_2}{\sqrt{n}}.\label{eq4-pf3}
\end{align}
with probability at least $1-\exp\{-c_4t\}+2\exp\{-c_5\log p\}$.
\begin{align*}
|\widehat{V}_{i,j}-V_{i,j}|&=|\widehat{\Omega}_{j,j}\widehat{\Omega}_{i,i}+\widehat{\Omega}_{i,j}\widehat{\Omega}_{j,i}-V_{i,j}|\\
&\leq \max_{j_1,i_1,j_2,i_2\in\{i,j\}}2|\widehat{\Omega}_{j_1,i_1}-\Omega_{j_1,i_1}|\Theta_{j_2,i_2}+\max_{j_1,i_1,j_2,i_2\in\{i,j\}}|\widehat{\Omega}_{j_1,i_1}-\Omega_{j_1,i_1}||\widehat{\Omega}_{j_2,i_2}-\Omega_{j_2,i_2}|\\
&=O_P(\max_{l\in\{i,j\}}\|\widehat{\Omega}_{l}-\Omega_l\|_2)
\end{align*}

The asymptotic normality of $\widehat{\zeta}_{i,j}=\Omega_{i,j}-\Omega_i^{\intercal}\widetilde{\Sig}\Omega_j$ follows from the Gaussian property and central limit theorem.

\end{proof}

\section{Proof of Theorem \ref{thm-fdr}}
Let
\[
   z_{i,j}=\frac{\widehat{\zeta}_{i,j}-\Omega_{i,j}}{V_{i,j}^{1/2}/\sqrt{n}}=\frac{\Omega_{i,j}-\Omega_i^{\intercal}\widetilde{\Sig}\Omega_j}{V_{i,j}^{1/2}/\sqrt{n}}.
\]
The following proof is largely based on the results in \citet{Liu13}. Specifically, we will first verify that $z_{i,j}$ is asymptotically normal uniformly in $i,j$ and the bias of the $\hat{z}_{i,j}$ is uniformly negligible under the current conditions. It is left to verify that the highly correlated $z_{i,j}$ are not too many under the current conditions, which follows from the arguments in \citet{Liu13}. We will highlight that the condition (12) in \citet{Liu13} can be omitted using the arguments in \citet{JJ19}.

\begin{lemma}[Uniform convergence in distribution]
\label{lem-unif}
Assume that \[
  \max_{i,j}|\widehat{T}_{i,j}|=o((n\log p)^{-1/2})~\text{and}~n\gg (\log p)^3.
     \]
     Then it holds that
\[
\max_{i,j}  |\hat{z}_{i,j}-z_{i,j}|=o_P((\log p)^{-1/2}),
\]
where
\begin{align*}
\max_{i,j}\left|\frac{\P\left(|z_{i,j}|>t\right)}{G(t)}-1\right|\leq C(\log p)^{-3/2}.
\end{align*}
\end{lemma}
\begin{proof}[Proof of Lemma \ref{lem-unif}]

By Lemma 6.1 in \citet{Liu13},
\begin{align*}
\max_{i,j}\left|\frac{\P\left(|z_{i,j}|>t\right)}{G(t)}-1\right|\leq C(\log p)^{-3/2}.
\end{align*}
We need to show that
\begin{equation}
\label{fdr-pf1}
   \max_{i,j}|\hat{z}_{i,j}-z_{i,j}|\leq \max_{i,j} \frac{|rem_{i,j}|}{V_{i,j}^{1/2}/\sqrt{n}}+\max_{i,j}|z_{i,j}||\widehat{V}_{i,j}^{1/2}/V_{i,j}^{1/2}-1|=o_P((\log p)^{-1/2}).
\end{equation}
By (\ref{eq4-pf3}),
\[
   \P\left(\max_{i,j}|rem_{i,j}|\geq c_1\max_{i,j}(\widehat{T}_{i,j}\frac{\log p}{n})^{1/2}+c_2\widehat{T}_{i,j}+c_3\frac{\log p}{n}\right)\leq 3\exp(-c_4\log p).
\]
On the other hand,
\begin{align*}
\max_{i,j}|\widehat{V}_{i,j}^{1/2}/V_{i,j}^{1/2}-1|=\max_{i,j}\frac{|\widehat{V}_{i,j}/V_{i,j}-1|}{\widehat{V}_{i,j}^{1/2}/V_{i,j}^{1/2}+1}.
\end{align*}
Notice that
\begin{align*}
\max_{i,j}|\widehat{V}_{i,j}-V_{i,j}|&=\max_{i,j}|\widehat{\Omega}_{j,j}\widehat{\Omega}_{i,i}+\widehat{\Omega}_{i,j}\widehat{\Omega}_{j,i}-V_{i,j}|\\
&\leq \max_{j_1,l_1,j_2,l_2}2|\widehat{\Omega}_{j_1,l_1}-\Omega_{j_1,l_1}|\Omega_{j_2,l_2}+\max_{j_1,l_1,j_2,l_2}|\widehat{\Omega}_{j_1,l_1}-\Omega_{j_1,l_1}||\widehat{\Omega}_{j_2,l_2}-\Theta_{j_2,l_2}|\\
&=O_P(\|\widehat{\Omega}-\Omega\|_{\infty,2})=O_p(\max_{i,j}|\widehat{T}_{i,j}|^{1/2}).
\end{align*}
As $\max_{i,j}|z_{i,j}|=O_P(\sqrt{\log p})$ in (\ref{fdr-pf1}), under the conditions of Lemma \ref{lem-unif},
\[
     \max_{i,j}|\hat{z}_{i,j}-z_{i,j}|=O_P(\max_{i,j}(\widehat{T}_{i,j}\log p)^{1/2}+\sqrt{n}\widehat{T}_{i,j}+\frac{\log p}{\sqrt{n}})=o_P((\log p)^{-1/2})
\]
for $n\gg \log p^3$.
\end{proof}

\begin{proof}[Proof of Theorem \ref{thm-fdr}]
The proof is largely based on the proof of Theorem 3.1 in \citet{Liu13}.
Specifically, Lemma \ref{lem-unif} and the conditions on $\mathcal{C}_i(\gam)$ and the sample size condition assumed in Theorem \ref{thm-fdr} guarantee all the conditions of Theorem 3.1 in \citet{Liu13}expect (12). (12) is not needed here because the range of $\hat{t}$ considered here is $[0,\sqrt{2\log q-2\log \log q}]$, while the range of $\hat{t}$ considered in \citet{Liu13} is $[0,2\sqrt{\log p}]$. Specifically, as explained in equation (14) in their paper, the condition (12) guarantees that 
\begin{equation}
\label{eq1-pf4}
    q_0G(\hat{t})\rightarrow \infty,
\end{equation}
where $G(t)=2(1-\Phi(t))$ is the tail probability of standard Gaussian.
As $G(t)\geq \frac{1}{t+1/t}\phi(t)$ and $q_0\asymp q$, it is easy to check that (\ref{eq1-pf4}) holds for arbitrary\\
 $\hat{t}\in [0,\sqrt{2\log q-2\log \log q}]$.

\end{proof}

\section{Estimation and inference with CLIME in one-sample case}
To understand the performance of the proposed debiasing procedure in one-sample case, we derive the convergence rate for
\[
\widehat{\Omega}_{i,j}^{(dCL)}=\widehat{\Omega}^{(\textup{CL})}_{j,i}+\widehat{\Omega}^{(\textup{CL})}_{i,j}-(\widehat{\Omega}^{(\textup{CL})}_{j})^{\intercal}\widehat{\Sig}^n\widehat{\Omega}^{(\textup{CL})}_{i},
\]
where $\widehat{\Sig}^n$ denotes the sample covariance matrix based on all the primary data.
  \begin{theorem}[Asymptotic normality for debiased CLIME]
\label{thm-db-clime}
Assume Condition \ref{cond1} and $s\log p\ll n$. For any fixed $i\neq j$,
\[
\widehat{\Omega}_{i,j}^{(dCL)}-\Omega_{i,j}=R_{i,j}+\widehat{T}^{(\textup{CL})}_{i,j},
\]
where
\[
   \frac{\sqrt{n}R_{i,j}}{V_{i,j}^{1/2}}\xrightarrow{D}  N(0,1)~\text{and}~\widehat{T}^{(\textup{CL})}_{i,j}=O_P(\frac{s\log p}{n}).
\]

\end{theorem}

According to our previous discussion, is easy to see that $\widehat{\Omega}_{i,j}^{(dCL)}$ is minimax optimal in $\mathbb{G}(s,\infty)$ for estimationg $\Omega_{i,j}$.

\begin{proof}[Proof of Theorem \ref{thm-db-clime}]
By the proof of Theorem \ref{thm-db},
\begin{align*}
\widehat{\Omega}_{i,j}^{(dCL)}-\Omega_{i,j}
&=\Omega_{i,j}-\Omega_i^{\intercal}\Sig^n\Omega_j +rem^{(\textup{CL})}_{i,j},
\end{align*}
where 
\begin{align*}
rem^{(\textup{CL})}&=\underbrace{(e_j^{\intercal}-\Omega_j^{\intercal}\Sig^n)(\widehat{\Omega}^{(\textup{CL})}_i-\Omega_i)}_{R_{1,i,j}}-\underbrace{(\widehat{\Omega}^{(\textup{CL})}_{j}-\Omega_j)^{\intercal}\Sig^n(\widehat{\Omega}^{(\textup{CL})}_i-\Omega_i)}_{R_{2,i,j}}\\
&\quad +\underbrace{(e_i^{\intercal}-\Omega_i^{\intercal}\Sig^n)(\widehat{\Omega}^{(\textup{CL})}_{j}-\Omega_j)}_{R_{3,i,j}}.
\end{align*}
We apply $\ell_1$-$\ell_{\infty}$ split to $R_{1,i,j}$ and to $R_{3,i,j}$. For $R_{2,i,j}$, we apply the estimation error bound of CLIME proved in Lemma \ref{lem0-thm1}. Standard arguments lead to desired results.

\end{proof}

\section{Proofs of minimax lower and upper bounds}
\label{sec-lq}

\subsection{Minimax lower bound under Frobenius norm}
\begin{proof}[Proof of the lower bound in Theorem \ref{thm-mini-frob2}]

We prove the minimax lower bounds for $q\in(0,1]$ first.
We will separately show that
\begin{align}
\inf_{\widehat{\Omega}}\sup_{\Omega\in\mathbb{G}_q(s,0)}\P\left(\frac{1}{p}\|\widehat{\Omega}-\Omega\|_F^2\geq \frac{s\log p}{n_{\mA_q}+n}\right)>1/4>0.\label{lb1-frob}
\end{align}
and
\begin{align}
\inf_{\widehat{\Omega}}\sup_{\Omega\in\mathbb{G}_q\left(h^q(\frac{\log p}{n})^{-q/2},h\right)}\P\left(\frac{1}{p}\|\widehat{\Omega}-\Omega\|_F^2\geq  h^q(\frac{\log p}{n})^{1-q/2}\wedge h^2\wedge \frac{s\log p}{n}\right)>1/4>0.\label{lb2-frob}
\end{align}

\underline{Proof of (\ref{lb1-frob})}.
\textit{(i) Step 1: Constructing the parameter set.}
Let $r=\lfloor p/2 \rfloor$. Let $B$ be the collection of row vectors $b=(v_j)_{j=1,\dots,p}$ such that $v_j=0$ for $1\leq j\leq p-r$ and $v_j\in\{0,1\}$ for $p-r+1\leq j\leq p$ under the constraint that $\|b\|_0=k$, where the value of $k$ will be specified later. We treat each $(b_1,\dots, b_r)\in B^r$ as an $r\times p$ matrix with the $i$-th row equal to $b_i$. Each $b_mi$ defines a $p\times p$ symmetric matrix $A_i(b_i)$ by making the $i$-th row and $-$th column of $A_i(b_i)$ equal to $b_i$ and the rest of entries 0. Note that $\sum_{i=1}^rA_i(b_i)$ is less than or equal to $k$. Let $\Gamma=\{0,1\}^r$. 
For each $\gam\in\Gamma$ and $b\in B^r$, we associate them with an inverse covariance matrix
\[
 \Omega(\gam,b)=I_p+\eps_{n,p} \sum_{i=1}^r \gam_iA_i(b_i),
\]
where $\eps_{n,p}$ will be decided later. Let $\theta=(\gam,b)\in\Theta$. Define the a collection $\mathcal{G}_*$ of inverse covariance matrices as
\[
   \mathcal{G}_*=\{\Omega(\theta): \Omega(\theta)=I_p+\eps_{n,p} \sum_{i=1}^r \gam_iA_i(b_i), \theta\in\Theta\}.
\]

We now specify $\eps_{n,p}$ and $k$ such that $\mathcal{G}_*\subseteq \mathbb{G}_q(s,0)$.
Let $\eps_{n,p}=\nu\sqrt{\log p/(n+n_{\mA_q})}$, $k=s$ and $\Omega^{(k)}=\Omega$ for any $k\in\mA$. It is easy to check that
\[
   \max_j\|\Omega_j\|_0\leq s~\text{and}~\max_{k\in \mA_q} \mathcal{D}_q(\Omega^{(k)},\Omega)=0.
\]

\textit{Step 2: Applying the general lower bound argument.} We use Lemma 2 in \citet{Cai16} for estimating a function $\psi(\theta)$. For readers' convenience, we present it here. Let $\gam_i(\theta)$ be the $i$-th coordinate of the first component of $\theta$.

\textbf{Lemma 2 in \citet{Cai16}}. For any estimator $T$ of $\psi(\theta)$ based on an observation from the experiment $\{\P_{\theta},\theta\in\Theta\}$, and any $s>0$,
\[
      \max_{\Theta}2^2\E_{\theta}[d^s(T,\psi(\theta))]\geq \alpha \frac{r}{2}\min_{1\leq i\leq r}\|\bar{\P}_{0,i}\wedge \bar{\P}_{1,i}\|,
\]
where 
\[
    \bar{\P}_{0,i}=\frac{1}{2^{r-1}|B^r|}\sum_{\theta}\{\P_{\theta}:\gam_i(\theta)=a\}
\]
and
\[
   \alpha=\min_{(\theta,\theta'): H(\gam(\theta),\gam(\theta'))\geq 1}\frac{d^s(\psi(\theta),\psi(\theta'))}{H(\gam(\theta),\gam(\theta'))}.
\]

Taking $d_s(A,B)$ as $\|A-B\|_F^2$,we have
\begin{equation}
\label{frob-eq1}
      \inf_{\widehat{\Omega}}\max_{\theta\in\Theta}2^2\E_{\theta}\|\widehat{\Omega}-\Omega(\theta)\|_F^2\geq \alpha \frac{p}{4}\min_{1\leq i\leq r}\|\bar{\P}_{0,i}\wedge \bar{\P}_{1,i}\|,\end{equation}
 where
 \[
 \alpha=\min_{(\theta,\theta'): H(\gam(\theta),\gam(\theta'))\geq 1}\frac{\|\Omega(\theta)-\Omega(\theta')\|_F^2}{H(\gam(\theta),\gam(\theta'))}.
 \]
 It is easy to see that
 \[
     \min_{(\theta,\theta'): H(\gam(\theta),\gam(\theta'))\geq 1}\frac{\|\Omega(\theta)-\Omega(\theta')\|_F^2}{H(\gam(\theta),\gam(\theta'))}\asymp \frac{s\log p}{n+n_{\mA_q}}.
 \]
It follows from Lemma 6 of \citet{Cai16} that there is a constant $c$ that
\[
  \min_{1\leq i\leq r}\|\bar{\P}_{0,i}\wedge \bar{\P}_{1,i}\|\geq c.
\]
In view of (\ref{frob-eq1}), the proof of (\ref{lb1-frob}) is complete now.

\underline{(ii) Proof of (\ref{lb2-frob})}. 

\underline{(ii-1)}. When $s\log p/n\leq  h^q(\frac{\log p}{n})^{1-q/2}\wedge h^2$, we consider 
\[
   \Omega\in \mathcal{G}^*,~\text{where}~\eps_{n,p}=\nu\sqrt{\log p/n}, k=s
\]
and $\Omega^{(k)}=I_p$, for $k\in\mA$. It is easy to check that
\[
   \max_j\|\Omega_j\|_0\leq s~\text{and}~\max_{k\in \mA_q} \mathcal{D}_q(\Omega^{(k)},\Omega)\leq s^{1/q}\sqrt{\frac{\log p}{n}}\leq h.
\]
Using the arguments in the proof of (\ref{lb1-frob}), we can show based on (\ref{frob-eq1}) that
 \[
    \alpha\geq \frac{s\log p}{n}~\text{and}~ \min_{1\leq i\leq r}\|\bar{\P}_{0,i}\wedge \bar{\P}_{1,i}\|\geq c.
 \]
\underline{(ii-2)}. 
When $h^q(\frac{\log p}{n})^{1-q/2}\leq s\log p/n  \wedge h^2$, we consider 
\[
   \Omega\in \mathcal{G}^*,~\text{where}~\eps_{n,p}=\nu\sqrt{\log p/n}, k=h^q(\frac{\log p}{n})^{-q/2}
\]
and $\Omega^{(k)}=I_p$, for $k\in\mA$. It is easy to check that
\[
   \max_j\|\Omega_j\|_0\leq h^q(\frac{\log p}{n})^{-q/2}\leq s~\text{and}~\max_{k\in \mA_q} \mathcal{D}_q(\Omega^{(k)},\Omega)\leq h.
\]
Using the arguments in the proof of (\ref{lb1-frob}), we can show based on (\ref{frob-eq1}) that
 \[
    \alpha\geq h^q(\frac{\log p}{n})^{1-q/2}~\text{and}~ \min_{1\leq i\leq r}\|\bar{\P}_{0,i}\wedge \bar{\P}_{1,i}\|\geq c.
 \]
 \underline{(ii-3)}. 
When $h^2\leq s\log p/n  \wedge h^q(\frac{\log p}{n})^{1-q/2}$, it implies that $h\leq \sqrt{\log p/n}$. We consider 
\[
   \Omega\in \mathcal{G}^*,~\text{where}~\eps_{n,p}=h, k=1
\]
and $\Omega^{(k)}=I_p$, for $k\in\mA$. It is easy to check that
\[
   \max_j\|\Omega_j\|_0=1\leq s~\text{and}~\max_{k\in \mA_q} \mathcal{D}_q(\Omega^{(k)},\Omega)\leq h.
\]
Using the arguments in the proof of (\ref{lb1-frob}), we can show based on (\ref{frob-eq1}) that
 \[
    \alpha\geq h^2~\text{and}~ \min_{1\leq i\leq r}\|\bar{\P}_{0,i}\wedge \bar{\P}_{1,i}\|\geq c.
 \]

We now the minimax lower bounds for $q=0$. 
First, it exactly follows from the proof of (\ref{lb1}) that
\begin{align*}
\inf_{\widehat{\Omega}}\sup_{\mathbb{G}_0(s,0)}\P\left(\frac{1}{p}\|\widehat{\Omega}-\Omega\|_F^2\geq c_1(n_{\mA_0}+n)^{-1/2}+c_2\frac{s\log p}{n_{\mA_0}+n}\right)>1/4>0.
\end{align*}
Next, we can show that 
\begin{align*}
\inf_{\widehat{\Omega}}\sup_{\mathbb{G}_0(h\wedge s,h\wedge s)\cap \{\Omega^{(k)}=I_p, k\in\mA_0\}}\P\left(\frac{1}{p}\|\widehat{\Omega}-\Omega\|_F^2\geq \frac{(h\wedge s)\log p}{n}\right)>1/4>0.
\end{align*}
In the parameter space $\mathbb{G}_0(h\wedge s,h\wedge s)\cap \{\Omega^{(k)}=I_p, k\in\mA_0\}$, $\Omega^{(k)}$ provides no information to $\Omega$ and hence it is equivalent to the one sample case with sparsity constraint $h\wedge s$.
\end{proof}

\subsection{Minimax upper bound under Frobenius norm}
\begin{algorithm}[H]
 \SetKwInOut{Input}{Input}
    \SetKwInOut{Output}{Output}
\SetAlgoLined
 \Input{Primary data $X$ and informative auxiliary samples $\{X^{(k)}\}_{k\in\mA_q}$}
 \Output{$\widehat{\Omega}^{\mA_q}$ }
 
\underline{Step 1}. For each $k\in\mA_q$, compute
\begin{align}
&\widehat{\Delta}^{(k)}=\argmin_{\Delta\in\R^{p\times p}} \|\Delta\|_1 \label{eq-Deltak-l0}\\
&\text{subject to} \quad\|\widehat{\Sig}\Delta-(\widehat{\Sig}^{(k)}-\widehat{\Sig})\|_{\infty,\infty}\leq \lam_{k}.\nonumber
\end{align}

Let
\[
   \widehat{\Delta}^{(init)}=\sum_{k\in\mA_q}\alpha_k\widehat{\Delta}^{(k)}.
\]
\begin{align}
\widehat{\Delta}^{\mA_q}&=\argmin_{\Delta\in\R^{p\times p}} \|\Delta\|_1 \label{eq-DeltaA-lq}\\
&\text{subject to} \quad\|\Delta-\widehat{\Delta}^{(init)}-\widehat{\Omega}^{(\textup{CL})}(\widehat{\Sig}^{\mA_q}-\widehat{\Sig}-\widehat{\Sig}\widehat{\Delta}^{(init)})\|_{\infty,\infty}\leq 2\lam_{\Delta}.\nonumber
\end{align}

\underline{Step 2}.
Compute
\begin{align}
&\widehat{\Theta}^{\mA_q}=\argmin_{\Theta\in\R^{p\times p}} \|\Theta\|_1\label{eq-Theta-est-lq}\\
&\text{subject to} \quad\|\widehat{\Sig}^{\mA}\Theta-(\widehat{\Delta}^{\mA_q}+I_p)^{\intercal}\|_{\infty,\infty}\leq \lam_{\Theta}.\nonumber
\end{align}

\underline{Step 3}
Let $\nu_q=h$ for $q=0$ and $\nu_q=h^q(\log p/n)^{-q/2}$ for $q\in (0,1)$.
\[
   \widehat{\Omega}^{\mA_q}=\left\{\begin{array}{ll}
      &\widehat{\Theta}^{\mA_q}~~\text{if}~\nu_q\ll s ~\text{and} ~n_{\mA_q}\gtrsim n\\
   &\widehat{\Omega}^{(\textup{CL})}~~\text{otherwise}.
   \end{array}\right.
\]

\caption{Realization of $\widehat{\Omega}^{\mA_q}$ for $q\in [0,1)$}
\label{alg1-l0}
\end{algorithm}

\begin{proof}[Proof of minimax upper bound in Theorem \ref{thm-mini-frob2}]
It follows from the proof of Lemma \ref{lem1-thm1} that $\Delta^{(k)}$ is a feasible solution to (\ref{eq-Deltak-l0}). 

(i) First consider \underline{$q=0$}.
For $\lam_k\geq c\sqrt{\frac{\log p}{n\wedge n_k}}$, with probability at least $1-\exp(-c_1\log p)-\exp(-c_2n)$
\begin{equation}
\label{eq1-pf5}
   \max_{1\leq j\leq p}\|\widehat{\Delta}_j^{(k)}-\Delta^{(k)}_j\|_1\leq  h\sqrt{\frac{\log p}{n\wedge n_k}}.
\end{equation}
Next, we analyze
\begin{align*}
\widehat{\Delta}^{(db)}-\Delta^{\mA_0}&=\widehat{\Delta}^{(init)}+\widehat{\Omega}^{(\textup{CL})}(\widehat{\Sig}^{\mA_0}-\widehat{\Sig}-\widehat{\Sig}\widehat{\Delta}^{(init)})-\Delta^{\mA_0}\\
&=rem_1+rem_2+rem_3
\end{align*}
as defined in (\ref{eq-rem}).
One can see that
\begin{align*}
\max_j\|\widehat{\Delta}^{(init)}_j-\Delta^{\mA_0}_j\|_1&\leq \max_j\sum_{k\in\mA_0}\alpha_k\|\widehat{\Delta}^{(k)}_j-\Delta^{(k)}_j\|_1\\
&\leq \max_j\sum_{k\in\mA_0}\frac{n_k}{n_{\mA_0}}h\sqrt{\frac{\log p}{n\wedge n_k}}\leq h\sqrt{\frac{\log p}{n}},
\end{align*}
where the last step is due to $n_{\mA_0}\gtrsim |\mA_0|n$. Similar to the analysis of $rem_1$, $rem_2$, and $rem_3$ in Lemma \ref{lem1-thm1}, we have
\[
    \|\widehat{\Delta}^{(db)}-\Delta^{\mA_0}\|_{\infty,\infty}\leq c_1\sqrt{\frac{\log p}{n}}+(h+s)\frac{\log p}{n}
\]
with probability at least $1-\exp(-c_1\log p)-\exp(-c_2n)$. Hence, we take  $\lam_{\Delta}=c_2\sqrt{\log p/n}$. We can show that when $h\leq s$ and $s\sqrt{\log p}\leq c_1\sqrt{n}$\begin{align*}
\|r_j(\widehat{\Delta}^{\mA})-r_j(\Delta^{\mA})\|_2^2\leq  c_3\frac{h\log p}{n}
\end{align*}
with probability at least $1-\exp(-c_1\log p)-\exp(-c_2n)$.
The rest of the proof follows from Lemma \ref{thm1}.

(ii) Consider any fixed \underline{$q\in(0,1)$}. Define $J_k=\{j: |\delta^{(k)}_j|\geq c_1\lam_k\}$. Notice that $|J_k|\leq h^q/\lam_k^{q}$ and $\|\delta^{(k)}_{J_k^c}\|_1\leq \lam_k^{1-q}h^q$. Hence, for $\lam_k\geq c\sqrt{\frac{\log p}{n\wedge n_k}}$, standard arguments give that with probability at least $1-\exp(-c_1\log p)-\exp(-c_2n)$
\[
   \max_{1\leq j\leq p}\|\widehat{\Delta}_j^{(k)}-\Delta^{(k)}_j\|_2^2\leq  (|J_k|\lam_k^2+\|\delta^{(k)}_{J_k^c}\|_1\lam_k)\leq c_1h^q\left(\frac{\log p}{n_0\wedge n_k}\right)^{1-q/2}.
\]
By Lemma 5 in \citet{Raskutti11b}, we have
\[
\max_{1\leq j\leq p}\|\widehat{\Delta}_j^{(k)}-\Delta^{(k)}_j\|_1\leq  h^q\left(\frac{\log p}{n_0\wedge n_k}\right)^{1/2-q/2}.
\]
The rest of the proof follows from above proof for $q=0$.
\end{proof}

\subsection{Minimax lower bounds for estimating $\Omega_{i,j}$}
\begin{proof}[Proof of the lower bounds in Theorem \ref{thm-mini-db2}]
We will first derive the minimax lower bound for $q\in(0,1]$ for $\mathbb{G}_q$( defined in \ref{eq-Gam}).
We will separately show that
\begin{align}
\inf_{\widehat{\Omega}_{i,j}}\sup_{\Omega\in\mathbb{G}_q(s,0)}\P\left(|\widehat{\Omega}_{i,j}-\Omega_{i,j}|\geq c_1(n_{\mA_q}+n)^{-1/2}+c_2\frac{s\log p}{n_{\mA_q}+n}\right)>1/4>0.\label{lb1}
\end{align}
and
\begin{align}
\inf_{\widehat{\Omega}_{i,j}}\sup_{\Omega\in\mathbb{G}_q(s,h)}\P\left(|\widehat{\Omega}_{i,j}-\Omega_{i,j}|\geq n^{-1/2}\wedge h+\frac{s\log p}{n}\wedge h^q(\frac{\log p}{n})^{1-q/2}\wedge h^2\right)>1/4>0.\label{lb2}
\end{align}
The parameter spaces in (\ref{lb1}) and (\ref{lb2}) are subspaces of $\mathbb{G}_q(s,h)$. Hence, the minimax rate in $\mathbb{G}_q(s,h)$ is lower bounded by the maximum of the two bounds, which is of order
\begin{align*}
  & (n_{\mA_q}+n)^{-1/2}+\frac{s\log p}{n_{\mA_q}+n}+n^{-1/2}\wedge h+\frac{s\log p}{n}\wedge h^q(\frac{\log p}{n})^{1-q/2}\wedge h^2\\
   &\asymp R_q+\frac{s\log p}{n_{\mA_q}+n}+\frac{s\log p}{n}\wedge h^q(\frac{\log p}{n})^{1-q/2}\wedge h^2.
\end{align*}
This is exactly the claim of Theorem \ref{thm-mini-db} and the lower bounds of Theorem \ref{thm-mini-db2}. Hence, we only need to prove (\ref{lb1}) and (\ref{lb2}).

(i) Proof of (\ref{lb1}). In the parameter space $\mathbb{G}_q(s,0)$, $\Omega^{(k)}=\Omega$ for any $k\in\mA$. Hence, it is equivalent to the one sample case with $n_{\mA_q}+n$ independent observations.  It follows from Theorem 5 in \citet{Zhao15} that for $3<s\leq c_0\min\{p^\nu,\sqrt{(n_{\mA_q}+n)/\log p}\}$ with $\nu<1/2$ and some constant $c_0$, (\ref{lb1}) holds.

(ii) Proof of (\ref{lb2}). 

(ii-1)
We first show that
\[
  \inf_{\widehat{\Omega}_{i,j}}\sup_{\Omega\in\mathbb{G}_q(s,h)\cap\{\Omega^{(k)}=I_p,k\in\mA_q\}}\P\left( |\widehat{\Omega}_{i,j}-\Omega_{i,j}|\geq c_1n^{-1/2}\wedge h\right)>1/2
\]
when $\frac{s\log p}{n}\wedge  h^q(\frac{\log p}{n})^{1-q/2}\wedge h^2< h\wedge n^{-1/2}$. Notice that $h<1$ implied by the constraint.
Consider
\begin{align*}
&\mathcal{H}_0: \Omega=\Omega^{(k)}=I_p\quad v.s.\\
& \mathcal{H}_1: \Omega=\Omega_1=\begin{pmatrix}
1 & b& \mathbf{0}^{\intercal}_{p-2}\\
b & 1&\mathbf{0}^{\intercal}_{p-2}\\
\mathbf{0}_{p-2}& \mathbf{0}_{p-2} &  I_{p-2} 
\end{pmatrix},~~ \Omega^{(k)}=I_p, ~k\in\mA_q,
\end{align*}
where $b=cn^{-1/2}\wedge h$. Notice that the distributions under $\mathcal{H}_0$ and $\mathcal{H}_1$ are within $\mathbb{G}_q(s,h)$ for $s>3$.
The distribution of $X^{(k)}$, $k\in\mA_q$ are unchanged under $\mathcal{H}_0$ and $\mathcal{H}_1$. Hence, the KL-divergence under $\mathcal{H}_0$ and $\mathcal{H}_1$ is
\begin{align*}
   KL(f_{\mathcal{H}_0},f_{\mathcal{H}_1})&=nKL\left(N(0, I_p), N(0,\Omega_1^{-1})\right)\\
   &=\frac{n}{2}\sum_{i=1}^p\log (1+\lam_j(\Omega_1))-\frac{n}{2}\sum_{j=1}^p\lam_j(\Omega_1),
\end{align*}
where $\lam_j(\Omega_1)$, $j=1,\dots,p$ are the eigenvalues of $\Omega_1$. It is easy to calculate that
\[
  \lam_1(\Omega_1)=1+|b|,\lam_2(\Omega_1)=\dots=\lam_{p-1}(\Omega_1)=1,\lam_p(\Omega_1)=1-|b|.
\]
By Taylor expansion of $\log (1+x)$, it is easy to show that
\begin{align*}
   KL(f_{\mathcal{H}_0},f_{\mathcal{H}_1})&\leq nb^2/2\leq c^2/2<1/8
\end{align*}
for $c<1/2$. We can take $b=cn^{-1/2}\wedge h$ for a small enough constant $c$. Standard arguments lead to
\[
  \inf_{\widehat{\Omega}}\sup_{\Omega\in\mathbb{G}_q(s,h)\cap\{\Omega^{(k)}=I_p,k\in\mA_q\}}\P\left( |\widehat{\Omega}_{1,2}-\Omega_{1,2}|\geq c_1n^{-1/2}\wedge h\right)>1/2.
\]

(ii-2)
We now show that
\[
  \inf_{\widehat{\Omega}_{i,j}}\sup_{\Omega\in\mathbb{G}_q(s,h)\cap\{\Omega^{(k)}=I_p,k\in\mA_q\}}\P\left( |\widehat{\Omega}_{i,j}-\Omega_{i,j}|\geq c\frac{s\log p}{n}\right)>1/2,
\]
when $n^{-1/2}\wedge h\leq \frac{s\log p}{n}\wedge  h^q(\frac{\log p}{n})^{1-q/2}\wedge h^2$ and $\frac{s\log p}{n}\leq  h^q(\frac{\log p}{n})^{1-q/2}\wedge h^2$.
Consider
\begin{align*}
&\mathcal{H}_0: \Sig=\Sig^{(k)}=\begin{pmatrix}
1 & b_0&\mathbf{0}^{\intercal}_{p-2}\\
b_0 & 1 &\mathbf{0}^{\intercal}_{p-2}\\
\mathbf{0}_{p-2}&\mathbf{0}^{\intercal}_{p-2} &  I_{p-2} 
\end{pmatrix}:=\bar{\Sig}_0\quad v.s.\\
& \mathcal{H}_{\delta}: \Sig=\begin{pmatrix}
1 & b_0&\delta^{\intercal}\\
b_0 & 1 &\mathbf{0}^{\intercal}_{p-2}\\
\delta &\mathbf{0}_{p-2} &  I_{p-2} 
\end{pmatrix}:=\bar{\Sig}_{\delta},~~ \Sig^{(k)}=\bar{\Sig}_0 ~k\in\mA_q,
\end{align*}
where $\|\delta\|_0=s-2$ and $\delta_j\in\{0, C_1\sqrt{\log p/n}\}$, $j=1,\dots,p-2$. 
Notice that
\begin{align*}
&\bar{\Omega}_{\delta}=\{\bar{\Sig}_{\delta}\}^{-1}=\begin{pmatrix}
\frac{1}{1-b_0^2-\|\delta\|_2^2} & -\frac{b_0}{1-b_0^2-\|\delta\|_2^2}&-\frac{\delta^{\intercal}}{1-b_0^2-\|\delta\|_2^2}\\
-\frac{b_0}{1-b_0^2-\|\delta\|_2^2}& \frac{1}{1-b_0^2-\|\delta\|_2^2}  &\frac{b_0\delta^{\intercal}}{1-b_0^2-\|\delta\|_2^2}\\
-\frac{\delta}{1-b_0^2-\|\delta\|_2^2}&\frac{b_0\delta}{1-b_0^2-\|\delta\|_2^2} &  I_{p-2} +\frac{\delta\delta^{\intercal}}{1-b_0^2-\|\delta\|_2^2}
\end{pmatrix}
\end{align*}
We can check that $\|\Delta^{(k)}\|_{\infty,1}=0$ under $\mathcal{H}_0$. Under $\mathcal{H}_1$,
\begin{align*}
\Delta^{(k)}&=\bar{\Omega}_{\delta}\Sig^{(k)}-I_p=\bar{\Omega}_{\delta}\begin{pmatrix}
0& 0&\delta^{\intercal}\\
0 & 0&\mathbf{0}^{\intercal}_{p-2}\\
\delta &\mathbf{0}_{p-2} &  0_{p-2} 
\end{pmatrix}\\
&=\begin{pmatrix}
\frac{-\|\delta\|_2^2}{1-b_0^2-\|\delta\|_2^2}& 0&\frac{\delta^{\intercal}}{1-b_0^2-\|\delta\|_2^2} \\
\frac{b_0\|\delta\|_2^2}{1-b_0^2-\|\delta\|_2^2}& 0&-\frac{b_0\delta^{\intercal}}{1-b_0^2-\|\delta\|_2^2}\\
\frac{1-b_0^2}{1-b_0^2-\|\delta\|_2^2}\delta&\mathbf{0}_{p-2}&  -\rho\delta\delta^{\intercal}
\end{pmatrix}.
\end{align*}
Hence, for $\log p/n=o(1)$, 
\begin{align*}
\max_j\|\Delta^{(k)}_j\|_{q}+\max_j\|r_j(\Delta^{(k)})\|_q&\leq c_1\|\delta\|_q+c_2\|\delta\|_2^2\\
&\leq C_1s^{1/q}(\log p/n)^{1/2}+s\log p/n\leq C_2<h.
\end{align*}
We see that
\[
   \{\bar{\Omega}_0\}_{1,2}-\{\bar{\Omega}_{\delta}\}_{1,2}=\frac{-b_0}{1-b_0^2}+\frac{b_0}{1-b_0^2-\|\delta\|_2^2}=\frac{b_0\|\delta\|_2^2}{1-b_0^2-\|\delta\|_2^2}.
\]
When $\|\delta\|_2^2=C_1^2(s-2)\log p/n<c_1<1-b_0^2$, $|\{\bar{\Omega}_0\}_{1,2}-\{\bar{\Omega}_{\delta}\}_{1,2}|\geq c_2\frac{s\log p}{n}$. 
One can bound the total variation distance using Lemma 1 and the proof of Theorem 5 in \citet{Zhao15}.

(ii-3)  We now show that
\[
  \inf_{\widehat{\Omega}_{i,j}}\sup_{\Omega\in\mathbb{G}_q(s,h)\cap\{\Omega^{(k)}=I_p,k\in\mA_q\}}\P\left( |\widehat{\Omega}_{i,j}-\Omega_{i,j}|\geq ch^q(\frac{\log p}{n})^{1-q/2}\right)>1/2,
\]
when $n^{-1/2}\leq \frac{s\log p}{n}\wedge h^q(\frac{\log p}{n})^{1-q/2}\wedge h^2$ and $h^q(\frac{\log p}{n})^{1-q/2}\leq \frac{s\log p}{n}\wedge h^2$. Notice that in this scenario $h\geq \sqrt{\log p/n}$. 

In this case, we take
\begin{align*}
&\mathcal{H}_0: \Sig=\Sig^{(k)}=\begin{pmatrix}
1 & b_0&\mathbf{0}^{\intercal}_{p-2}\\
b_0 & 1 &\mathbf{0}^{\intercal}_{p-2}\\
\mathbf{0}_{p-2}&\mathbf{0}^{\intercal}_{p-2} &  I_{p-2} 
\end{pmatrix}:=\bar{\Sig}_0\quad v.s.\\
& \mathcal{H}_{\delta}: \Sig=\begin{pmatrix}
1 & b_0&\delta^{\intercal}\\
b_0 & 1 &\mathbf{0}^{\intercal}_{p-2}\\
\delta &\mathbf{0}_{p-2} &  I_{p-2} 
\end{pmatrix}:=\bar{\Sig}_{\delta},~~ \Sig^{(k)}=\bar{\Sig}_0 ~k\in\mA,
\end{align*}
where $\|\delta\|_0$ is the integer part of $\{h/(\sqrt{\log p/n})\}^q$ and $\delta_j\in\{0, C_1\sqrt{\log p/n}\}$, $j=1,\dots,p-2$.  Since $h\geq \sqrt{\log p/n}$, $\|\delta\|_0\geq 1$. We can also check that 
\[
   \max_j\|\Delta^{(k)}_j\|_q+\max_j\|r_j(\Delta^{(k)})\|_q \leq c_1\|\delta\|_q+c_2\|\delta\|_2^2\leq c_3h+c_4h^q(\frac{\log p}{n})^{1-q/2}\leq c_5h,
   \] 
   where the last step is due to $h^q(\frac{\log p}{n})^{1-q/2}\leq h$ when $h\gtrsim 1$ and $h^q(\frac{\log p}{n})^{1-q/2}\leq h^2=o(h)$ when $h\ll 1$.
   The rest of the proof follows from the proof of (ii-2).

(ii-4)  We now show that
\[
  \inf_{\widehat{\Omega}_{i,j}}\sup_{\Omega\in\mathbb{G}_q(s,h)\cap\{\Omega^{(k)}=I_p,k\in\mA_q\}}\P\left( |\widehat{\Omega}_{i,j}-\Omega_{i,j}|\geq h^2\right)>1/2,
\]
when $n^{-1/2}\leq \frac{s\log p}{n}\wedge h^q(\frac{\log p}{n})^{1-q/2}\wedge h^2$ and $h^2\leq \frac{s\log p}{n}\wedge h^q(\frac{\log p}{n})^{1-q/2}$. Notice that in this scenario $h< \sqrt{\log p/n}$. 

In this case, we take
\begin{align*}
&\mathcal{H}_0: \Sig=\Sig^{(k)}=\begin{pmatrix}
1 & b_0&\mathbf{0}^{\intercal}_{p-2}\\
b_0 & 1 &\mathbf{0}^{\intercal}_{p-2}\\
\mathbf{0}_{p-2}&\mathbf{0}^{\intercal}_{p-2} &  I_{p-2} 
\end{pmatrix}:=\bar{\Sig}_0\quad v.s.\\
& \mathcal{H}_{\delta}: \Sig=\begin{pmatrix}
1 & b_0&\delta^{\intercal}\\
b_0 & 1 &\mathbf{0}^{\intercal}_{p-2}\\
\delta &\mathbf{0}_{p-2} &  I_{p-2} 
\end{pmatrix}:=\bar{\Sig}_{\delta},~~ \Sig^{(k)}=\bar{\Sig}_0 ~k\in\mA,
\end{align*}
where $\|\delta\|_0=1$ and $\delta_j\in\{0, h\}$, $j=1,\dots,p-2$.  We can also check that 
\[
   \max_j\|\Delta^{(k)}_j\|_q+\max_j\|r_j(\Delta^{(k)})\|_q \leq c_1\|\delta\|_q+c_2\|\delta\|_2^2\leq c_3h+ch^2\leq c_5h
   \] 
    in this case. The rest of the proof follows from the proof of (ii-2).

Finally, we show the results for $q=0$. 
First, it exactly follows from the proof of (\ref{lb1}) that
\begin{align*}
\inf_{\widehat{\Omega}_{i,j}}\sup_{\Omega\in\mathbb{G}_0(s,0)}\P\left(|\widehat{\Omega}_{i,j}-\Omega_{i,j}|\geq c_1(n_{\mA_0}+n)^{-1/2}+c_2\frac{s\log p}{n_{\mA_0}+n}\right)>1/4>0.
\end{align*}
Next, we can show that 
\begin{align*}
\inf_{\widehat{\Omega}_{i,j}}\sup_{\Omega\in\mathbb{G}_0(h\wedge s,h\wedge s)\cap \{\Omega^{(k)}=I_p, k\in\mA_0\}}\P\left(|\widehat{\Omega}_{i,j}-\Omega_{i,j}|\geq n^{-1/2}+\frac{(h\wedge s)\log p}{n}\right)>1/4>0.
\end{align*}
In the parameter space $\Omega\in\mathbb{G}_0(h\wedge s,h\wedge s)\cap \{\Omega^{(k)}=I_p, k\in\mA_0\}$, $\Omega^{(k)}$ provides no information to $\Omega_{1,2}$ and hence it is equivalent to the one sample case with sparsity constraint $h\wedge s$.

The proof is complete by combining above two results.
\end{proof}

\subsection{Minimax upper bounds for estimating $\Omega_{i,j}$}
\begin{proof}[Proof of minimax upper bounds in Theorem \ref{thm-mini-db2}]
First consider \underline{$q=0$}.
For $\widehat{\Omega}^{\mA_0}$ defined in Algorithm \ref{alg1-l0},
\[
   \widehat{\Omega}^{(db,0)}=\widehat{\Omega}^{\mA_0}_{j,i}+\widehat{\Omega}^{\mA_0}_{i,j}-(\widehat{\Omega}^{\mA_0}_{j})^{\intercal}\widetilde{\Sig}\widehat{\Omega}^{\mA_0}_{i}.
\]
The proof follows from the proof of Theorem \ref{thm-db} with setting $\hat{v}_j=v_j$ and the upper bound proof of Theorem \ref{thm-mini-frob2}. 

Next consider $q\in (0,1]$.
For $\widehat{\Omega}^{\mA_q}$ defined in Algorithm \ref{alg1-l0},
\[
   \widehat{\Omega}^{(db,q)}=\left\{\begin{array}{ll}
  & \widehat{\Omega}^{\mA_q}_{j,i}+\widehat{\Omega}^{\mA_q}_{i,j}-(\widehat{\Omega}^{\mA_q}_{j})^{\intercal}\widetilde{\Sig}\widehat{\Omega}^{\mA_q}_{i}~\text{if}~h\gtrsim n^{-1/2}\\
   & \widehat{\Omega}^{\mA_q}_{j,i}+\widehat{\Omega}^{\mA_q}_{i,j}-(\widehat{\Omega}^{\mA_q}_{j})^{\intercal}\widetilde{\Sig}^{\mA_q}\widehat{\Omega}^{\mA_q}_{i}~\text{if}~h\ll n^{-1/2},
   \end{array}\right.
\]
where 
\[
    \widetilde{\Sig}^{\mA_q}=\sum_{k\in \mA_q}\frac{n_k}{n_{\mA_q}+n}\widetilde{\Sig}^{(k)}+\frac{n}{n_{\mA_q}+n}\widetilde{\Sig}
    \]
    for some $\widetilde{\Sig}^{(k)}$ independent of $\widehat{\Sig}^{(k)}$ for all $k\in\mA_q$. Again, this can be achieved by sample splitting of $X^{(k)}$, $k\in\mA_q$.

When $h\gtrsim n^{-1/2}$, it follows from the proof of Theorem \ref{thm-db} with setting $\hat{v}_j=v_j$ and upper bound proof of Theorem \ref{thm-mini-frob2} that
\begin{align}
|\widecheck{\Omega}^{(db,q)}_{i,j}-\Omega_{i,j}|=O_P\left(n^{-1/2}+\frac{s\log p}{n_{\mA_q}+n}+ h^q\delta_n^{2-q}\wedge \frac{s\log p}{n}\right).\label{re1-db-lq}
\end{align}

When $h\ll n^{-1/2}$, tt holds that
\begin{align*}
\widehat{\Omega}_{i,j}^{(db,q)}-\Omega_{i,j}
&=\Omega_{i,j}-\Omega_i^{\intercal}\widetilde{\Sig}^{\mA_q}\Omega_j +rem_{i,j},
\end{align*}
where 
\begin{align*}
rem_{i,j}&=\underbrace{(e_j^{\intercal}-\Omega_j^{\intercal}\widetilde{\Sig}^{\mA_q})(\widehat{\Omega}^{\mA_q}_i-\Omega_i)}_{R_{1,i,j}}-\underbrace{(\widehat{\Omega}^{\mA_q}_j-\Omega_j)^{\intercal}\widetilde{\Sig}^{\mA_q}(\widehat{\Omega}^{\mA_q}_i-\Omega_i)}_{R_{2,i,j}}\\
&\quad +\underbrace{(e_i^{\intercal}-\Omega_i^{\intercal}\widetilde{\Sig}^{\mA_q})(\widehat{\Omega}^{\mA_q}_j-\Omega_j)}_{R_{3,i,j}}.
\end{align*}
We first bound $R_{1,i,j}$ and $R_{3,i,j}$.
\begin{align*}
|R_{1,i,j}|&\leq |\Omega_j^{\intercal}(\Sig-\widetilde{\Sig}^{\mA_q})(\widehat{\Omega}^{\mA_q}_i-\Omega_i)|\\
&\leq |\Omega_j^{\intercal}(\Sig^{\mA_q}-\widetilde{\Sig}^{\mA_q})(\widehat{\Omega}^{\mA_q}_i-\Omega_i)|+|\Omega_j^{\intercal}(\Sig^{\mA_q}-\Sig)(\widehat{\Omega}^{\mA_q}_i-\Omega_i)|\\
&=O_P(\frac{\|\widehat{\Omega}^{\mA_q}_i-\Omega_i\|_2}{\sqrt{n_{\mA_q}+n}}) + \|\Omega_j^{\intercal}(\Sig^{\mA_q}-\Sig)\|_2\|\widehat{\Omega}^{\mA_q}_i-\Omega_i\|_2\\
&=o_P((n_{\mA_q}+n)^{-1/2})+ \max_{k\in\mA_q}\|\Delta^{(k)}_i\|_2o_P(1),
\end{align*}
where the second last line is due to the independence of $\widetilde{\Sig}^{\mA_q}$ and $\widehat{\Omega}^{\mA_q}$ and the last line is due to the sub-additivity of $\ell_2$-norm and the definition of $\Delta^{(k)}$. As $\max_{k\in\mA_q}\|\Delta^{(k)}_i\|_2\leq \max_{k\in\mA_q}\|\Delta^{(k)}_i\|_q\leq h$, we arrive at
\begin{align*}
|R_{1,i,j}|+|R_{3,i,j}|=o_P(n_{\mA_q}+n)^{-1/2}+h).
\end{align*}
For $R_{2,i,j}$,
\begin{align*}
|R_{2,i,j}|&\leq |(\widehat{\Omega}^{\mA_q}_j-\Omega_j)^{\intercal}\Sig^{\mA_q}(\widehat{\Omega}^{\mA_q}_i-\Omega_i)|(1+O_P(n_{\mA_q}+n)^{-1/2})\\
&=O_P\left(\max_i\|\widehat{\Omega}^{\mA_q}_i-\Omega_i\|_2^2\right).
\end{align*}
Hence, we arrive at when $h\ll n^{-1/2}$,
\begin{align}
\label{re2-db-lq}
|rem_{i,j}|=o_P(n_{\mA_q}+n)^{-1/2}+h)+O_P(\frac{s\log p}{n_{\mA_q}+n}+ h^q\delta_n^{2-q}\wedge \frac{s\log p}{n}).
\end{align}
Finally, we analyze
\begin{align}
\Omega_{i,j}-\Omega_i^{\intercal}\widetilde{\Sig}^{\mA_q}\Omega_j&=\Omega_{i}^\intercal\Sig^{\mA_q}\Omega_j-\Omega_i^{\intercal}\widetilde{\Sig}^{\mA_q}\Omega_j+\Omega_{i}^T(\Sig-\Sig^{\mA_q})\Omega_j\nonumber\\
&=O_P(n_{\mA_q}+n)^{-1/2})+\|\Omega_j\|_2\max_{k\in\mA_q}\|\Delta^{(k)}_i\|_q\nonumber\\
&=O_P(n_{\mA_q}+n)^{-1/2})+h.\label{re3-db-lq}
\end{align}
Combining (\ref{re2-db-lq}) and (\ref{re3-db-lq}), we have when $h\ll n^{-1/2}$,
\begin{align*}
|\widehat{\Omega}_{i,j}^{(db,q)}-\Omega_{i,j}|=O_P((n_{\mA_q}+n)^{-1/2}+h)+O_P(\frac{s\log p}{n_{\mA_q}+n}+ h^q\delta_n^{2-q}\wedge \frac{s\log p}{n}).
\end{align*}
Together with (\ref{re1-db-lq}), we arrive at desired results.
\end{proof}

\section{Data applications}
\subsection{Hubs detected by two methods}
\label{sec-hub}
\begin{table}[!htbp]
\begin{tabular}{|c|l|}
\hline
 tissue &Top 5 hubs\\
\hline
A.C. cortex& \texttt{BTG2; ROBO1; SEMA3A; MTPN; SHANK3; OLIG2}\\
\hline
C.B. ganglia &\texttt{ZSWIM6; PTEN; FAIM2; ARHGAP35; SHANK3} \\
\hline 
C. hemisphere & \texttt{LHX4; NTRK2} \\
\hline
Cerebellum & \texttt{LHX4; ZSWIM6; SPOCK1; SUFU; CBLN1} \\
\hline
Cortex & \texttt{ZEB2; CDH11} \\
\hline
F. cortex &\texttt{ROBO1; HES1; SEMA3A; CLN8; CSNK1D} \\
\hline
Hippocampus & \texttt{SATB2; ERBB4; SEMA3A; NIN; CRKL} \\
\hline
Hypothalamus & \texttt{ROBO2; SOX4; SEMA3A; SOX1} \\
\hline
N.A.B. ganglia & \texttt{CHD5; ZSWIM6; SEMA3A; DRD2; SHANK3} \\
\hline
P.B. ganglia & \texttt{GLI2; SALL1; SHANK3; NPY} \\
\hline
\end{tabular}
\caption{The list of whose degree is among the top 5 largest in each target tissue detected by Trans-CLIME at FDR level $\alpha=0.1$.}
\label{table-hub}
\end{table}

\begin{table}[!htbp]
\begin{tabular}{|c|l|}
\hline
 tissue &Top 5 hubs\\
\hline
A.C. cortex& \texttt{HES5; PROX1; SEMA3A; PTEN; DYNC2H1; ATF5}\\
\hline
C.B. ganglia &\texttt{HES5; NTRK2; CSNK1D; ARHGAP35} \\
\hline 
C. hemisphere & \texttt{LHX4; UNC5D} \\
\hline
Cerebellum & \texttt{MTPN; CBLN1} \\
\hline
Cortex & \texttt{EPHB2; UQCRQ; PSEN1; NFE2L1; ARHGAP35; CRKL; SHANK3; NPY} \\
\hline
F. cortex &\texttt{SZT2; DRD1; CSNK1D; CRKL; GABRB1} \\
\hline
Hippocampus & \texttt{LDB1; CRK; CSNK1E; DCC} \\
\hline
Hypothalamus & \texttt{LMO4; PTEN; KNDC1; DRD2} \\
\hline
N.A.B. ganglia & \texttt{IFT172; DYNC2H1; SHH} \\
\hline
P.B. ganglia & \texttt{HES5; ROBO1; EPHB3; RORA; OLIG2} \\
\hline
\end{tabular}
\caption{The list of nodes whose degree is among the top 5 largest in each target tissue detected by CLIME at FDR level $\alpha=0.1$.}
\label{table-hub-cl}
\end{table}
\subsection{Sample sizes}
The list of 13 tissues in consideration and their sample sizes (Table \ref{table-tis}). The first 10 tissues are considered as the target tissues individually (in the order of the $x$-axis in Figure \ref{fig1-data}). The last three tissues are only employed as auxiliary tissues as their sample sizes are relatively small.
\begin{table}[!htbp]
\begin{tabular}{|c|c|c|}
\hline
no.& tissue & sample size\\
\hline
1&A.C. cortex& 109\\
\hline
2 & Brain caudate basal ganglia &144\\
\hline 
3 & Brain cerebellar hemisphere & 125\\
\hline
4 & Brain cerebellum & 154\\
\hline
5 & Brain cortex & 136\\
\hline
6& Brain frontal cortex & 118\\
\hline
7 & Brain hippocampus & 111\\
\hline
8 & Brain hypothalamus & 108\\
\hline
9& Brain nucleus accumbens basal ganglia & 130\\
\hline
10 & Brain putamen basal ganglia & 111\\
\hline
11& Brain amygdala & 88\\
\hline
12 & Brain spinal cord cervical& 83\\
\hline
13 & Brain substantia nigra & 80\\
\hline
\end{tabular}
\caption{The list of 13 tissues in consideration and their sample sizes.}
\label{table-tis}
\end{table}

\end{document}